\documentclass{amsart}

\usepackage{booktabs}
\usepackage{amsbsy}
\usepackage{amsaddr}
\usepackage{mathtools}
\usepackage{color,soul}

\usepackage{amsmath,amsthm}
\allowdisplaybreaks
\usepackage{mathrsfs}
\usepackage{amscd,amssymb, amsfonts, verbatim,subfigure, enumerate}
\usepackage[mathcal]{eucal}
\usepackage[super]{nth}

\usepackage[pdftex]{graphicx}

\usepackage{mathpazo}
\usepackage{color,slashed}

\renewcommand{\sl}{\mathfrak{sl}}

\newcommand{\Res}{\op{Res}}

\newcommand{\g}{\mathfrak{g}}

\renewcommand{\det}{\operatorname{det}}

\newcommand{\br}{\overline}

\newcommand{\C}{\mathbb C}

\newcommand{\Z}{\mathbb Z}

\newcommand{\op}{\operatorname}

\newcommand{\mbb}{\mathbb}

\newcommand{\R}{\mbb R}

\newcommand{\dbar}{\br{\partial}}

\DeclareMathOperator{\Tr}{Tr}

\newtheoremstyle{thm}
  {7pt}
  {7pt}
  {\itshape}
  {}
  {\bf}
  {.}
  {5pt}
  {\thmnumber{#2 }\thmname{#1}\thmnote{ (#3)}}

\newtheoremstyle{def}
  {7pt}
  {10pt}
  {\itshape}
  {}
  {\bf}
  {.}
  {5pt}
  {\thmnumber{#2} \thmname{#1}\thmnote{ (#3)}}

\newtheoremstyle{rem}
  {4pt}
  {10pt}
  {}
  {}
  {\itshape}
  {:}
  {3pt}
  {}

\newtheoremstyle{texttheorem}
  {8pt}
  {8pt}
  {\itshape}
  {}
  {\bf}
  {. \hspace{5pt}}
  {3pt}
  {}

\usepackage{tikz-cd}

\newtheorem*{theorem*}{Theorem}
\newtheorem*{lemma*}{Lemma}
\newtheorem*{corollary*}{Corollary}
\newtheorem*{proposition*}{Proposition}
\newtheorem*{definition*}{Definition}

\newtheorem*{conjecture}{Conjecture}

\newtheorem{theorem}{Theorem}[subsection]
\newtheorem{thm-def}{Theorem/Definition}[theorem]
\newtheorem{proposition}[theorem]{Proposition}

\newtheorem{lemma}[theorem]{Lemma}

\numberwithin{equation}{subsection}

\theoremstyle{def}

\theoremstyle{rem}
\newtheorem*{remark}{Remark}

\usepackage{hyperref}

\newcommand{\PT}{\mbb{PT}}
\newcommand{\CP}{\mbb{CP}}

\renewcommand{\i}{\mathrm{i}}

\title{Higher genus twistor spaces and the celestial torus}
\author{Seraphim Jarov}

\address{Perimeter Institute for Theoretical Physics}
\email{sjarov@perimeterinstitute.ca}

\begin{document}

\begin{abstract}
    This paper studies novel four-dimensional integrable field theories that are deformations of self-dual Yang-Mills. They are engineered by considering holomorphic Chern-Simons and BF type theories on covers of twistor space obtained by pulling back the vector bundle $\mathcal{O}(1)^2\to\CP^1$ to hyperelliptic or elliptic curves. Compactifying to 4d yields an integrable theory, which in the examples I study, are determined to leading order. The form of the higher-order corrections are bootstrapped, and I argue that the index structure and coefficients of these terms are fixed by integrability. The celestial chiral algebras of these theories are shown to live on hyperelliptic and elliptic curves, respectively. Symmetry reducing these integrable deformations to 2d yields an example of a hyperelliptic and elliptic integrable model governing a deformation of Hitchin’s equations.
\end{abstract}

\maketitle

\tableofcontents

\section{Introduction}

Twistor theory provides a beautiful framework for constructing integrable theories in dimension 4 via the Penrose transform \cite{Penrose:1977in,Penrose:1976jq}. Most notable for this work is the correspondence between holomorphic BF theory on twistor space and self-dual Yang-Mills (SDYM) theory on $\R^4$ \cite{Witten:2003nn,Mason:2005zm,Boels:2006ir}. The significance of this example lies not only in its rich structure as a 4d integrable model (see eg. \cite{Mason:1991rf}), but also in the recent interest from the celestial holography community. In \cite{Costello:2022wso}, a framework was studied for computing form factors in self-dual gauge theory from chiral correlators living on the celestial sphere. This idea was promoted to top down constructions of celestial holography in \cite{Costello:2022jpg} and \cite{Bittleston:2024efo} by leveraging the twisted holography program \cite{Costello:2018zrm} on twistor spaces.

This paper is concerned with the following natural extensions of the above works:
\begin{itemize}
    \item \textit{Are there 4d theories whose celestial chiral algebra lives on the torus or other higher genus surfaces?}
    \item \textit{Are there deformations of SDYM theory that do not spoil integrability?}
\end{itemize}

I will show that the answer to both of the above questions is yes by building explicit examples. The construction follows a suggestion made by Costello in \cite{Costello:2021bah} to study twistor actions descending from theories defined on a new geometry - a double cover of twistor space. This geometry is constructed by starting with some higher genus curve $\Sigma\xrightarrow[]{\pi}\CP^1$ and then considering the pullback bundle
\begin{align*}
\PT_\Sigma=(\pi^*\mathcal{O}(1)\oplus\pi^*\mathcal{O}(1)\to\Sigma).    
\end{align*}
This looks like a twistor space where the twistor sphere has been replaced by $\Sigma$. This suggests that theories found by compactifying along the $\Sigma$ direction will have celestial chiral algebras living on $\Sigma$. Also, in analogy with the 4d Chern-Simons story \cite{Costello:2013zra,Costello:2019tri}, this setup provides a means of constructing 4d integrable models whose spectral parameter lives on $\Sigma$. In this paper, I provide a prescription for finding twistor actions that descend from holomorphic theories on $\PT_\Sigma$ and demonstrate that they exhibit the desirable properties mentioned above. 

I study two interesting cases in this paper. The first is holomorphic Chern-Simons theory on $\PT_\Sigma$ when $\Sigma$ is a genus 3 curve. The nice property of this example is that $\PT_\Sigma$ is Calabi-Yau, which makes for a natural setting to study the B-model as done in twisted holography. This is shown to descend to a twistorial theory which upon compactifying to 4d, reduces to a remarkably simple integrable deformation of SDYM whose algebra of collinear singularities is exactly the level 0 Kac-Moody algebra studied in the celestial holography literature\footnote{Sometimes referred to as the $\mathcal{S}$ algebra \cite{Strominger:2021mtt,Guevara:2021abz,Himwich:2021dau,Adamo:2021lrv}.} on the genus 3 curve $\Sigma$. While I was not able to explicitly write down the compactified theory to all orders, I derived the leading order lagrangian and bootstrapped the form of all high-order corrections. The coefficients and index structure of these corrections are fixed by integrability.

The second example is a holomorphic BF theory on $\PT_\Sigma$ when $\Sigma$ is a genus one curve (a torus). This is shown to descend to a deformation of self-dual gauge theory coupled to a complex scalar. Similarly to the genus 3 example, I determine the leading order Lagrangian and bootstrap the form of higher-order corrections. By studying the collinear limits of the theory, I also find a deformation of the level 0 Kac-Moody algebra studied in the celestial holography literature on the torus $\Sigma$. One could thus regard this as a construction of the dual theory of the so-called celestial torus. However, this approach remains signature agnostic and hence does not appear to have relation to the celestial torus studied in \cite{Atanasov:2021oyu}.

In both cases, I symmetry reduce the 4d integrable deformations to find deformations of Hitchin's equations \cite{Hitchin:1986vp}. These deform the Hitchin's equations into a hyperelliptic/elliptic model describing the moduli of $G$-bundles on the given curve $\Sigma$. It seems that the 4d deformations of BF theory that give rise to these Hitchin-type models could be related to the 3d BF description of Hitchin's equations presented in \cite{Vicedo_2022}.

\subsection{Main results and future work}

The main findings of this work are outlined below:
\begin{itemize}
    \item Section \ref{sec:model_build} explains a new method to compute twistor actions for 4d integrable models whose spectral parameters take values in a higher genus Riemann surface $\Sigma$. Two specific examples are worked through: one with $\Sigma$ having genus 3 and one with $ \Sigma$ having genus 1. Using symmetry to write down leading-order Penrose transforms for these examples, I find that they produce two different novel integrable deformations of SDYM theory.
    \item Section \ref{sec:celestial} computes the celestial chiral algebras for the deformations of SDYM. I show that these deformations modify the $\mathcal{S}$ algebra into a level 0 Kac-Moody algebra living on the Riemann surfaces used in the constructions of the theories in both cases. In the case where $\Sigma$ has genus 1, this deformation gives a dual theory for the so-called celestial torus.
    \item Sections \ref{sec:bootstrap} and \ref{sec:genus1_bootstrap} show that for both deformations of SDYM I consider, they exhibit basic properties of an integrable model with the vanishing of certain amplitudes. These amplitudes are computed using standard Feynman diagrams and generalized using the chiral algebra methods of \cite{Costello:2022wso}. I also bootstrap the form of all higher-order corrections in the 4d theories descending from the twistor actions presented in section \ref{sec:model_build}. The index structure and coefficients of these corrections are shown to be fixed by integrability.
    \item Sections \ref{sec:sym_red} and \ref{sec:genus1_symred} compute the symmetry reduction of both 4d deformations and recovers a deformed Hitchin's system \cite{Hitchin:1986vp} which described the moduli of $G$-bundles on $\Sigma$. In the genus 3 case, I also relate this construction to a 4d Chern-Simons theory on $\Sigma\times\mathbb{R}^2$ and make a conjectural correspondence between holomorphic Chern-Simons theory on $\PT_\Sigma$, 4d Chern-Simons theory, and higher genus integrability in dimensions 2 and 4.
\end{itemize}

\subsection*{Future directions}

In sections \ref{sec:bootstrap} and \ref{sec:genus1_bootstrap}, I bootstrapped the form of all vertices in a new example of a 4d hyperelliptic and elliptic integrable model. I also argue that one can recursively fix the index contractions and coefficients of the higher-order corrections using integrability. This amounts to computing amplitudes to all orders and tuning the correction terms so that all tree-level amplitudes vanish. It would be interesting to compute some of these amplitudes to see if there is a pattern governing how the coefficients and indices behave at higher order. Perhaps something analogous to the recursion relations found in \cite{Gabai:2018tmm} can be found for the theories in this paper. The authors of \cite{Ren:2022sws} also have an interesting perspective on constraining bulk theories that admit a celestial dual. Alternatively, performing an explicit Penrose reduction of the twistor actions by gauge fixing would be the most satisfying derivation of the 4d model, but it seems to be quite difficult.

Finding the full 4d integrable models would lead to a better understanding of the 2d deformations of Hitchin's equations found via symmetry reduction in sections \ref{sec:sym_red} and \ref{sec:genus1_symred}. These deformations are profound as they suggest that Hitchin's equations can be deformed from describing the moduli of $G$-bundles on the spacetime $\R^2$ to the curve $\Sigma$ just by adding simple vertices that depend on a constant background tensor. In the genus 3 example, there is a 4d Chern-Simons description of the 2d hyperelliptic model. It would be interesting if one could use 4d Chern-Simons methods to get an explicit description of the 2d model. I am also interested in how the 4d integrable models I find can be related to the 3d BF description of Hitchin's system presented in \cite{Vicedo_2022}.

One of the most interesting features of SDYM is that it can be deformed to full Yang-Mills by adding a $\Tr(B^2)$ term to the action. This feature was discussed in detail in \cite{Costello:2022wso} as they computed some Yang-Mills amplitudes using the celestial chiral algebras to compute SDYM form factors in the presence of the $\Tr(B^2)$ operator. It is thus natural to ask what happens to the integrable deformations of SDYM discussed in this paper if we add the $\Tr(B^2)$ operator. A study on this would uncover any properties that the hyperelliptic/elliptic deformations might pass on to full Yang-Mills theory. The first step in this direction would be to find the new amplitudes this deformation adds to Yang-Mills. In the genus 3 deformation, the simplest is the introduction of a 4-minus amplitude that can be computed using the chiral algebra. I found this to be proportional to the $--++$ Parke-Taylor formula with an extra factor of the polynomial defining the genus 3 curve. It would be interesting to see what this deformation does to more complicated amplitudes.

\subsection*{Acknowledgements}

I would like to thank my PhD advisor, Kevin Costello, for suggesting this project and guiding me along the way. I also sincerely thank Roland Bittleston for our many discussions about this project and for teaching me most of what I know about twistor theory. Research at Perimeter Institute is supported in part by the Government of Canada through the Department of Innovation, Science and Economic Development and by the Province of Ontario through the Ministry of Colleges and Universities.

\section{Background}\label{sec:bgrnd}

\subsection{Review of the twistor space of $\mathbb{R}^4$}\label{subsec:twistor_review}

With this subsection, I review twistor basics using inhomogeneous coordinates to fix some notation for the rest of the paper. For a pedagogical review of twistors in homogeneous coordinates, see the lecture notes \cite{Adamo:2017qyl}. For examples of other works that provide some background on and use twistors in inhomogeneous coordinates, see \cite{Costello:2022wso,Costello:2023hmi}.

The twistor space of $\mathbb{R}^4$, denoted by $\mathbb{PT}$, is the total space of the holomorphic vector bundle 
\begin{align*}
    \mathcal{O}(1)\oplus\mathcal{O}(1)\to\mathbb{CP}^1.
\end{align*}
Intuitively, we can picture this bundle via the real diffeomorphism $\mathbb{CP}^1\times\mathbb{R}^4\cong\mathbb{PT}$. The $\CP^1$ coordinate I will denote by the inhomogeneous coordinate $z$ and on the $\mathcal{O}(1)$ fibres, I call the coordinates $v^{\dot{1}}$ and $v^{\dot{2}}$ which I choose to each have a first order pole at $z=\infty$ (they are sections of $\mathcal{O}(1)$). To get an explicit form of the coordinates $v^{\dot{1}}$ and $v^{\dot{2}}$, let $x^\mu$ be the standard coordinates on $\R^4$ with the Euclidean metric and define
\begin{align}\label{eq:R4_coord}
    x^{\alpha\dot{\alpha}} = \frac{1}{\sqrt{2}}\begin{pmatrix}
        x^0+\i x^3 & x^2+\i x^1\\
        -x^2+\i x^1 & x^0-\i x^3
    \end{pmatrix}
\end{align}
where $\alpha,\dot{\alpha} \in\{1,2\},\{\dot{1},\dot{2}\}$ are $SU(2)$ spinor indices with opposite chirality. It is then convenient to use the following complex coordinates
\begin{align*}
    u^{\dot{\alpha}} = x^{1\dot{\alpha}}\quad\hat{u}^{\dot{\alpha}} = x^{2\dot{\alpha}}.
\end{align*}
on $\R^4$. Using this, the fibre coordinates can be defined as $v^{\dot{\alpha}} = u^{\dot{\alpha}}+z\hat{u}^{\dot{\alpha}}$. These coordinates enjoy the transformation law
\begin{align*}
    z\mapsto1/z\implies v^{\dot{\alpha}}\mapsto v^{\dot{\alpha}}/z
\end{align*}
with the conjugation rules $\hat{v}^{\dot{\alpha}} = -(-\bar{v}^{\dot{2}},\bar{v}^{\dot{1}})/\bar{z}$ and $\hat{z} = -1/\bar{z}$. These new coordinates provide a bijection between the points in $\R^4$ and complex lines $\CP^1\subset \PT$ in twistor space. Moreover, the conformal structure of the spacetime is encoded in the complex structure of $\PT$ \cite{Penrose:1976jq}.

The last object we need is the choice of a volume form. It is well known that $\PT$ is not Calabi-Yau as the canonical bundle is $\mathcal{O}(-4)$. While this is not an issue to define theories such as holomorphic BF theory, holomorphic Chern-Simons theory is not immediately compatible with this geometry. To get around this issue, Costello \cite{Costello:2021bah} chooses the meromorphic volume form expressed in the coordinates that were just introduced
\begin{align}\label{eq:PT_Omega}
    \Omega = z^{-2}d zd v^{\dot{1}}d v^{\dot{2}}
\end{align}
which has second order poles at $z=0,\infty$ and then imposes boundary conditions on the gauge field so that the action is well defined.

\subsection{Examples of holomorphic theories on twistor space}

A remarkable feature of the twistor space just introduced, due to the work of Penrose and Ward \cite{Penrose:1977in,Ward:1977ta}, is that certain holomorphic field theories on $\PT$ are equivalent to massless theories on $\R^4$ via the Penrose transform. Examples of theories and their Penrose transform are presented in this subsection.

\subsection*{Holomorphic Chern-Simons theory}

In \cite{Costello:talk}, it was conjectured by Costello and later verified by Bittleston and Skinner \cite{Bittleston:2020hfv} that holomorphic Chern-Simons theory on twistor space reduces to a 4d WZW model on $\R^4$. To discuss holomorphic Chern-Simons theory on twistor space, it will be useful first to see what happens on a Calabi-Yau 3-fold $C$.

The Calabi-Yau condition guarantees the existence of a holomorphic 3-form $\Omega_C$, so the action of holomorphic Chern-Simons is given by
\begin{align*}
    \int_C\Omega_C\wedge\Tr\left(a\dbar a+\frac{2}{3}a^3\right)
\end{align*}
where $a\in\Omega^{0,1}(C,\g)$ is the gauge field. The action describes the flatness of the partial connection $a$. The gauge transformation one has with this theory is
\begin{align*}
    a\mapsto h^{-1}ah+h^{-1}\dbar h
\end{align*}
for a group valued field $h$.

The story is not as simple on twistor space as there is no choice of holomorphic 3-form. There does exist, however, a choice of meromorphic 3-form $\Omega$, introduced in equation \ref{eq:PT_Omega}. To have a well-defined action, boundary conditions must be imposed on the gauge field at the poles of $\Omega$, which are located at $z = 0,\infty$ ($\infty$ is the point at infinity on $\CP^1$). One choice, as presented in \cite{Bittleston:2020hfv}, which yields the 4d WZW model \cite{Donaldson:1985zz} upon reducing to $\R^4$ is to demand that $a$ vanishes to first order at both poles. It is also important to note that gauge transformations must obey the same boundary conditions.

With these boundary conditions, Bittleston and Skinner \cite{Bittleston:2020hfv} showed that
\begin{align*}
    \int_{\PT}\Omega\wedge\Tr\left(a\dbar a+\frac{2}{3}a^3\right)\leftrightarrow\frac{1}{2}\int_{\R^4}\Tr(J\wedge*J)+\frac{1}{3}\int_{\R^4\times[0,1]}\mu\wedge\Tr(J^3)
\end{align*}
where $\mu = dx^{0\dot{0}}\wedge dx^{1\dot{1}}-dx^{0\dot{1}}\wedge dx^{1\dot{0}}$ (these coordinates are defined in equation \ref{eq:R4_coord}.) and the field $J = -d\sigma\sigma^{-1}$ is the Maurer-Cartan current for the group valued field $\sigma$.

\subsection*{Holomorphic BF theory}\label{subsec:sdym_review}

On twistor space, the action one studies for the holomorphic BF theory is
\begin{align*}
    \int_{\mathbb{PT}}\Tr(bF^{0,2}(a))
\end{align*}
with gauge fields $a\in\Omega^{0,1}(\PT,\g\otimes\mathcal{O})$ and $b\in \Omega^{3,1}(\PT,\g\otimes\mathcal{O}(-4))$. The gauge transformations are
\begin{align*}
    \delta a = \dbar\chi+[a,\chi]\quad \delta b = \dbar\nu+[a,\nu]+[b,\chi]
\end{align*}
with $\chi\in \Omega^{0,0}(\PT,\g\otimes\mathcal{O})$ and $\nu\in \Omega^{3,0}(\PT,\g\otimes\mathcal{O}(-4))$. Note that this makes sense since the Lagrangian must live in the canonical bundle $K_\mathbb{PT} = \mathcal{O}(-4)$ so that it can be integrated. The equations of motion are easily seen to be $F^{0,2}(a) = 0$ and $\dbar_ab=0$.

On $\R^4$, the theory one finds is SDYM \cite{Witten:2003nn,Mason:2005zm,Ward:1977ta,Boels:2006ir}. The field content for SDYM is
\begin{align*}
    A\in\Omega^1(\R^4,\g)\quad B\in\Omega^2_-(\R^4,\g)
\end{align*}
and the action is
\begin{align*}
    \int_{\R^4}\Tr(BF(A)_-).
\end{align*}
The equation of motion is the vanishing of the self-dual part of the field strength. Note that the field $B$ is a negative helicity gluon while $A$ is a positive helicity gluon.

\subsection{Amplitudes in the spinor-helicity formalism}

I will make use of the spinor-helicity formalism, where massless states can be decomposed as
\begin{align*}
    p_{\alpha\dot{\alpha}} = \lambda_\alpha\tilde{\lambda}_{\dot{\alpha}}
\end{align*}
for a pair of spinors $(\lambda_\alpha,\tilde{\lambda}_{\dot{\alpha}})$. Notice that these are the same dotted and undotted $SU(2)$ spinor indices introduced in the twistor review \ref{subsec:twistor_review}. We can also fix 
\begin{align*}
    \lambda_\alpha = \begin{pmatrix}1&z\end{pmatrix}
\end{align*}
for $z\in\CP^1$ by scaling both $\lambda_\alpha$ and $\tilde{\lambda}_{\dot{\alpha}}$ while keeping $p_{\alpha\dot{\alpha}}$ fixed. Notice that this allows us to define polynomials in $z\in\CP^1$ using tensors. A degree $4$ polynomial $H(z)$, for example, corresponds to a rank $4$ symmetric tensor $H^{\alpha\beta\gamma\delta}$ by defining
\begin{align*}
    H(z) = H^{\alpha\beta\gamma\delta}\lambda_\alpha \lambda_\beta \lambda_\gamma \lambda_\delta.
\end{align*}
This notation will be used in the rest of the paper.

Since scattering states with left-handed spinor-helicity $\lambda_\alpha$ lift to twistor representatives supported at the point $z\in\CP^1$, the coordinate $z$ is identified with the $\CP^1$ coordinate on twistor space.

Indices can be raised and lowered using the Levi-Civita symbols $\epsilon_{\alpha\beta},\epsilon_{\dot{\alpha}\dot{\beta}}$ with the conventions
\begin{align*}
    \lambda_\alpha = \epsilon_{\alpha\beta}\lambda^\beta\quad \epsilon^{\alpha\gamma}\epsilon_{\gamma\beta} = \delta^\alpha_\beta
\end{align*}
and similarly for dotted indices. I will also make use of the angle and square bracket notation for $\text{SL}(2,\C)$ invariant spinor contractions in the following manner:
\begin{align*}
    a_\alpha b^\alpha = \langle ab\rangle\quad a_{\dot{\alpha}}b^{\dot{\alpha}} = [ab].
\end{align*}
I will often label momentum states by integers, in which case I use the notation 
\begin{align*}
    \lambda_{1\alpha}\lambda_2^\alpha = \langle12\rangle = z_1-z_2\quad \tilde{\lambda}_{1\dot{\alpha}}\tilde{\lambda}_2^{\dot{\alpha}} = [12].
\end{align*}

\subsection*{Useful identities}

Momentum conservation will be used throughout this paper. This amounts to applying the relation
\begin{align*}
    \sum_{i=1}^np_i^{\alpha\dot{\alpha}} = \sum_{i=1}^n\lambda_i^\alpha\tilde{\lambda}_i^{\dot{\alpha}} = 0
\end{align*}
for an $n$ point diagram. By contracting with dotted and undotted spinors, the above gives several algebraic relations on the variables $\langle ij\rangle$ and $[kl]$. For the sum of two momenta, we also have
\begin{align*}
    p_{12}^2 = (p_1+p_2)^2 = 2\langle12\rangle[12].
\end{align*}
I will also make extensive use of the Schouten identity
\begin{align*}
    \lambda_1^\alpha\langle23\rangle+\lambda_3^\alpha\langle12\rangle+\lambda_2^\alpha\langle31\rangle = 0
\end{align*}
and similarly for dotted spinors.

\subsection{Celestial chiral algebras and form factors}

I will take the perspective on celestial chiral algebras pioneered by Costello and Paquette \cite{Costello:2022wso}. For illustrative purposes, I will review the celestial chiral algebra for self-dual gauge theory. 

Massless gauge theory states of positive and negative helicity are denoted
\begin{align*}
    J_a(\tilde{\lambda};z) \quad \tilde{J}_a(\tilde{\lambda};z) 
\end{align*}
where $a$ is a colour index and the energy has been absorbed into $\tilde{\lambda}$. Note that I chose to describe the state using $z$ while one could equivalently use the spinor $\lambda$. The spacetime interpretation of these states is to identify $J$ with the positive helicity gluon $A$ and $\tilde{J}$ with the negative helicity gluon $B$ of SDYM introduced in section \ref{subsec:sdym_review}. The OPEs of the chiral algebra can then be computed from collinear limits of the spacetime theory. They are \cite{Costello:2022wso}
\begin{align*}
    J_a(\tilde{\lambda}_1;z_1)J_b(\tilde{\lambda}_2;z_2)\sim\frac{f_{ab}^c}{\langle12\rangle}J_c(\tilde{\lambda}_1+\tilde{\lambda}_2;z_2)\qquad J_a(\tilde{\lambda}_1;z_1)\tilde{J}_b(\tilde{\lambda}_2;z_2)\sim\frac{f_{ab}^c}{\langle12\rangle}\tilde{J}_c(\tilde{\lambda}_1+\tilde{\lambda}_2;z_2).
\end{align*}
The positive helicity states form the Kac-Moody algebra at level zero for the Lie algebra $\mathfrak{sl}(N_c)[\tilde{\lambda}^{\dot{1}},\tilde{\lambda}^{\dot{2}}]$, which is also known as the $\mathcal{S}$-algebra \cite{Strominger:2021mtt}. 

This can also be framed in terms of soft generators $J_a[p,q](z)$ so that the hard generator is written
\begin{align*}
    J_a(\tilde{\lambda};z) = \sum_{m,n\geq0}\frac{\tilde{\lambda}_i^m\tilde{\lambda}_i^n}{m!n!}J_a[m,n](z)
\end{align*}
and the OPEs are
\begin{align*}
    J_a[p,q](z_1)J_b[r,s](z_2)\sim \frac{f_{ab}^c}{\langle12\rangle}J_c[p+r,q+s](z_2).
\end{align*}
There is a similar formulation for the negative helicity states. The conformal dimensions of the soft states are outlined in table \ref{tb:soft_states}.

\begin{table}[h]
    \centering
    \renewcommand{\arraystretch}{1.3}
    \begin{tabular}{c | c | c}
        \textbf{States} & \textbf{Labels} & \textbf{Conformal dimension} \\ 
        \midrule
        $J_a[m,n]$ & $m,n \geq 0$ & $1 - \frac{(m+n)}{2}$ \\
        $\tilde{J}_a[m,n]$ & $m,n \geq 0$ & $-1-\frac{(m+n)}{2}$ \\
    \end{tabular}
    \caption{The states of the celestial chiral algebra for SDYM.}
    \label{tb:soft_states}
\end{table}

\subsection*{Form factors}

Costello and Paquette recently developed a program for computing gauge theory amplitudes from self-dual form factors using the celestial chiral algebra \cite{Costello:2022wso}. The basic idea is to choose a conformal block for the chiral algebra corresponding to the insertion of a local operator in the spacetime theory. One can then compute an $n$-point form factor by repeatedly performing chiral algebra OPEs in the appropriate conformal block. As described in \cite{Costello:2022wso}, these form factors compute scattering integrands in full Yang-Mills theory when the correct operator insertions are chosen. This theory is developed and explained in detail in the beautiful paper \cite{Costello:2022wso}.

Let's review the simple example of the colour-ordered $--+$ amplitude. The correct conformal block to use is
\begin{align*}
    \langle\Tr(B^2)|\tilde{J}^a(\tilde{\lambda}_1;z_1)\tilde{J}^b(\tilde{\lambda}_2;z_2)\rangle = \Tr(ab)(z_1-z_2)^2 = \Tr(ab)\langle12\rangle^2
\end{align*}
which corresponds to an insertion of $\Tr(B^2)$ in the spacetime theory as explained in \cite{Costello:2022wso}. To compute the amplitude from the chiral correlator, we then insert positive helicity states and compute OPEs:
\begin{align*}
    \langle\Tr(B^2)|\tilde{J}^a(\tilde{\lambda}_1;z_1)\tilde{J}^b(\tilde{\lambda}_2;z_2)J^c(\tilde{\lambda}_3;z_3)\rangle & = f_d^{cb}\frac{1}{\langle23\rangle}\langle\Tr(B^2)|\tilde{J}^a(\tilde{\lambda}_1;z_1)\tilde{J}^d(\tilde{\lambda}_2+\tilde{\lambda}_3;z_2)\rangle\\
    &\quad +f_d^{ca}\frac{1}{\langle13\rangle}\langle\Tr(B^2)|\tilde{J}^d(\tilde{\lambda}_1+\tilde{\lambda}_3;z_1)\tilde{J}^b(\tilde{\lambda}_2;z_2)\rangle\\
    & = \frac{\langle12\rangle^3}{\langle13\rangle\langle23\rangle}f^{abc}.
\end{align*}
Where I have ignored colour factors above and only computed the colour ordered term. This matches the famous Parke-Taylor formula at 3 points \cite{Parke:1986gb}. One can then easily see iteratively that the $n$ point Parke-Taylor formula can be derived by inserting arbitrarily many positive helicity states as done in \cite{Costello:2022wso}.

\section{Building theories on the double cover}\label{sec:model_build}

In this section, I build a double cover of the twistor space of $\R^4$, which I denote as $\PT_\Sigma$ for $\Sigma$ an (hyper)elliptic curve. I then define the holomorphic Chern-Simons and BF theory on $\PT_\Sigma$. The cases of interest for this work are holomorphic Chern-Simons theory when $\Sigma$ is a genus 3 curve and holomorphic BF theory when $\Sigma$ is a genus 1 curve. These theories are shown to pushforward to twistor theories with twice the number of fields. The examples I chose both Penrose transform to deformations of SDYM theory. I write down these deformations up to third-order interactions using symmetry.

\subsection{Building the double cover}

In section \ref{sec:bgrnd}, it was stated that the twistor space of $\R^4$ is the total space of the rank 2 complex vector bundle $\mathcal{O}(1)\oplus\mathcal{O}(1)\to\CP^1$. If we define an elliptic curve $\Sigma$ by the relation $w^2=H(z)$ for $H(z)$ a polynomial on $\CP^1$, then this defines a double cover $\Sigma\xrightarrow[2:1]{\pi}\CP^1$. In \cite{Costello:2021bah}, Costello proposed building a double cover of twistor space by considering the total space of
\begin{align*}
    \pi^*\mathcal{O}(1)\oplus\pi^*\mathcal{O}(1)\to\Sigma 
\end{align*}
which we will denote by $\PT_\Sigma$. As a real manifold, there is a diffeomorphism 
\begin{align*}
    \PT_\Sigma\cong\Sigma\times\R^4.
\end{align*}
In this paper, I focus on the cases where $H(z)$ is a degree $4$ and $8$ polynomial (so $\Sigma$ has genus 1 and 3, respectively). The degree $8$ case produces a Calabi-Yau 3-fold, which was Costello's original motivation for considering $\PT_\Sigma$ as the topological string can be studied directly without demanding special conditions on the gauge fields. 

As discussed in section \ref{sec:bgrnd}, the poles in the twistor space volume form $\Omega$ require fixing boundary conditions on the gauge field of the holomorphic Chern-Simons theory. On the double cover, we denote the volume form by $\Omega_\Sigma$ which is holomorphic when $\PT_\Sigma$ is Calabi-Yau. In the genus 1 case, $\PT_\Sigma$ is not Calabi-Yau and we will need a choice of meromorphic volume form as we had in the $\PT$ example. Recall that on $\PT$, we had defined the coordinates $v_{\dot{\alpha}}\in\mathcal{O}(1)$. By abusing notation, we can fix a set of coordinates on $\PT_\Sigma$ as $v_{\dot{\alpha}}\in\pi^*\mathcal{O}(1)$. A standard fact in the theory of elliptic curves is that there is a canonical, holomorphic 1-form $\gamma$ on $\Sigma$. With this, we can define the meromorphic volume form
\begin{align*}
    \Omega_\Sigma = \gamma dv_{\dot{1}}dv_{\dot{2}}\in \Omega^{3,0}(\PT_\Sigma,\pi^*\mathcal{O}(2))
\end{align*}
which has two second-order poles at the preimages of the point at infinity $\pi^{-1}(\infty)$ and we note that we used the fact that pullbacks commute with the tensor product $\pi^*\mathcal{O}(1)\otimes\pi^*\mathcal{O}(1) = \pi^*\mathcal{O}(2)$.

Holomorphic Chern-Simons theory in the genus 1 setting, where the meromorphic volume form is crucial, is discussed in subsection \ref{subsec:equiv_4dCS}.

\subsection{Holomorphic Chern-Simons theory on the genus 3 double cover}

In this section, I compute the twistor theory corresponding to holomorphic Chern-Simons on $\PT_\Sigma$ when $\Sigma$ has genus 3. The genus 3 case is easier since $\PT_\Sigma$ is Calabi-Yau, the genus 1 case requires special boundary conditions and is discussed in subsection \ref{subsec:equiv_4dCS}.

In section \ref{sec:bgrnd}, I asserted that holomorphic Chern-Simons theory can be defined on any Calabi-Yau 3-fold with a gauge field free to live in the trivial bundle. We can perform a similar analysis on $\PT_\Sigma$ when our curve has genus 3. The existence of a holomorphic 3-form $\Omega_\Sigma$ allows us to define
\begin{align}\label{eq:hCS_X}
    \int_{\PT_\Sigma}\Omega_\Sigma\Tr\left(\alpha\dbar\alpha+\frac{2}{3}\alpha^3\right)
\end{align}
with $\alpha\in \Omega^{0,1}(\PT_\Sigma,\g\otimes\mathcal{O}_{\PT_\Sigma})$. On the double cover, I use Greek letters for the gauge fields to avoid confusion. In appendix \ref{app:line_bdl}\footnote{This type of calculation is standard for those with a basic complex geometry background and I have added it only to make this work more widely accessible.}, I showed that $\pi_*\mathcal{O}_{\PT_\Sigma} = \mathcal{O}\oplus\mathcal{O}(-\deg H/2)$. So, our field transforms as
\begin{align*}
    \alpha\xrightarrow[]{\pi_*}a+\tilde{a}\sqrt{H}
\end{align*}
by using the basis $\{1,\sqrt{H}\}$ for the pushforward bundle. Note that we interpret $a\in\mathcal{O}$ as the even field under crossing branches of $\Sigma$ and $\tilde{a}\in\mathcal{O}(-\deg H/2)$ as the odd field. We can then pushforward the action and keep the even terms:
\begin{align}\label{eq:hCS_twistor_action}
     \nonumber\int_{\PT_\Sigma}\Omega_\Sigma\Tr\left(\alpha\dbar\alpha+\frac{2}{3}\alpha^3\right) & \xrightarrow[]{\pi_*} \int_\PT \frac{\Omega}{\sqrt{H}}\Tr\left(2\tilde{a}\dbar a\sqrt{H}+2(\tilde{a}a^2\sqrt{H}+\tilde{a}^3H^{3/2})\right)\\
    &= 2\int_\PT \Omega\Tr\left(\tilde{a}F(a)+\tilde{a}^3H\right).
\end{align}
Where we used the fact that $\Omega_\Sigma\xrightarrow[]{\pi_*}\Omega/\sqrt{H}$. Plugging in $\deg H = 8$, we can find that the 4d field content is a pair of gauge fields\footnote{Fields twisted by $\mathcal{O}(2s-2)$ correspond to spin $|s|$ and helicity $\text{sign}(s)$ fields on spacetime.}
\begin{align*}
    a\in\mathcal{O}\mapsto A\quad \tilde{a}\in\mathcal{O}(-4)\mapsto B.
\end{align*}
The only leading order Lagrangian compatible with the Lorentz symmetry (ie. the only vertex in which Lorentz indices contract appropriately) is given by
\begin{align}\label{eq:genus3_deformation}
    \int_{\R^4}\Tr(BF(A)+B_{\alpha\beta}D_\gamma^{\dot{\alpha}}B_{\delta\rho}D_{\sigma\dot{\alpha}}B_{\zeta\eta}H^{\alpha\beta\gamma\delta\rho\sigma\zeta\eta})+\text{higher order terms}
\end{align}
where $D$ is a covariant derivative $D = d+[A,\cdot]$. The constant symmetric tensor $H^{\alpha\beta\gamma\delta\rho\sigma\zeta\eta}$ defines the polynomial corresponding to the curve by
\begin{align*}
    H(z) = H^{\alpha\beta\gamma\delta\rho\sigma\zeta\eta}\lambda_\alpha...\lambda_\eta.
\end{align*}

Notice that this is a deformation of the SDYM theory by the cubic vertex $B_{\alpha\beta}d_\gamma^{\dot{\alpha}}B_{\delta\rho}d_{\sigma\dot{\alpha}}B_{\zeta\eta}H^{\alpha\beta\gamma\delta\rho\sigma\zeta\eta}$ which I have made gauge invariant by replacing the derivatives with covariant derivatives. In section \ref{sec:bootstrap} I constrain the form of all higher-order terms that contribute to the action, hence determining the 4d lagrangian up to index contractions and coefficients.

\subsection{Holomorphic BF theory on the genus 1 double cover}

Recall in section \ref{sec:bgrnd} we defined holomorphic BF theory on twistor space by imposing the condition that $b$ be a section of $\mathcal{O}(-4)$ so that the action could be integrated against the volume form $\Omega\in\Omega^{3,0}(\PT,\mathcal{O}(4))$. The field $a$ was then allowed to live in the trivial bundle.

On the double cover, the analog of the field $b$, which I call $\beta$, plays a similar role in that we use it to cancel the poles in the integration measure. The analog of $a$, denoted by $\alpha$, is then free to live in the trivial bundle on the double cover. The field content is then
\begin{align*}
    \alpha\in\Omega^{0,1}(\PT_\Sigma,\g\otimes\mathcal{O}_{\PT_\Sigma})\quad\beta\in \Omega^{3,1}(\PT_\Sigma,\g\otimes K_\Sigma\otimes\pi^*\mathcal{O}(-2))
\end{align*}
and the action is
\begin{align*}
    \int_{\PT_\Sigma}\Tr(\beta F^{0,2}(\alpha)).
\end{align*}
In appendix \ref{app:line_bdl}, I compute the pushforward of the trivial and canonical bundle on $\Sigma$. Using those calculations, we have that 
\begin{align*}
    \pi_*(\mathcal{O}_\Sigma) = \mathcal{O}\oplus\mathcal{O}(-\deg H/2)\quad \pi_*(K_\Sigma\otimes\pi^*\mathcal{O}(-2)) = \mathcal{O}(\deg H/2-4)\oplus\mathcal{O}(-4).
\end{align*}
Concretely, we can think of each of the fields $\alpha$ and $\beta$ as splitting into even and odd fields: a component that does not change sign when crossing branches and one that picks up a minus sign. If we take $\{1,\sqrt{H(z)}\}$ as a basis for the functions on the pushforward bundle, then we can write 
\begin{align*}
    \alpha \xrightarrow[]{\pi_*} a+\tilde{a}\sqrt{H}\quad \beta\xrightarrow[]{\pi_*} b+\tilde{b}\sqrt{H}
\end{align*}
where $a\in \mathcal{O}$, $\tilde{a}\in\mathcal{O}(-\deg H/2)$, $b\in\mathcal{O}(-4)$, and $\tilde{b}\in\mathcal{O}(\deg H/2-4)$. We think of $a,b$ as the even fields and $\tilde{a},\tilde{b}$ as the odd fields. Plugging these expressions into the holomorphic BF action and keeping the even terms, we get
\begin{align}\label{eq:genus1_twistor_action}
    \int_{\PT_\Sigma}\Tr(\beta F^{0,2}(\alpha))\xrightarrow[]{\pi_*} \int_{\PT}\Tr(bF(a) +\tilde{b}\dbar\tilde{a} + a\tilde{a}\tilde{b}+b\tilde{a}\tilde{a}H(z)).
\end{align}
We have thus determined the twistor theory corresponding to holomorphic BF theory on $\PT_\Sigma$. Plugging in $\deg H = 4$ (since $\Sigma$ is genus 1), we can also determine the field content of the corresponding theory on $\R^4$:
\begin{align*}
    a\in\mathcal{O}\mapsto A\quad b\in\mathcal{O}(-4)\mapsto B\quad \tilde{a}\in\mathcal{O}(-2)\mapsto\Phi\quad \tilde{b}\in\mathcal{O}(-2)\mapsto\tilde{\Phi}.
\end{align*}
Where $A,B$ are gauge fields and $\Phi,\tilde{\Phi}$ are scalars. Using symmetry, we can constrain the $\R^4$ lagrangian up to cubic interactions to be
\begin{align*}
\int_{\mathbb{R}^4}\Tr(BF(A)+\tilde{\Phi}D*D\Phi+B_{\alpha\beta}D_{\dot{\alpha}\delta}\Phi D^{\dot{\alpha}}_\gamma\Phi H^{\alpha\beta\delta\gamma})+\text{higher order terms}  
\end{align*}
where $H^{\alpha\beta\gamma\delta}$ is the constant symmetric tensor defining the polynomial $H(z)$ and $D = d+[A,\cdot]$. Our polynomial is defined as
\begin{align*}
    H(z) = H^{\alpha\beta\gamma\delta}\lambda_\alpha \lambda_\beta \lambda_\gamma \lambda_\delta\quad \lambda_\alpha = \begin{pmatrix}
        1&z
    \end{pmatrix}.
\end{align*}

To conclude our discussion of this genus 1 setting, we note that the theory we arrived at on $\R^4$ is a deformation of the SDYM theory coupled to a complex massless scalar. I leave the task of determining the higher order interactions for future work.

\subsection{Remarks on the section}\label{subsec:rmks}

In this section, we have constructed two examples of twistorial theories from theories living on $\PT_\Sigma$ for different choices of the curve $\Sigma$. This method should apply more generally to 'nice' theories by following the steps:
\begin{enumerate}
    \item Fix the curve $\Sigma$, then make a choice of holomorphic/meromorphic volume form $\Omega_\Sigma$ and identify which bundle it is a section of. When $\PT_\Sigma$ is not Calabi-Yau, this bundle will typically be $\mathcal{O}_\Sigma\otimes\pi^*\mathcal{O}(n)$ for $n$ an integer depending on $\Sigma$.
    \item Demand that the field content on twistor space be such that the kinetic term lives in the canonical bundle of $\PT_\Sigma$. This can be thought of as imposing boundary conditions on the fields.
    \item Compute the pushforward of the bundles in which the fields live and write the fields in terms of even and odd components using the basis $\{1,\sqrt{H}\}$ on the pushforward bundle.
    \item Plug the pushforward of the fields into the $\PT_\Sigma$ action.
\end{enumerate}

As mentioned earlier, this works for 'nice' theories. Let's go over an example of a theory that I have not understood how to treat using these methods. The Kodaira-Spencer theory \cite{Bershadsky:1993cx} which describes the closed string sector of the topological string is described by the action 
\begin{align*}
    \int_C\Omega_C\wedge\left(\dbar\mu\partial^{-1}\mu+\frac{1}{3}\mu^3\right)\lrcorner\Omega_C
\end{align*}
for $C$ a Calabi-Yau 3-fold and Beltrami differential $\mu\in\Omega^{0,1}(C,T^{1,0}C)$ which is divergence free. This action is of particular interest as it was shown in \cite{Costello:2019jsy} that coupling $SO(8)$-holomorphic Chern-Simons to Kodaira-Spencer theory cancels the anomalies to all loop orders. This allows one to study holomorphic Chern-Simons at the quantum level. One motivation for us to explore the quantization of theories on $\PT_\Sigma$ is to find a holographic correspondence similar to the one studied in \cite{Costello:2023hmi}. For now, we remark that the unusual kinetic term in the Kodaira-Spencer action makes it unclear what the corresponding twistor action will look like.

\section{Celestial chiral algebras}\label{sec:celestial}

Celestial chiral algebras dual to twistor theories have been studied in recent works \cite{Costello:2022wso,Costello:2023hmi,Bittleston:2023bzp} as the twistor approach provides a clear top-down construction of celestial dualities. The purpose of this section is not to build the full chiral algebra dual to a new twistor theory, but to provide supporting evidence for the following main result:

\begin{conjecture}
The celestial chiral algebra of a holomorphic theory defined on $\PT_\Sigma$ lives on the curve $\Sigma$.
\end{conjecture}
To support this claim, I demonstrate that the tree-level collinear singularities in the two deformations of SDYM theory computed in section \ref{sec:model_build} form chiral algebras that live on the curve $\Sigma$.

\subsection{The genus three deformation}\label{subsec:genus3_OPE}

Recall that holomorphic Chern-Simons theory on $\PT_\Sigma$ for $\Sigma$ a genus 3 curve descended to the following spacetime theory
\begin{align*}
    \int_{\R^4}\Tr(BF(A)+B_{\alpha\beta}D_\gamma^{\dot{\alpha}}B_{\delta\rho}D_{\sigma\dot{\alpha}}B_{\zeta\eta}H^{\alpha\beta\gamma\delta\rho\sigma\zeta\eta})+\text{higher order terms}.
\end{align*}

I will show that its celestial chiral algebra is precisely the level 0 Kac-Moody algebra for $\mathfrak{sl}(N_c)[\tilde{\lambda}_{\dot{1}},\tilde{\lambda}_{\dot{2}}]$ on $\Sigma$. 

\begin{figure}
    \centering
    \includegraphics[width=0.5\linewidth]{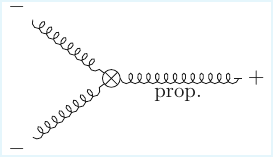}
    \caption{The new $--\to+$ splitting function introduced by the $BDBDBH$ vertex (denoted by the $\otimes$ symbol).}
    \label{fig:genus3_splitfn}
\end{figure}

\subsection*{Building the chiral algebra}

The states of the chiral algebra are organized in table \ref{tb:soft_states}. Let's forget about the interaction term involving the tensor $H^{\alpha\beta\gamma\delta\rho\sigma\zeta\eta}$ in the Lagrangian for a moment. The simplest OPE is
\begin{align}\label{eq:OPE}
    J_a[p,q](z_1)J_b[r,s](z_2)\sim \frac{f_{ab}^c}{\langle12\rangle}J_c[p+r,q+s](z_2).
\end{align}
where $\langle12\rangle=z_1-z_2$. We also make use of hard generators
\begin{align*}
    J_a(\tilde{\lambda};z) = \sum_{m,n\geq0}\frac{\tilde{\lambda}_i^m\tilde{\lambda}_i^n}{m!n!}J_a[m,n](z)
\end{align*}
so that the OPE in \ref{eq:OPE} becomes
\begin{align*}
    J_a(\tilde{\lambda}_1;z_1)J_b(\tilde{\lambda}_2;z_2)\sim\frac{f_{ab}^c}{\langle12\rangle}J_c(\tilde{\lambda}_1+\tilde{\lambda}_2;z_2).
\end{align*}
As explained in the background section \ref{sec:bgrnd}, this forms the Kac-Moody algebra at level 0 for the Lie algebra $\mathfrak{sl}(N_c)[\tilde{\lambda}^{\dot{1}},\tilde{\lambda}^{\dot{2}}]$. The rest of the OPEs\footnote{Note that in the celestial chiral algebra construction, only singular terms are considered in the OPE.} are
\begin{align*}
    J_a(\tilde{\lambda}_1;z_1)\tilde{J}_b(\tilde{\lambda}_2;z_2)\sim\frac{f_{ab}^c}{\langle12\rangle}\tilde{J}_c(\tilde{\lambda}_1+\tilde{\lambda}_2;z_2)\qquad \tilde{J}_a(\tilde{\lambda}_1;z_1)\tilde{J}_b(\tilde{\lambda}_2;z_2)\sim 0.
\end{align*}
Next, if we deform the SDYM action by adding the vertex $B_{\alpha\beta}D_\gamma^{\dot{\alpha}}B_{\delta\rho}D_{\sigma\dot{\alpha}}B_{\zeta\eta}H^{\alpha\beta\gamma\delta\rho\sigma\zeta\eta}$, then we introduce a new splitting function and hence OPE to the chiral algebra. The new $--\to+$ splitting function is given by inserting this vertex and plane waves on external legs along with a propagator (see Fig. \ref{fig:genus3_splitfn})
\begin{align*}
    \lambda_{1\delta}\lambda_{1\gamma}\lambda_{2\rho}\lambda_{2\sigma}&H^{\alpha\beta\delta\gamma\rho\sigma\epsilon}_\eta p_{1\alpha}^{\dot{\alpha}}p_{2\beta\dot{\alpha}}(p_1+p_2)_{\dot{\beta}}^{(\eta}\epsilon_\epsilon^{\chi)}\frac{e^{i(p_1+p_2)\cdot x}}{p_{12}^2}f_{ab}^c = \frac{\langle1^32^4|H\rangle^\chi \langle q2\rangle(\tilde{\lambda}_1+\tilde{\lambda}_2)_{\dot{\beta}}e^{i\lambda_2(\tilde{\lambda}_1+\tilde{\lambda}_2)x}}{\langle q2\rangle\langle12\rangle}f_{ab}^c \\
    & = \frac{\langle1^32^5|H\rangle q^\chi(\tilde{\lambda}_1+\tilde{\lambda}_2)_{\dot{\beta}}  e^{i\lambda_2(\tilde{\lambda}_1+\tilde{\lambda}_2)x}}{\langle12\rangle\langle q2\rangle}f_{ab}^c+\frac{\langle q1^32^4|H\rangle \lambda_2^\chi (\tilde{\lambda}_1+\tilde{\lambda}_2)_{\dot{\beta}}  e^{i\lambda_2(\tilde{\lambda}_1+\tilde{\lambda}_2)x}}{\langle12\rangle\langle q2\rangle}f_{ab}^c\\
    & = \frac{H(z_2)}{\langle12\rangle}f_{ab}^cJ_c(\tilde{\lambda}_1+\tilde{\lambda}_2;z_2).
    \end{align*}
I used the Schouten identity in the second equality. In the third equality, I discarded the pure gauge term $\sim \lambda_2^\chi (\tilde{\lambda}_1+\tilde{\lambda}_2)_{\dot{\beta}}  e^{i(p_1+p_2)\cdot x}$ (derivative of a plane wave) and the fact that we are working in the collinear limit $\lambda_1\sim\lambda_2$. So, the deformed chiral algebra has the following OPEs:
\begin{align*}
    \tilde{J}_a(\tilde{\lambda}_1;z_1)\tilde{J}_b(\tilde{\lambda}_2;z_2)&\sim \frac{H(z_2)}{\langle12\rangle}f_{ab}^cJ_c(\tilde{\lambda}_1+\tilde{\lambda}_2;z_2)\\
    J_a(\tilde{\lambda}_1;z_1)J_b(\tilde{\lambda}_2;z_2)&\sim\frac{f_{ab}^c}{\langle12\rangle}J_c(\tilde{\lambda}_1+\tilde{\lambda}_2;z_2)\\
    J_a(\tilde{\lambda}_1;z_1)\tilde{J}_b(\tilde{\lambda}_2;z_2)&\sim\frac{f_{ab}^c}{\langle12\rangle}\tilde{J}_c(\tilde{\lambda}_1+\tilde{\lambda}_2;z_2).
\end{align*}
We can then define new states
\begin{align}\label{eq:genus3_OPE}
    J_{\pm,a}(\tilde{\lambda}_1;z_1) = J_a(\tilde{\lambda}_1;z_1)\pm\frac{1}{\sqrt{H}}\tilde{J}_a(\tilde{\lambda}_1;z_1)
\end{align}
which have OPEs
\begin{align*}
    J_{\pm,a}(\tilde{\lambda}_1;z_1)J_{\pm,b}(\tilde{\lambda}_2;z_2) & \sim 2\frac{f_{ab}^c}{\langle12\rangle}J_{\pm,c}(\tilde{\lambda}_1+\tilde{\lambda}_2;z_2)
\end{align*}
and
\begin{align*}
    J_{\pm,a}(\tilde{\lambda}_1;z_1)J_{\mp,b}(\tilde{\lambda}_2;z_2) & \sim 0.
\end{align*}
These are computed in appendix \ref{app:OPE}. We thus see that our chiral algebra is formed by two states that form the level-0 Kac-Moody algebra when the states are on the same branch of $\Sigma$ (their $\pm$ label match) and commute when on different branches. Moreover, their labels are swapped upon crossing branches. It is then easy to see that $z_i$ is simply the local variable on $\Sigma$ so this chiral algebra is equivalent to a single state $\mathcal{J}_a(\tilde{\lambda}_1;z_1)$ where we can now take $z\in\Sigma$ with OPE
\begin{align*}
    \mathcal{J}_{a}(\tilde{\lambda}_1;z_1)\mathcal{J}_{b}(\tilde{\lambda}_2;z_2) & \sim \frac{f^c_{ab}}{z_1-z_2}\mathcal{J}_{c}(\tilde{\lambda}_1+\tilde{\lambda}_2;z_2).
\end{align*}
This is the Kac-Moody algebra for $\mathfrak{sl}(N_c)[\tilde{\lambda}_{\dot{1}},\tilde{\lambda}_{\dot{2}}]$ at level 0 on $\Sigma$. I have illustrated this transformation of the algebra in Fig. \ref{fig:genus3_CA}.

\begin{figure}
    \centering
    \includegraphics[width=\linewidth]{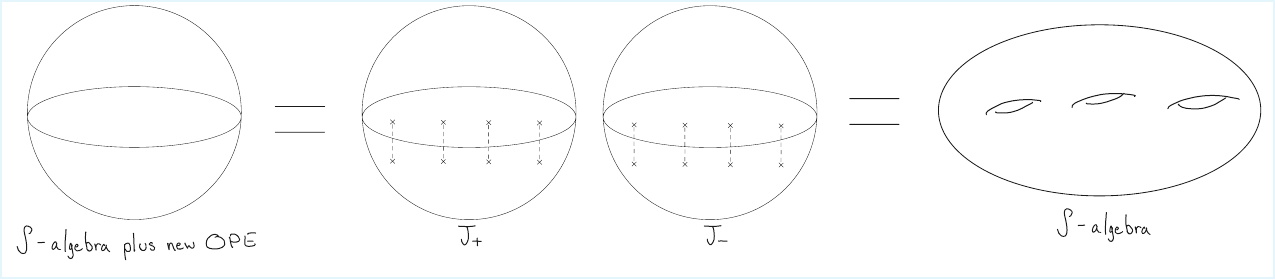}
    \caption{The transformation of the celestial chiral algebra for SDYM when we add the genus 3 deformation. The chiral algebra starts life as the ordinary SDYM chiral algebra with the added $\tilde{J}(0)\tilde{J}(z)\sim \frac{H(z)}{z}J(z)$ OPE. After defining the new states $J_\pm$, we recast this deformed algebra as copies of the $\mathcal{S}$-algebra on two copies of the celestial sphere with branch cuts. Finally, we see that this is the $\mathcal{S}$-algebra for a single state on the genus 3 curve $\Sigma$.}
    \label{fig:genus3_CA}
\end{figure}

\subsection{The genus one deformation}\label{subsec:genus1_OPE}

Recall the action introduced in section \ref{sec:model_build}
\begin{align*}
\int_{\mathbb{R}^4}\Tr(BF(A)+\tilde{\Phi}D_A*D_A\Phi+B_{\alpha\beta}D_{\dot{\alpha}\delta}\Phi D^{\dot{\alpha}}_\gamma\Phi H^{\alpha\beta\delta\gamma})+\text{higher order terms}.
\end{align*}

\subsection*{Building the chiral algebra}

The states of the chiral algebra are organized in table \ref{tb:chiral_algebra}. The simplest non-vanishing OPEs (ignoring the new $BD\Phi D\Phi H$ vertex) are
\begin{align*}
    J_a(\tilde{\lambda}_1;z_1)\tilde{J}_b(\tilde{\lambda}_2;z_2) & \sim \frac{f_{ab}^c}{\langle12\rangle}\tilde{J}_c(\tilde{\lambda}_1+\tilde{\lambda}_2;z_2)\quad J_a(\tilde{\lambda}_1;z_1)\mathcal{O}_b(\tilde{\lambda}_2;z_2) \sim \frac{f_{ab}^c}{\langle12\rangle}\mathcal{O}_c(\tilde{\lambda}_1+\tilde{\lambda}_2;z_2)\\
    J_a(\tilde{\lambda}_1;z_1)\tilde{\mathcal{O}}_b(\tilde{\lambda}_2;z_2) & \sim \frac{f_{ab}^c}{\langle12\rangle}\tilde{\mathcal{O}}_c(\tilde{\lambda}_1+\tilde{\lambda}_2;z_2)\quad \mathcal{O}_a(\tilde{\lambda}_1;z_1)\tilde{\mathcal{O}}_b(\tilde{\lambda}_2;z_2) \sim \frac{f_{ab}^c}{\langle12\rangle}\tilde{J}_c(\tilde{\lambda}_1+\tilde{\lambda}_2;z_2)\\
    J_a(\tilde{\lambda}_1;z_1)J_b(\tilde{\lambda}_2;z_2)&\sim\frac{f_{ab}^c}{\langle12\rangle}J_c(\tilde{\lambda}_1+\tilde{\lambda}_2;z_2).
\end{align*}

\begin{table}[h]
    \centering
    \renewcommand{\arraystretch}{1.3}
    \begin{tabular}{c | c | c}
        \textbf{States} & \textbf{Labels} & \textbf{Conformal dimension} \\ 
        \midrule
        $J_a[m,n]$ & $m,n \geq 0$ & $1 - \frac{(m+n)}{2}$ \\
        $\mathcal{O}_a[m,n]$ & $m,n \geq 0$ & $-\frac{(m+n)}{2}$ \\
        $\tilde{J}_a[m,n]$ & $m,n \geq 0$ & $-1-\frac{(m+n)}{2}$ \\
        $\tilde{\mathcal{O}}_a[m,n]$ & $m,n \geq 0$ & $ - \frac{(m+n)}{2}$ \\
    \end{tabular}
    \caption{The states of the celestial chiral algebra for SDYM coupled to two massless scalars. $J,\tilde{J}$ correspond to the gauge fields $A,B$ and $\mathcal{O},\tilde{\mathcal{O}}$ correspond to the scalars $\Phi,\tilde{\Phi}$. Here $a$ is a Lie algebra index.}
    \label{tb:chiral_algebra}
\end{table}

Adding the vertex $B_{\alpha\beta}D_{\dot{\alpha}\delta}\Phi D^{\dot{\alpha}}_\gamma\Phi H^{\alpha\beta\delta\gamma}$ introduces new OPEs to the chiral algebra. I will now explain that the presence of this vertex deforms the Kac-Moody type algebra into a Kac-Moody type algebra on $\Sigma$!

The chiral algebra OPEs are exactly the tree-level splitting functions of the 4d theory. The new splitting functions coming from the $BD\phi D\phi H$ vertex are depicted in Fig. \ref{fig:feyn}. Let's compute the splitting functions for the new vertex.

First, we compute the $\mathcal{O}_a(\tilde{\lambda}_1;z_1)\tilde{J}_b(\tilde{\lambda}_2;z_2)$ OPE in the presence of this vertex. Plugging in plane waves $\Phi = e^{ip_1\cdot x}, B_{\gamma\delta} = \lambda_{2\gamma}\lambda_{2\delta}e^{ip_2\cdot x}$, we can compute the OPE
\begin{align*}
    H_{\alpha\beta\gamma\delta}\frac{\lambda_1^\alpha\tilde{\lambda}_1^{\dot{\alpha}}\lambda_2^\gamma\lambda_2^\delta p_{12\dot{\alpha}}^\beta}{p_{12}^2}e^{i(p_1+p_2)\cdot x}f_{ab}^c = H_{\alpha\beta\gamma\delta}\frac{\lambda_1^\alpha\lambda_2^\beta\lambda_2^\gamma\lambda_2^\delta}{\langle12\rangle}e^{i(\tilde{\lambda}_1+\tilde{\lambda}_2)^{\dot{\alpha}}\lambda_2^\beta x_{\dot{\alpha}\beta}}f_{ab}^c = \frac{H(z_2)}{\langle12\rangle}f_{ab}^c\tilde{\mathcal{O}}_c(\tilde{\lambda}_1+\tilde{\lambda}_2;z_2)
\end{align*}
where $p_{12\dot{\alpha}}^\beta=(\lambda_1\tilde{\lambda}_1+\lambda_2\tilde{\lambda}_2)_{\dot{\alpha}}^\beta$ and we work in the collinear limit $\lambda_1\sim\lambda_2$.

To compute the $\mathcal{O}_a(\tilde{\lambda}_1;z_1)\mathcal{O}_b(\tilde{\lambda}_2;z_2)$ OPE, first notice that in the gauge $\partial^{\alpha\dot{\alpha}}A_{\alpha\dot{\alpha}}=0$, we can parameterize $A^{\alpha\dot{\alpha}} = \frac{\zeta^\alpha\tilde{\lambda}_1^{\dot{\alpha}}}{\langle\zeta1\rangle}e^{ip_1\cdot x}$ and the propagator $\partial_{(\alpha}^{\dot{\alpha}}\langle A_{\beta)\dot{\alpha}}(x)B^{\gamma\delta}(y)\rangle(p) = \delta^4(x-y)\epsilon_{(\alpha}^\gamma\epsilon_{\beta)}^\delta$  is given by
\begin{align*}
    \langle A_{\alpha\dot{\alpha}}B^{\beta\gamma}\rangle(p) = \frac{p_{\dot{\alpha}}^{(\beta}\epsilon_\alpha^{\gamma)}}{p^2}.
\end{align*}
So the OPE which corresponds to the diagram on the right in Fig. \ref{fig:feyn} can be computed as
\begin{align*}
    e^{ip_1x}e^{ip_2x}H^{\alpha\beta\delta}_\gamma p_{1\dot{\alpha}\alpha}p_{2\beta}^{\dot{\alpha}}&\frac{(p_1+p_2)_{\dot{\beta}}^{(\gamma}\epsilon_\delta^{\rho)}}{p_{12}^2}f_{ab}^c = \frac{\langle 12^2|H\rangle^\rho \langle q2\rangle(\tilde{\lambda}_1+\tilde{\lambda}_2)_{\dot{\beta}} e^{i(p_1+p_2)x}}{\langle q2\rangle\langle12\rangle}f_{ab}^c\\
    & = \frac{\langle12^3|H\rangle q^\rho (\tilde{\lambda}_1+\tilde{\lambda}_2)_{\dot{\beta}} e^{i\lambda_2(\tilde{\lambda}_1+\tilde{\lambda}_2)x}}{\langle q2\rangle\langle12\rangle}f_{ab}^c + \frac{\langle q12^2|H\rangle \lambda_2^\rho (\tilde{\lambda}_1+\tilde{\lambda}_2)_{\dot{\beta}} e^{i\lambda_2(\tilde{\lambda}_1+\tilde{\lambda}_2)x}}{\langle q2\rangle\langle12\rangle}f_{ab}^c\\
    & = \frac{H(z_2)}{\langle12\rangle}f_{ab}^cJ_c(\tilde{\lambda}_1+\tilde{\lambda}_2;z_2).
\end{align*}

\begin{figure}
    \centering
    \includegraphics[width=0.95\linewidth]{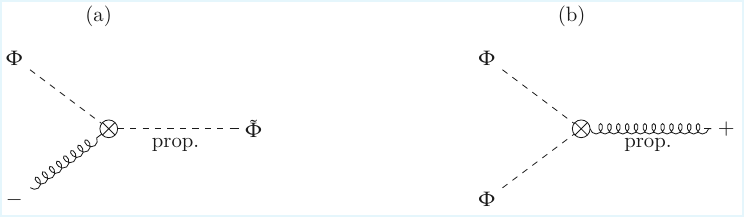}
    \caption{The new splitting functions coming from the introduction of the $BD\Phi D\Phi H$ vertex (denoted by $\otimes$).}
    \label{fig:feyn}
\end{figure}
Where we used the Schouten identity in the second equality. In the last equality, I discarded the pure gauge term proportional to $\lambda_2^\rho (\tilde{\lambda}_1+\tilde{\lambda}_2)_{\dot{\beta}} e^{i\lambda_2(\tilde{\lambda}_1+\tilde{\lambda}_2)x}$ (total derivative of a plane wave) and used the fact that we are working in the collinear limit $\lambda_1\sim\lambda_2$. Altogether, the following OPEs have been added to the theory
\begin{align*}
    \mathcal{O}_a(\tilde{\lambda}_1;z_1)\tilde{J}_b(\tilde{\lambda}_2;z_2)\sim \frac{H(z_2)}{\langle12\rangle}f_{ab}^c\tilde{\mathcal{O}}_c(\tilde{\lambda}_1+\tilde{\lambda}_2;z_2)\quad \mathcal{O}_a(\tilde{\lambda}_1;z_1)\mathcal{O}_b(\tilde{\lambda}_2;z_2)\sim \frac{H(z_2)}{\langle12\rangle}f_{ab}^cJ_c(\tilde{\lambda}_1+\tilde{\lambda}_2;z_2).
\end{align*}

We can form new states for our chiral algebra in terms of $H(z)$
\begin{align}\label{eq:chiral_alg_states}
    J_{\pm,a}(\tilde{\lambda}_1;z_1) = J_{a}(\tilde{\lambda}_1;z_1) \pm \frac{1}{\sqrt{H}}\mathcal{O}_{a}(\tilde{\lambda}_1;z_1) \quad \mathcal{O}_{\pm,a}(\tilde{\lambda}_1;z_1) = \tilde{\mathcal{O}}_{a}(\tilde{\lambda}_1;z_1)\pm\frac{1}{\sqrt{H}}\tilde{J}_{a}(\tilde{\lambda}_1;z_1)
\end{align}
with conformal weights 1 and 0, respectively.

The OPEs of the states defined in equation \ref{eq:chiral_alg_states} take the form:
\begin{align*}
    J_{\pm,a}(\tilde{\lambda}_1;z_1)J_{\pm,b}(\tilde{\lambda}_2;z_2)&\sim \frac{2}{\langle12\rangle}f_{ab}^cJ_{\pm,c}(\tilde{\lambda}_1+\tilde{\lambda}_2;z_2)\quad J_{+,a}(\tilde{\lambda}_1;z_1)J_{-,b}(\tilde{\lambda}_2;z_2)\sim 0\\
    J_{\pm,a}(\tilde{\lambda}_1;z_1)\mathcal{O}_{\pm,b}(\tilde{\lambda}_2;z_2)&\sim \frac{2}{\langle12\rangle}f_{ab}^c\mathcal{O}_{\pm,c}(\tilde{\lambda}_1+\tilde{\lambda}_2;z_2)\quad J_{+,a}(\tilde{\lambda}_1;z_1)\mathcal{O}_{-,b}(\tilde{\lambda}_2;z_2)\sim0.
\end{align*}
The explicit computations of the OPEs are presented in appendix \ref{app:OPE}. I have just shown that the states $J_{\pm,a}(\tilde{\lambda}_1;z_1),\mathcal{O}_{\pm,b}(\tilde{\lambda}_2;z_2)$ define two independent chiral algebras that are swapped under crossing the branch cuts of $\Sigma$. In other words, our states and OPEs above locally define a chiral algebra on the torus $\Sigma$.

\subsection{Speculations on loop corrections}

As presented, the chiral algebras constructed in this section cannot be treated at the quantum level as they correspond to anomalous theories on $\PT_\Sigma$ \cite{Costello:2021bah}. In the holomorphic Chern-Simons case, the anomaly must be cancelled by introducing an axion like field \cite{Costello:2021bah}. This is a delicate subject on the double cover because of the unusual kinetic term I discussed in subsection \ref{subsec:rmks}. I will perform an in-depth study of this theory in a follow-up \cite{Jarov}. In the BF case, one could cure the anomaly by coupling fermions on the double cover and studying
\begin{align*}
    \int_{\PT_\Sigma}\Omega_\Sigma\wedge \Tr(\beta F^{0,2}(\alpha)+\bar{\psi}\dbar_\alpha\psi)
\end{align*}
with $\psi\in\Omega^{0,1}(\PT_\Sigma,\pi^*\mathcal{O}(-1))\otimes\g$, $\bar{\psi}\in\Omega^{0,1}(\PT_\Sigma,\pi^*\mathcal{O}(-1))\otimes R$ (the representation $R$ is chosen to cure an anomaly as explained in \cite{Costello:2021bah}), and now treating $\beta$ as a (0,1) form by introducing the meromorphic volume form. The procedure of pushing to even and odd fields under crossing branches can then be used to treat the fermions. In 4d, this theory will have no scattering at loop level \cite{Costello:2021bah,Costello:2022wso} when the appropriate matter representation is chosen.

While I leave a detailed study of the 4d quantum integrable theory for future study, we can still understand what quantum corrections must appear in the chiral algebra. On the double cover, the state $\mathcal{O}_\pm$ couples to the field $\beta$ and $J_\pm$ couples to $\alpha$. Since the analysis of \cite{Costello:2022upu} was local on twistor space, it naturally extends to our double covering of twistor space. So, we expect to find the following gluon corrections to the chiral algebra
\begin{align*}
J_{\pm,a}[1](z_1)J_{\pm,b}[1](z_2)
\sim&
\frac{[12]}{\langle 12\rangle}
\frac{1}{(2\pi i)^28}
\left(
f_{ea}{}^{c}f_{db}{}^{e}
+
f_{da}{}^{e}f_{eb}{}^{c}
\right)
:J_{\pm,c} \mathcal{O}_\pm^{d}:
(z_1)
\\
&+
\frac{[12]}{\langle 12\rangle}
\frac{1}{(2\pi i)^28}
\left(
g_{ka}{}^{i}g_{jb}{}^{k}
+
g_{ja}{}^{k}g_{kb}{}^{i}
\right)
:M_{\pm,i} \tilde{M}_\pm^{j}:
(z_1)
\\
&-
\frac{[12]}{\langle 12\rangle}
\frac{1}{(2\pi i)^348}
\mathrm{Tr}_{\mathfrak g\oplus\pi R}
\left([t_a,t_b]t_c\right)
\partial_z \mathcal{O}_\pm^{c}(z_1)
\\
&+
\frac{[12]}{\langle 12\rangle^2}
\frac{1}{(2\pi i)^224}
\mathrm{Tr}_{\mathfrak g\oplus\pi R}
\left([t_a,t_b]t_c\right)
\mathcal{O}_\pm^{c}(z_1)\\
J_{\pm,a}[1](z_1)\mathcal{O}_\pm^{b}[1](z_2)
\sim&
\frac{1}{(2\pi i)^28}
\frac{[12]}{\langle 12\rangle}
f_{ca}{}^{e}f_{ed}{}^{b}
:\mathcal{O}_\pm^{c}\mathcal{O}_\pm^{d}:
.
\end{align*}
where we have introduced the chiral algebra states $M_\pm,\tilde{M}_{\pm}$ coupling to the fermions we added to the theory on the double cover.

\section{Comments on the Lax formulation}\label{sec:lax}

The analysis of the previous section gave a concrete interpretation of how the deformation vertices promote the integrable structure of self-dual Yang-Mills to higher genus. Indeed one of the characteristic features of self duality is the Kac-Moody symmetry present in the splitting functions. The promotion of this to a Kac-Moody symmetry on $\Sigma$ by the deformations clearly suggests that the integrable nature of the theory depends on the spectral data of the Riemann surface $\Sigma$. Another defining feature of self-dual Yang-Mills is its presentation in terms of a Lax pair. It is normally challenging to find an explicit Lax formulation for a generic integrable model. However, as explained to me by Lionel Mason, there is a way that we can understand the Lax connection at leading order for these deformed theories. 

The purpose of this section is to gather intuition on how the new cubic vertex in the genus 3 case deforms the spectral data in terms of the Lax formulation. It should be made clear that I do not claim to have fully specified a Lax for these deformations. A proper treatment of the Lax would require knowledge of the higher order corrections that I discuss in section \ref{sec:bootstrap}.  The genus 1 case can also be studied in this language, but the story is not as clean because of the additional scalars.

Let me start by explaining the Lax formulation for ordinary self-dual Yang-Mills and how we extend the formulation to describe the self-dual 2-form $B$. The equations of motion of self-dual Yang-Mills sets the anti-self-dual part of the field strength to zero $F^{\alpha\beta}(A) = 0$. It is well known that this can be recast in the Lax formulation with Lax connection $\mathcal{L}_{\dot{\alpha}}=\lambda_{\alpha}D^\alpha_{\dot{\alpha}}$ and auxiliary linear system $\mathcal{L}_{\dot{\alpha}}s=0$. In the Chalmers-Siegel presentation as a BF theory that we have been studying, we must upgrade this to a rank 2 system since we need to describe the dynamics of the self-dual 2-form $B_{\alpha\beta}$. The leading order genus-3 deformed equation of motion 
\begin{align*}
    F^{\alpha\beta}(A) = D_{\alpha_1}^{\dot{\alpha}}B_{\alpha_2\alpha_3}D_{\alpha_4\dot{\alpha}}B_{\alpha_5\alpha_6}H^{\alpha_1...\alpha_6\alpha\beta}
\end{align*}
can then be recast in terms of the linear system
\begin{align*}
    \begin{pmatrix}
        \lambda^{\alpha}D_{\alpha\dot{\alpha}}&0\\
        0&\lambda^{\alpha}D_{\alpha\dot{\alpha}}
    \end{pmatrix}\begin{pmatrix}
            s\\\tilde{s}_{\alpha\beta\gamma}
        \end{pmatrix}=
        \begin{pmatrix}
            0&\lambda_{\alpha_1}\lambda_{\alpha_2} D_{\alpha_3\dot{\alpha}}B_{\alpha_4\alpha_5}H^{\alpha_1...\alpha_8}\\
            D_{\alpha\dot{\alpha}}B_{\beta\gamma}&0
        \end{pmatrix}\begin{pmatrix}
            s\\\tilde{s}_{\alpha\beta\gamma}
        \end{pmatrix}
\end{align*}
where we have introduced the higher spin frame $\tilde{s}_{\alpha\beta\gamma}$. We can immediately see similarities to the deformed chiral algebra of the previous section. Recall that the the chiral algebra was deformed by the OPE $\tilde{J}_a(\tilde{\lambda}_1;z_1)\tilde{J}_b(\tilde{\lambda}_2;z_2)\sim \frac{H(z_2)}{\langle12\rangle}f_{ab}^cJ_c(\tilde{\lambda}_1+\tilde{\lambda}_2;z_2)$ which lifted the theory to the Kac-Moody algebra on $\Sigma$. The analogous term here is the top right term in the second matrix above. In the $H\to0$ limit, the Lax returns to the one describing BF theory just as the limit returned the chiral algebra to the Kac-Moody algebra on $\CP^1$. 

Similarly to the field redefinition in the chiral algebra I used to interpret it as living on $\Sigma$, I suspect that there is an analogous field redefinition of the gauge fields. Recall that $A$ and $B$ (resp. $J$ and $\tilde{J}$) represented the even and odd degrees of freedom on the double cover. The field redefinition to $J_\pm$ described fields on the two separate branches of the double cover which we could understand as locally defining a Kac-Moody algebra involving one tower $\mathcal{J}$ living on $\Sigma$. I suspect that one could redefine $A,B\to A_\pm$ in a similar way so that we get a Lax formulation locally describing the Lax for a single field $\mathcal{A}$ with spectral parameter living on $\Sigma$. I leave such an analysis for future work.

\section{Bootstrapping vertices part I: The genus 3 deformation}\label{sec:bootstrap}

Since the deformation of SDYM theory coming from the genus 3 case was so simple, it seems plausible that the higher-order vertices can be determined. This section constrains the form of all vertices in the theory and proves that the coefficients are uniquely determined by integrability.

A useful feature due to the integrability of theories descending from twistor space is that they exhibit no tree-level scattering. To see why this is the case, recall that plane wave scattering states are localized at points on the twistor $\CP^1$. To compute the amplitude on twistor space, we can fix a Lorenz gauge with a metric making the radius of the twistor $\CP^1$ arbitrarily large. This then suppresses any propagators that depend on two points $z,z'\in \CP^1$ in the large radius limit. For generic kinematics, plane waves are supported at different points, so they cannot interact. It is possible to leverage this fact to constrain interaction terms in the 4d action.

Throughout this section, I provide supporting evidence towards the integrability of the action proposed in equation \ref{eq:genus3_deformation}. This is in the form of vanishing amplitudes via standard Feynman diagram calculations and generalizations done via the chiral algebra. I will then use symmetry, gauge invariance, and twistor arguments to fix the types of vertices that can descend from the twistor action. The key result from this analysis is proven in subsection \ref{subsec:full_theory}, where I apply my findings in the proof of 

\begin{theorem}\label{thm:4d_action}
    The only integrable and gauge invariant 4d action that is consistent with the Lorentz symmetries of the twistor action given in equation \ref{eq:hCS_twistor_action} is
    \begin{align*}
        \int_{\R^4}\Tr(BF+BDBDBH+\text{terms of order $H^n$})
    \end{align*}
    for $n>1$ and the higher order terms take the form
    \begin{align*}
        B^{2n-1}DBDBH^{n}
    \end{align*}
    where the form of the index contractions and coefficients of these higher-order vertices are fixed by integrability.
\end{theorem}

\subsection{The 4 point amplitude}\label{subsec:4pt_amp}

Let me start the analysis by showing that the leading order lagrangian derived in section \ref{sec:model_build} exhibits basic integrability: The vanishing of the $---+$ amplitude introduced by the new vertex. I will use standard Feynman diagram techniques in this section, and in the next section, I show that one can reproduce and generalize the calculation to 3 minus, $n$ plus amplitudes using the form factor techniques developed in \cite{Costello:2022wso}.

\subsection*{The channel diagrams}

\begin{figure}
    \centering
    \includegraphics[width=0.95\linewidth]{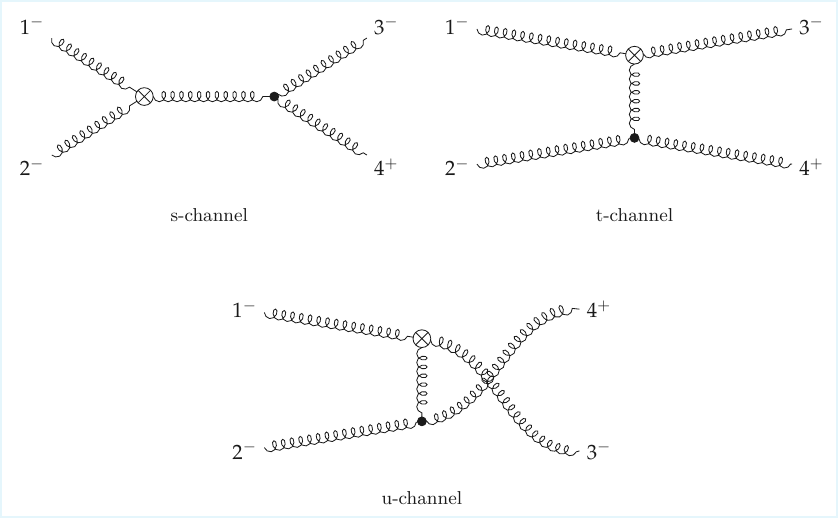}
    \caption{The channel diagrams contributing to the $---+$ amplitude. Note that the $---$ vertex is labelled by the $\otimes$ symbol while the $++-$ vertex is labelled by the $\bullet$ symbol.}
    \label{fig:channel_diagrams}
\end{figure}

We consider the 4-point $---+$ amplitude which includes the three Feynman diagrams shown in Fig. \ref{fig:channel_diagrams}. The basic tools we need are the following: 
\begin{itemize}
    \item Our only propagator is
    \begin{align*}
        \langle A_{\alpha}^{\dot{\alpha}}B_{\beta\delta}\rangle(p) = \frac{\epsilon_{\alpha(\gamma}p_{\beta)}^{\dot{\alpha}}}{p^2}.
    \end{align*}
    \item By gauge fixing our positive helicity gluon $A$, using a reference spinor $q$, external leg contractions with momentum $p_i$ are $\frac{q^{\alpha}\tilde{\lambda}_i^{\dot{\alpha}}}{\langle qi\rangle}$. While external leg contractions for the negative helicity gluons with momentum $p_j$ take the form $\lambda_j^\alpha\lambda_j^\beta$.
    \item The only nontrivial vertex contributions comes from the 3-minus term which adds a factor of $p_{i\gamma}^{\dot{\alpha}}p_{j\sigma\dot{\alpha}}H^{\alpha\beta\gamma\delta\rho\sigma\zeta\eta}$ and the 3-minus, 1-plus term which contributes $p_{i\gamma}^{\dot{\alpha}}H^{\alpha\beta\gamma\delta\rho\sigma\zeta\eta}$.
\end{itemize}

Let's compute the s-channel:
\begin{align*}
    3\frac{q_\alpha\tilde{\lambda}_4^{\dot{\alpha}}\lambda_3^\alpha\lambda_3^\beta}{\langle4q\rangle}\frac{\epsilon_{\beta(\gamma}p_{43\delta)\dot{\alpha}}}{p_{43}^2}H^{\gamma\delta\epsilon\zeta\eta\theta\tau\kappa}p^{\dot{\beta}}_{2\epsilon }p_{1\dot{\beta}\theta}\lambda_{2\zeta}\lambda_{2\eta}\lambda_{1\tau}\lambda_{1\kappa}f_{ab}^ef_{cde} & = -3\frac{\langle q3\rangle[12]}{\langle4q\rangle\langle43\rangle}\langle 1^32^33^2|H\rangle f_{ab}^ef_{cde}.
\end{align*}
using the fact that $p_{43}^2 = 2[43]\langle43\rangle$ and 
\begin{align*}
    \tilde{\lambda}_4^{\dot{\alpha}} \lambda_3^\beta \epsilon_{\beta(\gamma}p_{43\delta)\dot{\alpha}} = \tilde{\lambda}_{4}^{\dot{\alpha}} \lambda_{3(\gamma}(\lambda_{4\delta)}\tilde{\lambda}_4^{\dot{\alpha}}+\lambda_{3\delta)}\tilde{\lambda}_{3\dot{\alpha}}) = [34]\lambda_{3(\gamma}\lambda_{3\delta)} = 2[34]\lambda_{3\gamma}\lambda_{3\delta} 
\end{align*}
since $\tilde{\lambda}_4^{\dot{\alpha}}\tilde{\lambda}_{4\dot{\alpha}}=0$. Similarly, the t-channel is
\begin{align*}
    3\frac{\langle q1\rangle[32]}{\langle 4q\rangle\langle41\rangle}\langle 1^22^33^3|H\rangle f_{ad}^ef_{bce}
\end{align*}
and the u-channel is
\begin{align*}
    3\frac{\langle q2\rangle[31]}{\langle4q\rangle\langle42\rangle}\langle 1^32^23^3|H\rangle f_{ac}^ef_{bde}.
\end{align*}
Applying the Schouten identity $\langle pq \rangle \langle rs \rangle + \langle pr \rangle \langle sq \rangle + \langle ps \rangle \langle qr \rangle = 0
$ gives
\begin{align*}
    -3\frac{\langle q3\rangle[12]}{\langle4q\rangle\langle43\rangle}\langle 1^32^33^2|H\rangle f_{ab}^ef_{cde} & = -3\frac{\langle q3\rangle\lambda_1[12]}{\langle4q\rangle\langle43\rangle}\langle 1^22^33^2|H\rangle f_{ab}^ef_{cde}\\
    & = 3\frac{[12]}{\langle4q\rangle\langle43\rangle}(\langle31\rangle\langle q1^22^33^2|H\rangle +3\langle1q\rangle\langle 1^22^33^3|H\rangle)f_{ab}^ef_{cde}
\end{align*}
and momentum conservation implies $[12]\langle14\rangle+[32]\langle34\rangle=0\implies [12]/\langle43\rangle=-[32]/\langle41\rangle$. So, with the Jacobi identity ($f_{ad}{}^{e} f_{bc}{}^{d} + f_{bd}{}^{e} f_{ca}{}^{d} + f_{cd}{}^{e} f_{ab}{}^{d} = 0$, so that $f_{ad}{}^ef_{bce}=-f_{ab}{}^ef_{cde}-f_{ac}{}^ef_{dbe}$), we see that 
\begin{align*}
    s+t = 3\frac{[12]\langle31\rangle}{\langle4q\rangle\langle43\rangle}\langle q1^22^33^2|H\rangle f_{ab}^ef_{cde}+3\frac{[12]\langle1q\rangle}{\langle4q\rangle\langle43\rangle}\langle 1^22^33^3|H\rangle f_{ac}^ef_{bde}.
\end{align*}
We can then Schouten the u-channel term to get
\begin{align*}
    3\frac{\langle q2\rangle[31]}{\langle4q\rangle\langle42\rangle}\langle 1^32^23^3|H\rangle f_{ac}^ef_{bde} & = 3\frac{\langle q2\rangle\lambda_1[31]}{\langle4q\rangle\langle42\rangle}\langle 1^22^23^3|H\rangle f_{ac}^ef_{bde}\\
    & = 3\frac{[31]}{\langle4q\rangle\langle42\rangle}(\langle12\rangle\langle q1^22^23^3|H\rangle + 3\langle q1\rangle\langle 1^22^33^3|H\rangle )f_{ac}^ef_{bde}.
\end{align*}
Momentum conservation gives $[12]\langle42\rangle+[13]\langle43\rangle=0\implies [12]/\langle43\rangle=[31]/\langle42\rangle$ which means that
\begin{align*}
    s+t+u = 3\frac{[12]\langle31\rangle}{\langle4q\rangle\langle43\rangle}\langle q1^22^33^2|H\rangle f_{ab}^ef_{cde}+3\frac{[31]\langle12\rangle}{\langle4q\rangle\langle42\rangle}\langle q1^22^23^3 |H\rangle f_{ac}^ef_{bde}.
\end{align*}
We can finally apply momentum conservation $[12]\langle13\rangle+[42]\langle43\rangle=[13]\langle12\rangle+[43]\langle42\rangle=0$ to see that the above is a contact term
\begin{align*}
    s+t+u = 3\frac{[42]}{\langle4q\rangle}\langle q1^22^33^2|H\rangle f_{ab}^ef_{cde}+3\frac{[43]}{\langle4q\rangle}\langle q1^22^23^3|H\rangle f_{ac}^ef_{bde}.
\end{align*}
In what follows, I will show that this is exactly cancelled by the contact term contributions.

\subsection*{The contact terms}

\begin{figure}
    \centering
    \includegraphics[width=0.55\linewidth]{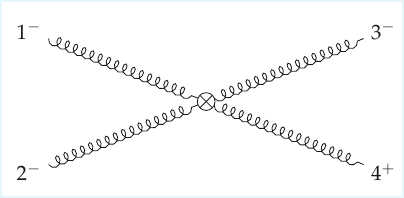}
    \caption{The 4-point contact term contributing to the $---+$ amplitude.}
    \label{fig:contact_term}
\end{figure}
The contact term Feynman diagrams included in this amplitude are presented in Fig. \ref{fig:contact_term}. The vertex generating these terms is of the form 
\begin{align*}
    \Tr(A_\gamma^{\dot{\alpha}}B_{\delta\rho}d_{\sigma\dot{\alpha}}B_{\zeta\eta}H^{\alpha\beta\gamma\delta\rho\sigma\zeta\eta}B_{\alpha\beta}).
\end{align*}
Computing the contact term gives
\begin{align*}
    \frac{[42]\lambda_2-[41]\lambda_1}{\langle4q\rangle}&\langle q1^22^23^2 |H\rangle f_{ab}^ef_{cde}+\frac{[43]\lambda_3-[41]\lambda_1}{\langle4q\rangle}\langle q1^22^23^2|H\rangle f_{ac}^ef_{bde}\\
    &\quad +\frac{[43]\lambda_3-[42]\lambda_2}{\langle 4q\rangle}\langle q1^22^23^2|H\rangle f_{ad}^ef_{bce}\\
    & = \frac{2[42]\lambda_2-[41]\lambda_1-[43]\lambda_3}{\langle4q\rangle}\langle q1^22^23^2|H\rangle f_{ab}^ef_{cde}\\
    &\quad +\frac{2[43]\lambda_3-[41]\lambda_1-[42]\lambda_2}{\langle4q\rangle}\langle q1^22^23^2|H\rangle f_{ac}^ef_{bde}\\
    & = 3\frac{[42]}{\langle4q\rangle}\langle q1^22^33^2|H\rangle f_{ab}^ef_{cde}+3\frac{[43]}{\langle4q\rangle}\langle q1^22^23^3|H\rangle f_{ac}^ef_{bde}.
\end{align*}
Which cancels exactly what we found from the channel diagrams. Therefore, the $---+$ amplitude vanishes as implied by integrability.

\subsection{Form factors}\label{subsec:form_factor}

While the vanishing of the $---+$ amplitude in the last subsection was a concrete check of integrability, we can do much better by appealing to a more abstract method leveraging the dual celestial chiral algebra derived in section \ref{sec:celestial}. In all form factor calculations, I only consider colour-ordered terms (terms whose colour factors contract with $\Tr(a_1...a_n)$) to ease the load since we only need to show that the form factors vanish when turning on momentum conservation.

The idea is to compute 3-minus, $n$-plus amplitudes by computing form factors in the presence of the operator $B_{\alpha\beta}D_\gamma^{\dot{\alpha}}B_{\delta\rho}D_{\sigma\dot{\alpha}}B_{\zeta\eta}H^{\alpha\beta\gamma\delta\rho\sigma\zeta\eta}$ by choosing the conformal block
\begin{align*}
    \langle B_{\alpha\beta}D_\gamma^{\dot{\alpha}}B_{\delta\rho}D_{\sigma\dot{\alpha}}B_{\zeta\eta}H^{\alpha\beta\gamma\delta\rho\sigma\zeta\eta}&|\tilde{J}(\tilde{\lambda}_1;z_1)\tilde{J}(\tilde{\lambda}_2;z_2)\tilde{J}(\tilde{\lambda}_3;z_3)\rangle\\
    &= [23]\langle11222333|H\rangle + [12]\langle11122233|H\rangle + [31]\langle11122333|H\rangle.
\end{align*}
This was computed by contracting the external legs of the operator $B_{\alpha\beta}D_\gamma^{\dot{\alpha}}B_{\delta\rho}D_{\sigma\dot{\alpha}}B_{\zeta\eta}H^{\alpha\beta\gamma\delta\rho\sigma\zeta\eta}$ by the negative helicity states labeled by the $\tilde{J}$ insertions. We can then insert an arbitrary number of positive helicity insertions and use the chiral algebra to reduce the form factor to the conformal block above, accompanied by poles arising from the chiral algebra. For an in-depth study of this computational method, see \cite{Costello:2022jpg}.

\subsection*{4-point form factor}

Let me start by performing a simple example using the chiral algebra correlators to compute the $---+$ amplitude.

Using the OPEs defined in subsection \ref{subsec:genus3_OPE} and only keeping colour ordered terms,
\begin{align*}
    \langle B_{\alpha\beta}D_\gamma^{\dot{\alpha}}&B_{\delta\rho}D_{\sigma\dot{\alpha}}B_{\zeta\eta}H^{\alpha\beta\gamma\delta\rho\sigma\zeta\eta}|\tilde{J}(\tilde{\lambda}_1;z_1)\tilde{J}(\tilde{\lambda}_2;z_2)\tilde{J}(\tilde{\lambda}_3;z_3)J(\tilde{\lambda}_4;z_4)\rangle\\
    & = \frac{1}{\langle34\rangle}\langle B_{\alpha\beta}D_\gamma^{\dot{\alpha}}B_{\delta\rho}D_{\sigma\dot{\alpha}}B_{\zeta\eta}H^{\alpha\beta\gamma\delta\rho\sigma\zeta\eta}|\tilde{J}(\tilde{\lambda}_1;z_1)\tilde{J}(\tilde{\lambda}_2;z_2)\tilde{J}(\tilde{\lambda}_3+\tilde{\lambda}_4;z_3)\rangle \\
    &\quad + \frac{1}{\langle41\rangle}\langle B_{\alpha\beta}D_\gamma^{\dot{\alpha}}B_{\delta\rho}D_{\sigma\dot{\alpha}}B_{\zeta\eta}H^{\alpha\beta\gamma\delta\rho\sigma\zeta\eta}|\tilde{J}(\tilde{\lambda}_1+\tilde{\lambda}_4;z_1)\tilde{J}(\tilde{\lambda}_2;z_2)\tilde{J}(\tilde{\lambda}_3;z_3)\rangle\\
    & = \frac{[12]\langle 1^32^33^2|H\rangle+[23]\langle 1^22^33^3|H\rangle+[24]\langle 1^22^33^3|H\rangle+[31]\langle 1^32^23^3|H\rangle+[41]\langle 1^32^23^3|H\rangle}{\langle34\rangle} \\ 
    &\quad + \frac{[12]\langle1^32^33^2|H\rangle+[42]\langle1^32^33^2|H\rangle+[23]\langle1^22^33^3|H\rangle+[31]\langle1^32^23^3|H\rangle+[34]\langle1^32^23^3|H\rangle}{\langle41\rangle} \\
    & = \frac{[24]\langle 1^22^33^3|H\rangle+[41]\langle 1^32^23^3|H\rangle}{\langle34\rangle} + \frac{[42]\langle1^32^33^2|H\rangle+[34]\langle1^32^23^3|H\rangle}{\langle41\rangle} \\
    &\quad -\frac{[23]\langle13\rangle}{\langle34\rangle\langle41\rangle}\langle1^22^33^3|H\rangle-\frac{[12]\langle13\rangle}{\langle34\rangle\langle41\rangle}\langle1^32^33^2|H\rangle-\frac{[31]\langle13\rangle}{\langle34\rangle\langle41\rangle}\langle1^32^23^3|H\rangle.
\end{align*}
This can be easily massaged into the expression derived for the form factor using Feynman diagrams in appendix \ref{app:4pt_ff}. It was already shown in subsection \ref{subsec:4pt_amp} that the 4-point amplitude vanishes on the support of momentum conservation, but it is also  easy to see that the above expression vanishes as well. One simply needs to apply the momentum conservation relations 
\begin{align*}
    [23]\langle13\rangle+[24]\langle14\rangle = 0, \quad [12]\langle13\rangle+[42]&\langle43\rangle=0,\quad -[31]\langle13\rangle = [24]\langle24\rangle,\\
    \text{and}\quad [41]\langle41\rangle+[34]\langle34\rangle &+ [24]\langle24\rangle=0.
\end{align*}

\subsection*{The $n$-point 3 minus form factor}

To find the general formula of the trace ordered $n$-point, 3-minus form factor, I will break the computation up into pieces by adopting the notation $J[1](\tilde{\lambda}_{\dot{\alpha}},z) = J[0,1](z)\tilde{\lambda}_{\dot{0}}+J[1,0](z)\tilde{\lambda}_{\dot{1}}, J(z)=J[0,0](z)$ and similarly for negative helicity states. The individual contributions are computed in appendix \ref{app:form_factor}. The full form factor is then given by summing the seven expressions there over $i = 4,...,n$ and $j>i$. We thus have the formula
\begin{align}\label{eq:gen_formfactor}
    \nonumber\left(\langle31\rangle[12]+\sum_{i=4}^n\langle 3i\rangle[i2]\right)&\frac{\langle1^32^33^2|H\rangle}{\langle34\rangle...\langle n1\rangle}+\left(\langle31\rangle[23]+\sum_{i=4}^n\langle i1\rangle[2i]\right)\frac{\langle1^22^33^3|H\rangle}{\langle34\rangle...\langle n1\rangle}\\
    &+\left(\langle31\rangle[31]+\sum_{i=4}^n\left(\langle i1\rangle[i1]+\langle3i\rangle[3i]+\sum_{j>i}^n\langle ij\rangle[ij]\right)\right)\frac{\langle1^32^23^3|H\rangle}{\langle34\rangle...\langle n1\rangle}.
\end{align}
Let me give a more formal presentation of the above formula with
\begin{proposition}\label{prop:genus3_npt_formfactor}
    The colour-ordered $n$-point, $3$-minus form factor in the theory is given by equation \ref{eq:gen_formfactor}.
\end{proposition}
\begin{proof}
    The proof is by induction. The $n = 4$ case produces the formula
    \begin{align*}
        \left(\langle31\rangle[12]+\langle 34\rangle[42]\right)\frac{\langle1^32^33^2|H\rangle}{\langle34\rangle...\langle n1\rangle}+(\langle31\rangle[23]+&\langle 41\rangle[24])\frac{\langle1^22^33^3|H\rangle}{\langle34\rangle...\langle n1\rangle}\\
    &+(\langle31\rangle[31]+\langle 41\rangle[41]+\langle34\rangle[34])\frac{\langle1^32^23^3|H\rangle}{\langle34\rangle...\langle n1\rangle}
    \end{align*}
    which matches the calculations of the 4-point form factor done in subsection \ref{subsec:form_factor}. 

    Next, assume the formula holds at $n$ points and call the formula I proposed for the $n$-point form factor $F_n$. Then the $n+1$-point formula has contributions coming from
    \begin{align*}
        \langle B_{\alpha\beta}D_\gamma^{\dot{\alpha}}B_{\delta\rho}D_{\sigma\dot{\alpha}}B_{\zeta\eta}H^{\alpha\beta\gamma\delta\rho\sigma\zeta\eta}|...J^{a_{n+1}}(z_{n+1})\rangle& = \frac{F_n\langle n1\rangle}{\langle n,n+1\rangle\langle n+1,1\rangle}\\
        \langle B_{\alpha\beta}D_\gamma^{\dot{\alpha}}B_{\delta\rho}D_{\sigma\dot{\alpha}}B_{\zeta\eta}H^{\alpha\beta\gamma\delta\rho\sigma\zeta\eta}|...J^{a_{i}}[1](z_{i})...J^{a_{n+1}}[1](z_{n+1})\rangle
        & = \frac{[i,n+1]\langle i,n+1\rangle\langle1^32^23^3|H\rangle}{\langle34\rangle...\langle n+1,1\rangle}\\
        \langle B_{\alpha\beta}D_\gamma^{\dot{\alpha}}B_{\delta\rho}D_{\sigma\dot{\alpha}}B_{\zeta\eta}H^{\alpha\beta\gamma\delta\rho\sigma\zeta\eta}|\tilde{J}^{a_{1}}[1](z_{1})...J^{a_{n+1}}[1](z_{n+1})\rangle&= \frac{[n+1,1]\langle1^32^23^3|H\rangle}{\langle34\rangle...\langle n,n+1\rangle}\\
        \langle B_{\alpha\beta}D_\gamma^{\dot{\alpha}}B_{\delta\rho}D_{\sigma\dot{\alpha}}B_{\zeta\eta}H^{\alpha\beta\gamma\delta\rho\sigma\zeta\eta}|\tilde{J}^{a_{1}}(z_{1})\tilde{J}^{a_{2}}[1](z_{2})...J^{a_{n+1}}[1](z_{n+1})\rangle&=\frac{[2,n+1]\langle1^22^33^3|H\rangle}{\langle34\rangle...\langle n,n+1\rangle}\\
        &\quad +\frac{[n+1,2]\langle3,n+1\rangle\langle1^32^33^2|H\rangle}{\langle34\rangle...\langle n,n+1\rangle\langle n+1,1\rangle}\\
        \langle B_{\alpha\beta}D_\gamma^{\dot{\alpha}}B_{\delta\rho}D_{\sigma\dot{\alpha}}B_{\zeta\eta}H^{\alpha\beta\gamma\delta\rho\sigma\zeta\eta}|\tilde{J}^{a_{1}}(z_{1})\tilde{J}^{a_{2}}(z_{2})\tilde{J}^{a_{3}}[1](z_{3})...J^{a_{n+1}}[1](z_{n+1})\rangle&=\frac{[3,n+1]\langle3,n+1\rangle\langle1^32^23^3|H\rangle}{\langle34\rangle...\langle n,n+1\rangle\langle n+1,1\rangle}.
    \end{align*}
    Summing up the above gives
    \begin{align*}
        \Biggl(\langle31\rangle[12]&+\langle 3,n+1\rangle[n+1,2]+\sum_{i=4}^n\langle 3i\rangle[i2]\Biggr)\frac{\langle1^32^33^2|H\rangle}{\langle34\rangle...\langle n+1,1\rangle}\\
        &+\left(\langle31\rangle[23]+\langle n+1,1\rangle[2,n+1]+\sum_{i=4}^n\langle i1\rangle[2i]\right)\frac{\langle1^22^33^3|H\rangle}{\langle34\rangle...\langle n+1,1\rangle}\\
    &+\Biggl(\langle31\rangle[31]+\langle n+1,1\rangle[n+1,1]+\langle 3,n+1\rangle[3,n+1]\\
    & +\sum_{i=4}^n\left(\langle i1\rangle[i1]+\langle3i\rangle[3i]+\langle i,n+1\rangle[i,n+1]+\sum_{j>i}^n\langle ij\rangle[ij]\right)\Biggl)\frac{\langle1^32^23^3|H\rangle}{\langle34\rangle...\langle n+1,1\rangle}
    \end{align*}
    which equals $F_{n+1}$ by rewriting it as a sum up to $n+1$. This completes the proof.
\end{proof}
I conclude this subsection on form factors with another check of integrability in
\begin{proposition}\label{prop:genus3_integrability}
    The $n$-point, $3$-minus form factor in the theory vanishes on the support of momentum conservation.
\end{proposition}
\begin{proof}
    It suffices to check this for the colour-ordered part. Momentum conservation gives us the following relation on the kinematic variables
    \begin{align*}
        \sum_{i=1}^n| i\rangle[i| = 0.
    \end{align*}
    We need to show that the coefficients of each of the 3 different $H$ contractions vanish. These coefficients are
    \begin{align*}
        \langle31\rangle[12]&+\sum_{i=4}^n\langle 3i\rangle[i2]\\ \langle31\rangle[23]&+\sum_{i=4}^n\langle i1\rangle[2i]\\\langle31\rangle[31]&+\sum_{i=4}^n\left(\langle i1\rangle[i1]+\langle3i\rangle[3i]+\sum_{j>i}^n\langle ij\rangle[ij]\right).
    \end{align*}
    The first expression vanishes on shell since
    \begin{align*}
        0 = \langle3|\left(\sum_{i=1}^n| i\rangle[i|\right)|2] = \sum_{i\ne2,3}^n\langle3 i\rangle[i2].
    \end{align*}
    Similarly, the second expression vanishes by contracting the sum of momenta with $|1\rangle$ and $[2|$. Finally, for the third expression, contracting the sum of momenta with $|1\rangle$ and $|1]$ gives
    \begin{align*}
        \sum_{i=4}^n\langle i1\rangle[i1] = -\langle 21\rangle[21]-\langle 31\rangle[31]
    \end{align*}
    while contracting with $\langle3|$ and $[3|$ gives
    \begin{align*}
        \sum_{i=4}^n\langle 3i\rangle[3i] = -\langle31\rangle[31]-\langle32\rangle[32]
    \end{align*}
    So, the first three terms add up in the following manner
    \begin{align*}
        \langle31\rangle[31]+\sum_{i = 4}^n\left(\langle i1\rangle[i1]+\langle3i\rangle[3i]\right) = -\langle 21\rangle[21]-\langle31\rangle[31]-\langle32\rangle[32] = -(p_1+p_2+p_3)^2/2.
    \end{align*}
    Momentum conservation then gives us
    \begin{align*}
        (p_1+p_2+p_3)^2/2 = \frac{1}{2}\left(\sum_{i = 4}^np_i\right)^2 = \sum_{i = 4}^n\sum_{j>i}^n\langle ij\rangle[ij].
    \end{align*}
    So, the last coefficient also vanishes.
\end{proof}

\subsection{Finding the full theory}\label{subsec:full_theory}

So far, I have built an action whose chiral algebra lives on the hyperelliptic curve defined by the polynomial $H(z)$ and exhibits no 3-minus, $n$-plus scattering. I will now bootstrap the form of all higher-order interactions. This provides a rare example of a higher-genus integrable model whose higher-order coefficients can be determined by imposing vanishing tree-level scattering.

\subsection*{Fixing all terms linear in $H$}

The section on form factors showed that there cannot exist any $3$-minus vertices with linear dependence on $H$ without spoiling integrability. I will now extend this result by fixing \textit{any} linear dependence on $H$ by using 4d symmetry and dimensional arguments.

A useful property of the SDYM lagrangian is that it has dimension zero. In other words, it is a 4d CFT. To constrain the 4d action further, I leverage this fact in
\begin{lemma}\label{lem:lin_H}
    The only integrable, gauge-invariant, dimensionless vertex that is compatible with the symmetries and linear in $H$ is
    \begin{align*}
        \Tr(B_{\alpha\beta}D_\gamma^{\dot{\alpha}}B_{\delta\rho}D_{\sigma\dot{\alpha}}B_{\zeta\eta}H^{\alpha\beta\gamma\delta\rho\sigma\zeta\eta}).
    \end{align*}
\end{lemma}
\begin{proof}
    Gauge invariance is manifest in the fact that the derivatives are covariant. Integrability was proven in subsection \ref{subsec:form_factor} when the vertex was added to the full theory with the correct coefficient.

    The fields of our theory have dimensions $|B| = 2$ and $|A| = 1$ so that the $BF$ term has dimension 4 to cancel against the dimension of the $\R^4$ integration measure. Since $|BDBDB|=8$, we see that $H$ must have dimension $-4$. So, the only operators that can accompany $H$ have dimension $8$. Let's organize all such operators. On shell, $F_{\alpha\beta} = 0$, so the only gauge invariant quantities we can use are $B_{\alpha\beta}$, $D_\alpha^{\dot{\alpha}}$, and $F_{\dot{\alpha}\dot{\beta}}$. Moreover, the only operators that can pair with $H$ must have $8$ free undotted indices to pair with the indices of $H$. The only possibilities are thus
    \begin{align*}
        B_{\alpha\beta}D_\gamma^{\dot{\alpha}}B_{\delta\rho}D_{\sigma\dot{\alpha}}B_{\zeta\eta} \quad\text{and}\quad B_{\alpha\beta}B_{\gamma\kappa}B_{\delta\rho}B_{\sigma\nu}B_{\zeta\eta}.
    \end{align*}
    The first vertex is the one we have included in the theory, and the second can be eliminated by noticing that it violates integrability. It introduces a 5 minus contact term with a factor of $H$ that cannot be cancelled against any Feynman diagrams built from the other vertices of the theory. Finally, notice that any vertex containing a factor of $F_{\dot{\alpha}\dot{\beta}}$ and $8$ free undotted indices will have dimension $>8$.
\end{proof}

The careful reader might wonder whether terms can be added to the action that do not depend on $H$. This is impossible:
\begin{lemma}\label{lem:BF_only}
    The only term that does not contain a factor of $H$ is the SDYM term.
\end{lemma}
\begin{proof}
    As discussed in the proof of the previous proposition, any term in the action must have dimension $4$. The only such gauge invariant quantities whose indices contract (hence being consistent with the symmetries of the theory) are
    \begin{align*}
        B_{\alpha\beta}B^{\alpha\beta}\quad \text{and}\quad F_{\dot{\alpha}\dot{\beta}}F^{\dot{\alpha}\dot{\beta}}.
    \end{align*}
    The first term is known to deform SDYM to full Yang-Mills, which is not integrable. The second term is topological and hence does not appear in perturbation theory.
\end{proof}

\subsection*{Constraining higher order vertices}

So far, I have only used spacetime arguments to constrain the vertices of the theory. It turns out that the form of all terms in the Lagrangian can be fixed if we appeal to twistor arguments! First, recall that the theory started life as the following twistor action
\begin{align*}
    \int_\PT \Omega\wedge\Tr \left(b\dbar a+\frac{1}{3}(ba^2+b^3H(z))\right)
\end{align*}
where $a\in\Omega^{0,1}(\PT,\mathcal{O})$ and $b\in\Omega^{0,1}(\PT,\mathcal{O}(-4))$. On spacetime, as has been discussed already, $a$ descends to a positive helicity gluon and $b$ to a negative helicity gluon. Building off of the work done in lemmas \ref{lem:lin_H} and \ref{lem:BF_only}, the only terms that must be constrained include powers of $H^n$ with $n>1$. 
\begin{figure}
    \centering
    \includegraphics[width=0.65\linewidth]{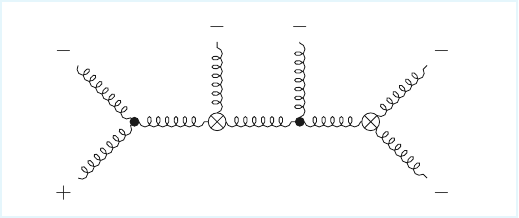}
    \caption{An example of a twistor space Feynman diagram with two powers of $H(z)$. I denote the $b^3H(z)$ by $\otimes$ and the $baa$ vertex by $\bullet$ in analogy with my 4d conventions.}
    \label{fig:twistor_H2_diagram}
\end{figure}

Before explaining the general case, let me explain the pedagogical example of determining the vertex that accompanies $H^2$. The only diagrams with two factors of $H$ that can be built on twistor space are ones with $5$ minus and an arbitrary number of positive helicity states. Examples of this are depicted in Fig. \ref{fig:twistor_H2_diagram}. The only such (gauge invariant) operators with dimension $12$ (which is needed as $|H^2|=-8$)  are
\begin{align*}
        B_{\omega\psi}B_{\phi\chi}B_{\alpha\beta}D_\gamma^{\dot{\alpha}}B_{\delta\rho}D_{\sigma\dot{\alpha}}B_{\zeta\eta} \quad\text{and}\quad B_{\omega\psi}B_{\phi\chi}B_{\alpha\beta}B_{\delta\rho}B_{\zeta\eta} F_{\dot{\alpha}\dot{\beta}}.
    \end{align*}
    But the second vertex above cannot be included, as there are free undotted indices that cannot contract against anything! Therefore, the only possible vertex of order $H^2$ is the first choice above, where two of the undotted indices on each factor of $H$ must contract into each other so that 12 free indices can contract with the operator. 

More generally, all possible vertices in the theory are described in
\begin{lemma}\label{lem:full_theory}
    The full theory is of the form
    \begin{align*}
        \int_{\R^4}\Tr(BF+BDBDBH+\text{terms of form $B^{2n-1}DBDBH^n$}).
    \end{align*}
    Where $n>1$, the form in which the indices contract on the higher order terms is unknown, and their coefficients are unknown.
\end{lemma}

\begin{proof}
    By lemmas \ref{lem:lin_H} and \ref{lem:BF_only}, we need to find the spacetime vertex accompanying $H^n$ for $n>1$. The twistor action had vertices of the form $ba^2$ and $b^3H(z)$. The only diagrams one can build on twistor space from these that contain a power of $H^n$ are $2n+1$ minus, $m$ plus diagrams for arbitrary $m$. Therefore, the vertex accompanying $H^n$ must contain $2n+1$ negative helicity gluons. The only spacetime operator with these properties and the correct dimension is
    \begin{align*}
        B^{2(n-1)}BDBDB
    \end{align*}
    which has $4n+4$ free undotted indices. Notice that on twistor space, we can always draw diagrams of order $H^n$ with $2n+1$ negative helicity gluons and arbitrarily many positive helicity gluons but none of these diagrams can correspond to a spacetime operator as introducing powers of $F_{\dot{\alpha}\dot{\beta}}$ would increase the dimension. If one were to replace the two derivatives by a factor of $F_{\dot{\alpha}\dot{\beta}}$ so that the operator was of the form $B^{2(n-1)}BBBF_{\dot{\alpha}\dot{\beta}}$, then this would have the correct dimension but the dotted indices could never contract meaning that the vertex violates the symmetries of the theory.
\end{proof}

\begin{remark}
At this point, I want to point out the interesting fact that we already have enough to write down the exact Lagrangian in the degenerate case $H(z) = z^8$. One can visualize this geometry as two Riemann spheres joined at the origin, where there is an order-4 singularity. In this case, the tensor defining the polynomial, $H(z)$, factorizes as
\begin{align*}
    H^{\alpha_1,...,\alpha_8} = \lambda_0^{\alpha_1}...\lambda_0^{\alpha_8}\quad\text{where}\quad \lambda_0 = \begin{pmatrix}
        1\\0
    \end{pmatrix}.
\end{align*}
This means that if $H$ contracts with itself on any index, then the term vanishes. It is quite easy to see by counting indices that there is at least one contraction between factors of $H$ in any vertex of the form $BDBDBH(HB^2)^n$ for $n>0$. In other words, the integrable model truncates at 
\begin{align*}
    \int_{\R^4}\Tr(BF+BdBdB\lambda_0^8).
\end{align*}
This is a new integrable model whose spectral parameter lives on the highly degenerate curve described above.

There is a geometric understanding of this example by starting on the double cover $\PT_\Sigma$ and interpreting this degeneration as transforming this genus 3 fibration into two copies of twistor space joined at the origin. One then considers two copies of holomorphic Chern-Simons theory on $\PT\sqcup \PT$ with appropriate boundary conditions on the gauge fields at the origin. One should then be able to recover the model above using standard gauge fixing methods on twistor space with this setup. I believe that this was done in \cite{Cole_2024} where they derive the LMP action from twistor space\footnote{I would like to thank Lewis T. Cole and Ryan A. Cullinan for explaining their work to me and pointing out a mistake in an early calculation I did relating to these degenerate models.}. However, their methods involved gauge fixing on twistor space so the action they arrive at looks quite different. It would be interesting to see how to relate these presentations.

Another interesting setup that was studied in \cite{Cole_2024} is when the curve degenerates into one of the form $H(z) = z^4(z-z_0)^4$, which corresponds to merging 4 branch points into one. This geometry looks like two Riemann spheres glued at the points $0,z_0$ with second-order singularities at each point. In this case, I unfortunately do not know whether the theory truncates to some finite order in my description, like in the $H(z) = z^8$ case. The authors of \cite{Cole_2024} took a different approach to compactifying holomorphic Chern-Simons theory on this degenerate curve to find a 4d gauged WZW model. It would be interesting to relate their description to the deformation of SDYM I have found in this paper when the tensor $H$ degenerates.
\end{remark}

The main result of the section, theorem \ref{thm:4d_action}, is then a simple consequence of lemma \ref{lem:full_theory}. 
\begin{proof}[Proof of theorem \ref{thm:4d_action}]
    Lemma \ref{lem:full_theory} guarantees that the theory is of the form
    \begin{align*}
        \int_{\R^4}\Tr(BF+BDBDBH+\text{terms of form $B^{2n-1}DBDBH^n$}).
    \end{align*}
    Let's constrain the $n = 2$ vertex using integrability. The theory truncated to the $n=1$ vertex generates a 5 minus amplitude containing two copies of $H$. Twistor integrability guarantees that this amplitude is a contact term. Since contact terms are in bijection with local operators, this uniquely fixes a 5 minus vertex at second order in $H$. This vertex will render the theory scattering free at second order in $H$.

    Next, suppose that the vertices are determined up to $n$th order in $H$. Then the theory has a $2n+3$ minus amplitude that must be a contact term. There thus exists a unique local operator including $2n+3$ $B$ insertions that restores integrability. This fixes the theory at $n+1$th order in $H$.
\end{proof}

\section{Bootstrapping vertices part II: The genus 1 deformation}\label{sec:genus1_bootstrap}

Recall the genus 1 deformation of SDYM coupled to a complex scalar
\begin{align*}
\int_{\mathbb{R}^4}\Tr(BF(A)+\tilde{\Phi}D*D\Phi+B_{\alpha\beta}D_{\dot{\alpha}\delta}\Phi D^{\dot{\alpha}}_\gamma\Phi H^{\alpha\beta\delta\gamma})+\text{higher order terms} 
\end{align*}
derived in section \ref{sec:model_build}. In this section, I will provide a similar analysis to the previous section and prove 

\begin{theorem}\label{thm:genus1_full_theory}
    The 4d theory descending from the genus 1 deformed twistor theory in equation \ref{eq:genus1_twistor_action} is determined to all orders as
    \begin{align*}
        \int_{\mathbb{R}^4}\Tr(BF(A)+\tilde{\Phi}D*D\Phi+BD\Phi D\Phi H+\text{terms of form $B\Phi^{2n-2}D\Phi D\Phi H^n$}).
    \end{align*}
    Where the form of the index contractions and coefficients of the higher order terms are fixed by integrability.
\end{theorem}

The analysis I did in section \ref{sec:bootstrap} for the genus 3 deformation was presented slowly and explicitly. Since bootstrapping the vertices in the genus 1 setting is so similar, this section will be more streamlined.

\subsection{Verifying basic integrability} 

In a similar spirit to the story told in section \ref{sec:bootstrap}, I will check that the form factor
\begin{align*}
    \langle B_{\alpha\beta}D_{\dot{\alpha}\delta}\Phi D^{\dot{\alpha}}_\gamma\Phi H^{\alpha\beta\delta\gamma}|\tilde{J}^{a_1}(\tilde{\lambda}_1;z_1)\mathcal{O}^{a_2}(\tilde{\lambda}_2;z_2)\mathcal{O}^{a_3}(\tilde{\lambda}_3;z_3)J^{a_4}(\tilde{\lambda}_4;z_4)...J^{a_n}(\tilde{\lambda}_n;z_n)\rangle
\end{align*}
vanishes on the support of momentum conservation. Note that as I did in the previous section, all form factor computations will only be concerned with colour-ordered terms.

\subsection*{The 4-point form factor}

Let me start with the base case, that is, the 4-point form factor. The first ingredient we need is the conformal block
\begin{align*}
    \langle B_{\alpha\beta}D_{\dot{\alpha}\delta}\Phi D^{\dot{\alpha}}_\gamma\Phi H^{\alpha\beta\delta\gamma}|\tilde{J}(\tilde{\lambda}_1;z_1)&\mathcal{O}(\tilde{\lambda}_2;z_2)\mathcal{O}(\tilde{\lambda}_3;z_3)\rangle\\
    &= [12]\langle1^32|H\rangle+[23]\langle1^223|H\rangle+[31]\langle1^33|H\rangle.
\end{align*}
We can then insert a positive helicity gluon and compute OPEs to find that
\begin{align*}
    \langle B_{\alpha\beta}D_{\dot{\alpha}\delta}\Phi D^{\dot{\alpha}}_\gamma\Phi H^{\alpha\beta\delta\gamma}|\tilde{J}[1](z_1)\mathcal{O}[1](z_2)\mathcal{O}(z_3)J(z_4)\rangle & = \frac{\langle31\rangle[12]\langle1^32|H\rangle}{\langle34\rangle\langle41\rangle}\\
    \langle B_{\alpha\beta}D_{\dot{\alpha}\delta}\Phi D^{\dot{\alpha}}_\gamma\Phi H^{\alpha\beta\delta\gamma}|\tilde{J}[1](z_1)\mathcal{O}(z_2)\mathcal{O}[1](z_3)J(z_4)\rangle & = \frac{\langle31\rangle[31]\langle1^33|H\rangle}{\langle34\rangle\langle41\rangle}\\
    \langle B_{\alpha\beta}D_{\dot{\alpha}\delta}\Phi D^{\dot{\alpha}}_\gamma\Phi H^{\alpha\beta\delta\gamma}|\tilde{J}(z_1)\mathcal{O}[1](z_2)\mathcal{O}[1](z_3)J(z_4)\rangle & = \frac{\langle31\rangle[23]\langle1^223|H\rangle}{\langle34\rangle\langle41\rangle}\\
    \langle B_{\alpha\beta}D_{\dot{\alpha}\delta}\Phi D^{\dot{\alpha}}_\gamma\Phi H^{\alpha\beta\delta\gamma}|\tilde{J}[1](z_1)\mathcal{O}(z_2)\mathcal{O}(z_3)J[1](z_4)\rangle & = \frac{[41]\langle1^33|H\rangle}{\langle34\rangle}\\
    \langle B_{\alpha\beta}D_{\dot{\alpha}\delta}\Phi D^{\dot{\alpha}}_\gamma\Phi H^{\alpha\beta\delta\gamma}|\tilde{J}(z_1)\mathcal{O}[1](z_2)\mathcal{O}(z_3)J[1](z_4)\rangle & = \frac{[24]\langle1^223|H\rangle}{\langle34\rangle}+\frac{[42]\langle1^32|H\rangle}{\langle41\rangle}\\
    \langle B_{\alpha\beta}D_{\dot{\alpha}\delta}\Phi D^{\dot{\alpha}}_\gamma\Phi H^{\alpha\beta\delta\gamma}|\tilde{J}(z_1)\mathcal{O}(z_2)\mathcal{O}[1](z_3)J[1](z_4)\rangle & = \frac{[34]\langle1^33|H\rangle}{\langle41\rangle}.
\end{align*}
So the full form factor is
\begin{align*}
    \langle B_{\alpha\beta}&D_{\dot{\alpha}\delta}\Phi D^{\dot{\alpha}}_\gamma\Phi H^{\alpha\beta\delta\gamma}|\tilde{J}(\tilde{\lambda}_1;z_1)\mathcal{O}(\tilde{\lambda}_2;z_2)\mathcal{O}(\tilde{\lambda}_3;z_3)J(\tilde{\lambda}_4;z_4)\rangle\\
    &= (\langle31\rangle[12]+\langle34\rangle[42])\frac{\langle1^32|H\rangle}{\langle34\rangle\langle41\rangle}+(\langle31\rangle[31]+\langle41\rangle[41]+\langle34\rangle[34])\frac{\langle1^33|H\rangle}{\langle34\rangle\langle41\rangle}\\
    &\quad +(\langle31\rangle[23]+\langle41\rangle[24])\frac{\langle1^223|H\rangle}{\langle34\rangle\langle41\rangle}.
\end{align*}
This vanishes on the support of momentum conservation. To see this, notice that
\begin{align*}
    \langle31\rangle[12]+\langle34\rangle[42] & = \langle3|\sum_ip_i|2] = 0\\
    \langle31\rangle[23]+\langle41\rangle[24] &=[2|\sum_ip_i|1\rangle=0 \\
    \langle31\rangle[31]+\langle41\rangle[41]+\langle34\rangle[34] & = \langle31\rangle[31]+\langle41\rangle[41]+\langle21\rangle[21] = -\langle1|\sum_ip_i|1] = 0
\end{align*}
using the fact that $2\langle34\rangle[34]=(p_1+p_2)^2 = 2\langle21\rangle[21]$ in the last line above.

\subsection*{The $n$-point form factor}

The formula for the colour-ordered $n$-point form factor is
\begin{align}\label{eq:genus1_npt_formfac}
    \langle B_{\alpha\beta}&D_{\dot{\alpha}\delta}\Phi D^{\dot{\alpha}}_\gamma\Phi H^{\alpha\beta\delta\gamma}|\tilde{J}(\tilde{\lambda}_1;z_1)\mathcal{O}(\tilde{\lambda}_2;z_2)\mathcal{O}(\tilde{\lambda}_3;z_3)J(\tilde{\lambda}_4;z_4)...J(\tilde{\lambda}_n;z_n)\rangle\\\nonumber
    & = \left(\langle31\rangle[12]+\sum_{i=4}^n\langle3i\rangle[i2]\right)\frac{\langle1^32|H\rangle}{\langle34\rangle...\langle n-1,n\rangle\langle n1\rangle}+\left(\langle31\rangle[23]+\sum_{i=4}^n\langle i1\rangle[2i]\right)\frac{\langle1^223|H\rangle}{\langle34\rangle...\langle n-1,n\rangle\langle n1\rangle}\\\nonumber
    &\quad+\left(\langle31\rangle[31]+\sum_{i=4}^n\left(\langle i1\rangle[i1]+\langle3i\rangle[3i]+\sum_{j>i}^n\langle ij\rangle[ij]\right)\right)\frac{\langle1^33|H\rangle}{\langle34\rangle...\langle n-1,n\rangle\langle n1\rangle}
\end{align}
which I verify with
\begin{proposition}
    The formula for the $n$-point colour-ordered $-\Phi\Phi+...+$ form factor is given by equation \ref{eq:genus1_npt_formfac}.
\end{proposition}
\begin{proof}
    The proof is nearly identical to the proof of proposition \ref{prop:genus3_npt_formfactor}.
\end{proof}
Integrability is then verified with 
\begin{proposition}
    The $n$-point $-\Phi\Phi+...+$ form factor vanishes on the support of momentum conservation.
\end{proposition}
\begin{proof}
The coefficients that need to vanish are     
\begin{align*}
        \langle31\rangle[12]&+\sum_{i=4}^n\langle 3i\rangle[i2]\\ \langle31\rangle[23]&+\sum_{i=4}^n\langle i1\rangle[2i]\\\langle31\rangle[31]&+\sum_{i=4}^n\left(\langle i1\rangle[i1]+\langle3i\rangle[3i]+\sum_{j>i}^n\langle ij\rangle[ij]\right).
    \end{align*}
    These were shown to vanish on the support of momentum conservation in the proof of proposition \ref{prop:genus3_integrability}.
\end{proof}

\subsection{Bootstrapping the full theory}

In section \ref{sec:bootstrap}, I did a slow and careful analysis in both the 4d theory and on twistor space to bootstrap the full 4d theory for the genus 3 deformation. This gets a little more complicated in this setting due to the extra scalars that enter the story, but the overall theme of the analysis is the same. Since the arguments will be analogous to the ones done in the genus 3 case, I will bypass some of the 4d arguments by relying more on twistor arguments.

\subsection*{Fixing linear in $H$ terms} I will now consider the space of possible linear in $H$ diagrams one can draw on twistor space and the possible spacetime vertices these can descend to.

\begin{lemma}\label{lem:genus1_dim_vertices}
    The only gauge invariant, dimensionless vertices that are compatible with the symmetries and linear in $H$ are
    \begin{align*}
        B^3H\quad B^2(a_1\Phi+a_2\tilde{\Phi})(b_1\Phi+b_2\tilde{\Phi})H\quad DBDBH\quad BD\Phi D\Phi H\quad BD\tilde{\Phi}D\tilde{\Phi}H\quad BD\Phi D\tilde{\Phi}H
    \end{align*}
    modulo topological terms. Note that the indices are free to contract in any way possible, and $a_1,a_2,b_1,b_2$ are constants.
\end{lemma}
\begin{proof}
The gauge invariant quantities we have are $\Phi,\tilde{\Phi},D_{\alpha\dot{\alpha}},B_{\alpha\beta},F_{\dot{\alpha}\dot{\beta}}$ with dimensions
\begin{align*}
    |\Phi|=|\tilde{\Phi}|=|D_{\alpha\dot{\alpha}}|=1\quad |B_{\alpha\beta}|=|F_{\dot{\alpha}\dot{\beta}}|=2.
\end{align*}
Any operator accompanying $H$, must have dimension $6$ to keep the Lagrangian dimensionless since $H$ has dimension $-2$. To be consistent with the Lorentz symmetry, any of our candidate operators must have 4 free undotted indices to contract against $H$. 

Let me first eliminate any terms involving $F_{\dot{\alpha}\dot{\beta}}$. Notice first that at most one copy of $F_{\dot{\alpha}\dot{\beta}}$ can appear in a dimension 6 operator with 4 free undotted indices. The operator will then have to involve two covariant derivatives to contract against the two dotted indices. So, the candidate operator so far looks like $DDF$ where we need to insert something of dimension 2 that introduces the correct number of free undotted indices. There is only one such choice:
\begin{align*}
    \Tr(D_{\dot{\alpha}}^\alpha B^{\beta\gamma}D_{\dot{\beta}}^\delta F^{\dot{\alpha}\dot{\beta}}H_{\alpha\beta\gamma\delta})
\end{align*}
up to moving the derivatives around using integration by parts. This term is clearly exact!

Next, let's consider operators with factors of $B$. If 3 factors of $B$ are present, then no other objects can appear since it already has dimension 6. Operators of this nature take the form
\begin{align*}
    \Tr(B_{\alpha\beta} B_{\gamma\rho}B_\delta^\rho H^{\alpha\beta\gamma\delta}).
\end{align*}
If there are two factors of $B$, then there are many options that work:
\begin{align*}
    \Tr(B^2H\cdot (a_1\Phi+a_2\tilde{\Phi})(b_1\Phi+b_2\tilde{\Phi}))\quad \Tr(DBDBH)
\end{align*}
where the indices contract in any way possible and $a_1,a_2,b_1,b_2$ are constants.

If there is one factor of $B$, then there must be at least two factors of $D$ to get the correct number of undotted indices. Notice that we cannot have a term with three derivatives (breaks symmetry), and four would involve a term where the derivative hits $B$, which vanishes on shell. So the only possibilities with one factor of $B$ are 
\begin{align*}
    BD\Phi D\Phi H\quad BD\tilde{\Phi}D\tilde{\Phi}H\quad BD\Phi D\tilde{\Phi}H.
\end{align*}
Finally, notice that there are no terms that only involve derivatives and scalars, as they vanish on shell.
\end{proof}

We can throw away most of the candidate vertices in lemma \ref{lem:genus1_dim_vertices} by appealing to twistor arguments:

\begin{lemma}
    The only vertex that can appear in the 4d theory of the genus 1 deformation is
    \begin{align*}
        BD\Phi D\Phi H.
    \end{align*}
\end{lemma}
\begin{proof}
    Recall the twistor action that our 4d theory descended from
    \begin{align*}
        \int_{\mathbb{PT}}\Tr(b\dbar a +\tilde{b}\dbar\tilde{a} + baa + a\tilde{a}\tilde{b}+b\tilde{a}\tilde{a}H(z)).
    \end{align*}
    The corresponding 4d theory cannot source Feynman diagrams that do not have a twistorial counterpart. For example, the $B^3H$ candidate vertex of the previous lemma cannot appear as it would source a 3-minus Feynman graph that does not appear in twistor space. It is clear that on twistor space, the only 3-point Feynman diagrams with a linear factor of $H$ descends to a $-\Phi\Phi$ diagram. There are also no kinetic terms with a factor of $H$. This eliminates several possible vertices and we are left with the following candidate vertices
    \begin{align*}
        B^2(a_1\Phi+a_2\tilde{\Phi})(b_1\Phi+b_2\tilde{\Phi})H\quad BD\Phi D\Phi H
    \end{align*}
    by lemma \ref{lem:genus1_dim_vertices}. The second vertex above is the one we want. As for the 4-point vertex, notice that on twistor space, there are only three 4-point diagrams with a linear factor of $H$. They are shown in Fig. \ref{fig:genus1_twistor_feyn} ($\tilde{a}\tilde{a}\to ba$, $b\tilde{a}\to a\tilde{a}$, $\tilde{a}\tilde{a}\to \tilde{a}\tilde{b}$). Notice that none of them involve two external $b$ gluons ($B$ on spacetime), so we can deduce that the $B^2(a_1\Phi+a_2\tilde{\Phi})(b_1\Phi+b_2\tilde{\Phi})H$ vertex cannot appear in the 4d theory.
    \begin{figure}
        \centering
        \includegraphics[width=0.85\linewidth]{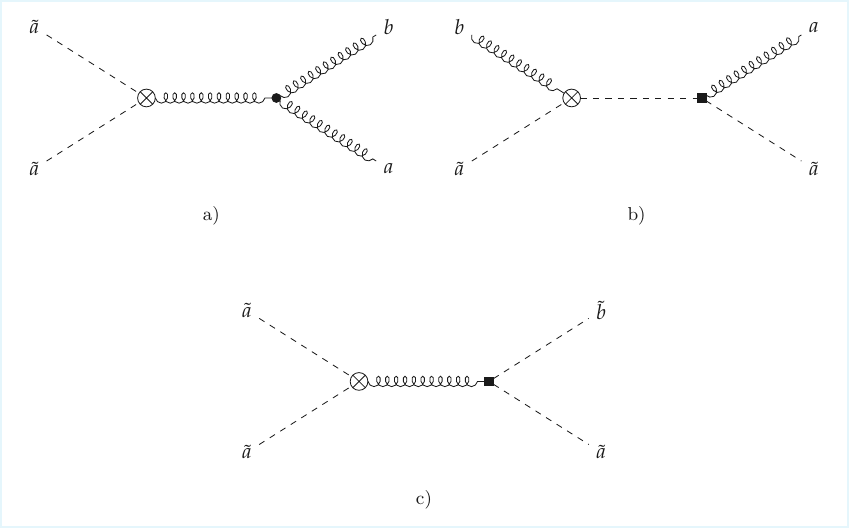}
        \caption{The possible 4-point Feynman diagrams one can draw on twistor space that have a single factor of $H(z)$. Note that the $baa$ and $a\tilde{a}\tilde{b}$ vertices are labelled by $\bullet$ and $\scalebox{0.6}{$\blacksquare$}$ respectively while the $b\tilde{a}\tilde{a}H(z)$ vertex is labelled by $\otimes$.}
        \label{fig:genus1_twistor_feyn}
    \end{figure}
\end{proof}
\begin{remark}
    In the proof of the previous lemma, I relied heavily on twistor arguments to eliminate candidate vertices in the theory. One could also take a 4d approach to eliminating them by noticing basic violations of integrability. For example, let's consider the 
    \begin{align*}
        \Tr(B^{\alpha\beta}B^\gamma_\rho B^{\delta\rho})H_{\alpha\beta\gamma\delta}.
    \end{align*}
    vertex. The 4 point $---+$ colour-stripped amplitude that this induces is given by the form factor
    \begin{align*}
        \langle \Tr(B^{\alpha\beta}B^\gamma_\rho B^{\delta\rho})H_{\alpha\beta\gamma\delta}|\tilde{J}(z_1)\tilde{J}(z_2)\tilde{J}(z_3)J(z_4)\rangle & = \frac{\langle12\rangle}{\langle13\rangle\langle23\rangle}\langle \Tr(B^{\alpha\beta}B^\gamma_\rho B^{\delta\rho})H_{\alpha\beta\gamma\delta}|\tilde{J}(z_1)\tilde{J}(z_2)\tilde{J}(z_3)\rangle\\
        & = \frac{\langle12\rangle}{\langle13\rangle}\langle1^223|H\rangle.
    \end{align*}
    Where the conformal block for our operator is given by
    \begin{align*}
        \langle \Tr(B^{\alpha\beta}B^\gamma_\rho B^{\delta\rho})H_{\alpha\beta\gamma\delta}|\tilde{J}(z_1)\tilde{J}(z_2)\tilde{J}(z_3)J(z_4)\rangle & = \langle23\rangle\langle1^223|H\rangle.
    \end{align*}
    This form factor is never zero on the support of momentum conservation. Therefore, such a vertex violates tree level integrability and therefore cannot be part of the theory. One could perform a similar analysis on the other candidate vertices.
\end{remark}

Let me also fix any terms that do not contain a factor of $H$:
\begin{lemma}
    The only vertices without a factor of $H$ have been accounted for.
\end{lemma}
\begin{proof}
    By looking at the twistor action, we see that we have accounted for both 3-point (with no factor of $H$) vertices that can appear. The only 4-point diagrams with no $H(z)$ dependence we can draw on twistor space are $aa\to\tilde{a}\tilde{b}$ and $aa\to ba$. The first diagram is already accounted for in the $\tilde{\Phi}D*D\Phi$ vertex. The second diagram cannot appear as a 4-point vertex in 4d, as it would have dimension greater than 4. 

    Finally, all 5-point and higher diagrams on twistor space would correspond to vertices in 4d of too high dimension. This completes the proof.
\end{proof}

\subsection*{Vertices with higher powers of $H$}

Twistor space Feynman diagrams tell us that the 4d vertex accompanying $H^n$ must have at least $2n$ scalars $\Phi$ and 1 negative helicity gluon $B$. The vertices are of the form $BD\Phi D\Phi \Phi...\Phi H^n$. This is the only possible vertex with a factor of $H^n$ since any other diagram we could draw on twistor space would introduce extra positive helicity gluons, which would lead to a vertex with dimension greater than 4 on spacetime.

We are now ready to wrap up the section with
\begin{proof}[Proof of theorem \ref{thm:genus1_full_theory}.]
So far, we have shown that the full 4d action takes the form
\begin{align*}
    \int_{\mathbb{R}^4}\Tr(BF(A)+\tilde{\Phi}D*D\Phi+BD\Phi D\Phi H+\text{terms of form $B\Phi^{2n-2}D\Phi D\Phi H^n$}).
\end{align*}
One can then recursively find the coefficients of all the higher vertices by computing amplitudes order by order and tuning the coefficients so that each amplitude vanishes. Thus, the full theory is determined by integrability.
\end{proof}

\section{Symmetry reduction for the genus 3 deformation and 4d Chern-Simons theory}\label{sec:sym_red}

In a series of papers \cite{Costello:2013zra,Costello:2013sla,Costello:2017dso,Costello:2018gyb,Costello:2019tri}, a 4-dimensional Chern-Simons type theory was introduced in the context of 2d integrability. There is an interesting connection between this theory and the work done in this paper, which I will now discuss. 

Since 2d integrability is not the main focus of this paper, I only give a brief outline of the 4d Chern-Simons story here and how it connects to the genus 3 case in this paper. I refer curious readers to the original papers or the nice review \cite{Lacroix:2021iit}.

\subsection*{4d Chern-Simons basics}

For an arbitrary Riemann surface $\Sigma$, the field of the 4d theory is a Lie algebra-valued 1-form
\begin{align*}
    \mathcal{A} = \mathcal{A}_xdx+\mathcal{A}_ydy+\mathcal{A}_0d\bar{z}\in\Omega^1(\R^2_{x,y}\times\Sigma,\g)
\end{align*}
which is holomorphic in the $\Sigma$ direction. The action is then given by
\begin{align*}
    \int_{\R^2\times\Sigma}\omega\wedge\Tr\left(\mathcal{A}d\mathcal{A}+\frac{2}{3}\mathcal{A}^3\right)
\end{align*}
for a meromorphic 1-form $\omega$ in the coordinate $z\in\Sigma$ and here $d$ is the exterior derivative on $\R^2\times\Sigma$. Note that the introduction of the 1-form $\omega$ was necessary as the Chern-Simons 3-form cannot be integrated on its own over a 4-manifold. Note also that requiring the gauge field to be a $(0,1)$-form was not strictly required but a consequence of the fact that wedging with $\omega$ kills any anti-holomorphic legs of $\mathcal{A}$.

The key feature of this action is that compactifying the $\Sigma$ direction yields a 2d integrable model with spectral parameter $z\in\Sigma$. The pole structure of $\omega$, along with the boundary conditions on the gauge field, is what determines the type of 2d integrable model one finds. Notice that boundary conditions must be imposed on $\mathcal{A}$ to keep the action regular. This construction and methods for compactifying the $\Sigma$ direction when $\Sigma = \CP^1$ are pedagogically presented in \cite{Lacroix:2021iit}. 

This story is reminiscent of the twistor correspondence. Indeed, twistor theory provides a way of deriving 4d integrable models similarly to how 4d Chern-Simons theory derives 2d integrable models (at least in the case where the curve is rational $\Sigma=\CP^1$). These theories were linked explicitly through symmetry reduction in \cite{Bittleston:2020hfv}. Symmetry reduction kills two real dimensions of a theory by demanding that the fields be orthogonal to two chosen vector fields. \cite{Bittleston:2020hfv} showed that symmetry reduction reduced an example of holomorphic Chern-Simons theory on twistor space to a 4d Chern-Simons theory of the form described above. They also showed explicitly that compactifying on twistor space or in the 4d Chern-Simons theory commutes with symmetry reduction (see Fig. 1 in \cite{Bittleston:2020hfv}).

\subsection*{The 6d-4d-2d diamond of Chern-Simons type theories}

Let me comment on the diamond correspondence of \cite{Bittleston:2020hfv} to set the scene for the rest of the section. Their paper proved that rendering the 6d holomorphic Chern-Simons action invariant under two vectors is equivalent to the 4d holomorphic-topological Chern-Simons action of Costello-Witten-Yamazaki. Schematically, they proved the equivalence below
\begin{align*}
    \underbrace{\text{Sol(6d hCS)}^{\R^2}=\text{Sol(4d CS)}}_{\text{Bittleston and Skinner}}
\end{align*}
where the superscript $\R^2$ denotes imposing invariance under two vectors. Their analysis was local on twistor space, so the result extends to our geometries when we study holomorphic Chern-Simons theories. One can then extend this chain of equivalent theories by compactifying along curves. \cite{Bittleston:2020hfv} studied this first on twistor space and in our case, we get the correspondence of theories shown in Fig. \ref{fig:diamond_theories}.

In this section, I describe and derive the theories on the nodes of the diamond in Fig. \ref{fig:diamond_theories}. As explained above, this correspondence of theories is commutative in that compactifying and symmetry reducing can be done in any order. However, I should emphasize that the compactification of the 4d Chern-Simons example coming from the genus 3 theory only has an abstract characterization given by \cite{derryberry2021laxformulationharmonicmaps,Costello:2019tri}. So, the comparison of the 2d integrable models coming from both sides of the diamond is less explicit. I provide a discussion on this in subsections \ref{subsec:cptfy_4dCS} and \ref{subsec:cmnt_4dCS}. In subsection \ref{subsec:equiv_4dCS}, I introduce an equivariant holomorphic Chern-Simons theory on $\PT_\Sigma$ when $\Sigma$ is a genus 1 curve. In this case, the 2d integrable model arising from compactifying the 4d Chern-Simons theory can be found explicitly using the methods of \cite{Lacroix:2023qlz}.

\begin{figure}
    \centering
    \includegraphics[width=0.85\linewidth]{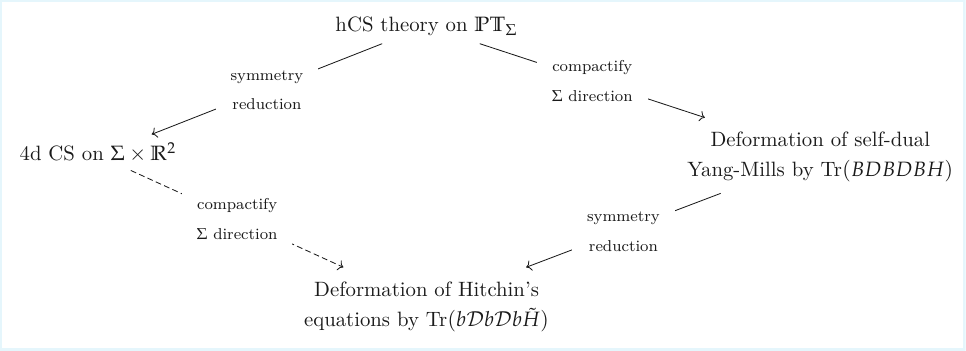}
    \caption{The diamond of theories discussed in this section, starting with the 6d theory at the top and the 4d theories it passes through, with the 2d theory at the bottom. I do not currently have an explicit description of the compactification of the 4d Chern-Simons theory, which is why I left that arrow broken. Note the abbreviations of Chern-Simons (CS) and holomorphic CS (hCS).}
    \label{fig:diamond_theories}
\end{figure}

In the rest of this section, I consider symmetry reductions of holomorphic Chern-Simons theory on $\PT_\Sigma$ in the genus 3 case and the 4d theory I derived in section \ref{sec:bootstrap}. This is to achieve the following:

\begin{enumerate}
    \item[1.] Explain the connection between the "higher genus" twistor space construction on $\PT_\Sigma$ and 4d Chern-Simons theory on $\R^2\times\Sigma$.
    \item[2.] Derive the 2d hyperelliptic integrable model corresponding to the deformation of SDYM found in section \ref{sec:bootstrap}.
\end{enumerate}

\subsection{Symmetry reduction on $\R^4$}\label{subsec:sym_red}

Recall the 4d action I derived in section \ref{sec:bootstrap}
\begin{align*}
    \int_{\R^4}\Tr(BF+BDBDBH+...)
\end{align*}
where we ignore higher-order terms for now.

Let's now choose vector fields to reduce along. Two generators of the translation group are given by
\begin{align}\label{eq:vec_f_sym}
    V_1 = c_1\frac{\partial}{\partial\bar{u}_{\dot{1}}}+c_2\frac{\partial}{\partial u_{\dot{1}}}\qquad V_2 = c_3\frac{\partial}{\partial\bar{u}_{\dot{2}}}+c_4\frac{\partial}{\partial u_{\dot{2}}}
\end{align}
for constants $c_i$ and using the coordinates on $\R^4$ introduced in section \ref{sec:bgrnd}. To ease the computation of the 2d action, it will be convenient to introduce the following coordinates
\begin{align}\label{eq:s_t_coords}
    s_1 &= \frac{c_1\bar u_{\dot 1}+c_2u_{\dot 1}}{c_1^2+c_2^2}\qquad
t_1 = \frac{-c_2\bar u_{\dot 1}+c_1u_{\dot 1}}{c_1^2+c_2^2} \\\nonumber
s_2 &= \frac{c_3\bar u_{\dot 2}+c_4u_{\dot 2}}{c_3^2+c_4^2}\qquad 
t_2 = \frac{-c_4\bar u_{\dot 2}+c_3u_{\dot 2}}{c_3^2+c_4^2}.
\end{align}

So that the vector fields simply take the form
\begin{align*}
    V_1=\frac{\partial}{\partial s_1}\qquad V_2=\frac{\partial}{\partial s_2}.
\end{align*}

The main step in the symmetry reduction is to demand the vanishing of the Lie derivatives
\begin{align}\label{eq:Lie_der}
    \mathcal{L}_{V_1}B = \mathcal{L}_{V_2}B = \mathcal{L}_{V_1}A = \mathcal{L}_{V_2}A = 0.
\end{align}
Technically, one should be careful about how the action of these translations lift to the $G$-bundle. It is important to choose a frame for the $G$-bundle that is invariant under the lifted action. 

It will be useful to note that the self-dual 2-forms in the new coordinates are given by
\begin{align*}
    \omega_1 := ds_1\wedge ds_2 + dt_1\wedge dt_2,\qquad
\omega_2 := ds_1\wedge dt_1 + ds_2\wedge dt_2,\qquad
\omega_3 := ds_1\wedge dt_2 - ds_2\wedge dt_1.
\end{align*}

To continue, let $T$ be the plane spanned by $t_1,t_2$. Then we can see the effect of the vanishing Lie derivatives directly on the gauge fields:
\begin{align*}
    A = A_T+\phi_1ds_1+\phi_2ds_2\quad B = b_1\omega_1 + b_2\omega_2 + b_3\omega_3
\end{align*}
where all fields are $\g$ valued and depend only on $t_i$ due to equation \ref{eq:Lie_der} and $A_T\in\Omega^1(T,\g)$. Next, let $d_T$ be the exterior derivative on $T$ and $D_T = d_T+[A_T,\cdot]$ so that we have the 2d curvature $F_T = D_TA_T$. Also, define $D_{t_i}\phi = d_{t_i}\phi+[A_{t_i},\phi]$. Then, keeping only the $d^2t\wedge d^2s$ terms in the $BF$ part of the action leaves
\begin{align*}
    \int_{T}
\mathrm{Tr}\Big(
 b_1\big(F_T + [\phi_1,\phi_2]\big)
 +
 b_2\big(D_{t_1}\phi_1 + D_{t_2}\phi_2\big)
 +
 b_3\big(D_{t_2}\phi_1 - D_{t_1}\phi_2\big)
\Big).
\end{align*}
Where we have normalized by a factor of the volume of the group of translations generated by $V_1,V_2$ to avoid dealing with compactification subtleties.

\begin{remark}
    The equations of motion for this theory are 
    \begin{align*}
        F_T + [\phi_1,\phi_2] = 0\quad
D_{t_1}\phi_1 + D_{t_2}\phi_2 = 0\quad
D_{t_2}\phi_1 - D_{t_1}\phi_2 = 0.
    \end{align*}
    Which can be recast as
    \begin{align*}
        D_T\Phi=0\quad F_T+[\phi_1,\phi_2] = 0
    \end{align*}
    by defining $\Phi = \phi_1dt_1+\phi_2dt_2$. These are exactly Hitchin's equations \cite{Hitchin:1986vp} as expected for a symmetry reduction of SDYM.
\end{remark}
Now, let's discuss what happens to the $BDBDBH$ vertex. To complete the reduction of this term we will need to split the 4d covariant derivative into the 2d covariant derivatives
\begin{align}\label{eq:2d_cov_der}
    \mathcal{D}_1 = D_{t_1}\quad \mathcal{D}_2 = D_{t_2}\quad \mathcal{D}_3 = [\phi_1,\cdot]\quad \mathcal{D}_4 = [\phi_2,\cdot].
\end{align}
The 4d covariant derivative is then given by
\begin{align*}
    D_{\alpha\dot{\alpha}} = \sigma^i_{\alpha\dot{\alpha}}\mathcal{D}_i
\end{align*}
for a choice of Pauli matrices $\sigma^i$. The vertex is then given by 
\begin{align*}
    \Tr(B_{\alpha\beta}D_\gamma^{\dot{\alpha}}B_{\delta\rho}D_{\sigma\dot{\alpha}}B_{\zeta\eta}H^{\alpha\beta\gamma\delta\rho\sigma\zeta\eta}) = \Tr(b_a\sigma^{i\dot{\alpha}}_\gamma \mathcal{D}_ib_b\sigma^j_{\sigma\dot{\alpha}}\mathcal{D}_jb_c)\omega^a_{\alpha\beta}\omega^b_{\delta\rho}\omega^c_{\zeta\eta}H^{\alpha\beta\gamma\delta\rho\sigma\zeta\eta}
\end{align*}
which we can write as 
\begin{align*}
    \Tr(b_a \mathcal{D}_ib_b\mathcal{D}^ib_c)\tilde{H}^{abc}
\end{align*}
by defining $\tilde{H}^{abc} = \omega^a_{\alpha\beta}\sigma^{i\dot{\alpha}}_\gamma\omega^b_{\delta\rho}\sigma_{i\sigma\dot{\alpha}}\omega^c_{\zeta\eta}H^{\alpha\beta\gamma\delta\rho\sigma\zeta\eta}$. So, the full 2d theory is
\begin{align}\label{eq:genus3_Hitchin}
    \int_{T}\Tr\left(
 b_1(F_T + [\phi_1,\phi_2])+
 b_2(D_{t_1}\phi_1 + D_{t_2}\phi_2)+
 b_3(D_{t_2}\phi_1 - D_{t_1}\phi_2\right)
+b_a \mathcal{D}_ib_b\mathcal{D}^ib_c\tilde{H}^{abc}+...).
\end{align}
Where again, we have normalized by a factor of the volume of the translation group. Finally, we note the schematic form of the higher-order vertices:
\begin{align*}
    \int_{\R^4}\Tr(B^{2n-1}DBDBH^n)\to \int_{\R^2}\Tr(b^{2n-1}\mathcal{D}b\mathcal{D}b\tilde{H}^n)
\end{align*}
where the indices on the factors of $\tilde{H}$ contract with the rest of the vertex in any way possible, and the coefficients are unknown but fixed by integrability. Note also that I have abused notation slightly here by writing $\tilde{H}^n$ since the way in which the indices contract with $H$ in these higher-order vertices may be different from the way I defined $\tilde{H}$ above.

\subsection*{Comments on Hitchin's system}

Let me conclude this subsection about the computation of the symmetry-reduced action by highlighting that Hitchin's equations describe the moduli of $G$-bundles on our space $T$. When presented as a BF type theory as discussed above, the integrable system describes the cotangent bundle of this moduli space and the claim of this work is that the introduction of the new deformation vertex modifies the spectral data in a $\Sigma$ dependent way. Understanding the explicit modification to the Hitchin moduli problem is challenging, but I present some speculations here.

In section \ref{sec:celestial}, I made the striking observation that the introduction of the $H$-deformed vertex lifted the Kac-Moody symmetry of self-dual Yang-Mills to $\Sigma$. This suggests that the spectral data of the theory gets modified to $\Sigma$. This continues to hold true in the symmetry reduced theory. Since all vectors in the plane can be rotated into one another, we can consider a symmetry reduction that effectively sets the $v_{\dot{2}}$ plane in twistor space to zero. This would correspond to restricting the chiral algebra tower of states to
\begin{align*}
    J_a[p,q](z),\tilde{J}_b[r,s](z)\to J_{\pm,a}[p,0](z).
\end{align*}
Due to the way the chiral algebra generators couple to the bulk gauge field on twistor space \cite{Costello:2022wso}. This means that we realize the same symmetry lifting in the deformed Hitchin's system.

The second way we can understand how the deformation term deforms the moduli problem is via the analysis of the Lax pair done in section \ref{sec:lax}. There I made some arguments about a leading order Lax formulation for the 4d theory and explained that one should be able to recast the second order Lax system into a first order one depending on a spectral parameter on $\Sigma$ by performing an appropriate field redefinition. One could easily symmetry reduce that equation to get a description of this deformed Hitchin's system. I conjecture that one would recover the form of the Lax matrix presented in \cite{Bittleston:2026tdr} with spectral parameter $z\in\Sigma$ under the appropriate field redefinition.

\subsection{Symmetry reduction on $\PT_\Sigma$}\label{subsec:sym_red_PTSigma}
Recall that the starting theory is holomorphic Chern-Simons theory on $\PT_\Sigma$ with gauge field $\alpha\in\Omega^{0,1}(\PT_\Sigma,\mathcal{O}_\Sigma\otimes\g)$ (see the setup in section \ref{sec:model_build}). The vector fields we symmetry reduce by are the $\PT_\Sigma$ uplifts of the vector fields chosen in the 4d symmetry reduction (equation \ref{eq:vec_f_sym}). The vector fields on $\PT_\Sigma$ are
\begin{align*}
    V_1=(c_1+z c_2)\frac{\partial}{\partial v_{\dot{1}}}\quad V_2=(c_3+z c_4)\frac{\partial}{\partial v_{\dot{2}}}
\end{align*}
for constants $c_i$. The first step in the symmetry reduction is to demand that
\begin{align*}
    \mathcal{L}_{V_1}\alpha = \mathcal{L}_{V_2}\alpha = 0.
\end{align*}
To ease computations, it is convenient to treat $\alpha$ as an ordinary $1$ form while imposing the following
\begin{align*}
    \partial_{v_{\dot{1}}}\vee \alpha = \partial_{v_{\dot{2}}}\vee \alpha = \partial_{z}\vee \alpha = 0.
\end{align*}
We can then use gauge transformations to require 
\begin{align*}
    \text{Re}V_1\vee \tilde{\alpha} = \text{Re}V_2\vee \tilde{\alpha} = \dbar_z\vee \tilde{\alpha} = 0
\end{align*}
where $\tilde{\alpha}$ is the gauge fixed field. To continue, I will have to say something about the volume form on $\PT_\Sigma$. $\PT_\Sigma$ is Calabi-Yau so it has a holomorphic volume form $\Omega_\Sigma$. However, the coordinates in the $\pi^*\mathcal{O}(1)^2$ directions both had first-order poles at the two preimages of the point at infinity, which I call $p_1,p_2$. So, the 2 form $dv^{\dot{1}}dv^{\dot{2}}$ has a second order pole at these two points on $\Sigma$. So, we want to find a 1 form on $\Sigma$ that cancels these poles to wedge against $dv^{\dot{1}}dv^{\dot{2}}$. There are 3 natural holomorphic 1-forms on $\Sigma$,
\begin{align*}
    dz/w\quad zdz/w\quad z^2dz/w
\end{align*}
which have second, first, and no zero at $p_1,p_2$ respectively. Also, recall that the $\Sigma$ coordinates were related via $w^2 = H(z)$. The clear choice for our holomorphic volume form is then 
\begin{align*}
    \Omega_\Sigma = \frac{dzdv^{\dot{1}}dv^{\dot{2}}}{w}.
\end{align*}
The symmetry-reduced action then takes the form
\begin{align*}
    V_1\vee V_2\vee(\Omega_\Sigma\wedge hCS(\tilde{\alpha}) = \frac{(c_1+zc_2)(c_3+zc_4)}{(z-z_\infty)^2}\frac{dz}{w}\wedge \Tr\left(\tilde{\alpha}d\tilde{\alpha}+\frac{2}{3}\tilde{\alpha}^3\right).
\end{align*}
Define the 1-form
\begin{align*}
    \omega = \frac{(c_1+zc_2)(c_3+zc_4)}{(z-z_\infty)^2}\frac{dz}{w}.
\end{align*}
\begin{remark}
    The 1-form $\omega$ must satisfy
    \begin{align*}
        \#\{\text{zeros of }\omega\}-\#\{\text{poles of }\omega\} = 2g-2
    \end{align*}
    which is a consequence of the Riemann-Roch theorem. Let's check this in our case. First, notice that the factor of $1/(z-z_\infty)^2$ introduces a second-order pole at both $(z_\infty,w)$ and $(z_\infty,-w)$ (I called these $p_1$ and $p_2$ earlier), which are cancelled by the second-order zeros at $p_1,p_2$ in $dz/w$. Now, the factor of $(c_1+zc_2)(c_3+zc_4)$ has 4 zeros: $(-c_1/c_2,\pm w)$ and $(-c_3/c_4,\pm w)$. So, we see that
    \begin{align*}
        \#\{\text{zeros of }\omega\}-\#\{\text{poles of }\omega\} = 4 = 2g-2.
    \end{align*}
\end{remark}
To continue, we will have to find an explicit form for the gauge fixed field $\tilde{\alpha}$. Recall that we lifted the field to an honest 1-form in the $\pi^*\mathcal{O}(1)^2$ directions so that the field has $d\bar{v}^{\dot{\alpha}}$ legs. Call the fields in these directions $\bar{\alpha}_{\dot{\alpha}}$.

We then constrain the field in the following way:
\begin{align*}
    \text{Re}V_1\vee \alpha & = 0 \implies (\bar{c}_1+\bar{c}_2\bar{z})\bar{\alpha}_{\dot{1}}+(c_1+c_2z)\alpha_{\dot{1}} = 0\\
    \text{Re}V_2\vee \alpha & = 0 \implies (\bar{c}_3+\bar{c}_4\bar{z})\bar{\alpha}_{\dot{2}}+(c_3+c_4z)\alpha_{\dot{2}} = 0.
\end{align*}
The $\Sigma$ component of the gauge field is left untouched while we gauge fix the other directions to satisfy the constraints by
\begin{align*}
    \tilde{\alpha}|_{\pi^*\mathcal{O}(1)^2} = \left(d\bar{v}^{\dot{1}}-\frac{\bar{c}_1+\bar{c}_2\bar{z}}{c_1+c_2z}dv^{\dot{1}}\right)\bar{\alpha}_{\dot{1}}+\left(d\bar{v}^{\dot{2}}-\frac{\bar{c}_3+\bar{c}_4\bar{z}}{c_3+c_4z}dv^{\dot{2}}\right)\bar{\alpha}_{\dot{2}}
\end{align*}

We then use the new coordinate system $x = C_1v_{\dot{1}}+\bar{C}_1\bar{v}_{\dot{1}}$ and $y = C_2v_{\dot{2}}+\bar{C}_2\bar{v}_{\dot{2}}$ for constants $C_1,C_2$ subject to
\begin{align*}
    C_1(c_1+zc_2)+\bar{C}_1(\bar{c}_1+\bar{z}\bar{c}_2) = 0\quad C_2(c_3+zc_4)+\bar{C}_2(\bar{c}_3+\bar{z}\bar{c}_4) = 0.
\end{align*}
One choice that works is the following:
\begin{align*}
    C_1 = i(\bar{c}_1+\bar{z}\bar{c}_2)\quad C_2 = i(\bar{c}_3+\bar{z}\bar{c}_4).
\end{align*}
Then we see that
\begin{align*}
    \tilde{\alpha}|_{\pi^*\mathcal{O}(1)^2} = -i(C_1dv^{\dot{1}}+\bar{C}_1d\bar{v}^{\dot{1}})\frac{\alpha_{\dot{1}}}{\bar{c}_1+\bar{c}_2\bar{z}}-i(C_2dv^{\dot{2}}+\bar{C}_2d\bar{v}^{\dot{2}})\frac{\alpha_{\dot{2}}}{\bar{c}_3+\bar{c}_4\bar{z}} = \mathcal{A}_xdx+\mathcal{A}_ydy
\end{align*}
where we make the definition
\begin{align*}
    \mathcal{A}_x = -i\frac{\bar{\alpha}_{\dot{1}}}{c_1+c_2z}\quad\mathcal{A}_y = -i\frac{\bar{\alpha}_{\dot{2}}}{c_3+c_4z}.
\end{align*}
The $\Sigma$ direction of our gauge field is unchanged $\mathcal{A}_0 = \alpha_0$ (up to an implicit pullback of the inclusion $\{c_1+zc_2 = c_3+c_4z = 0\}\hookrightarrow\PT_\Sigma$).

Notice that we have introduced poles in our gauge field, which cancel the zeros in $\omega$. We have thus constructed our 4d Chern-Simons theory with action
\begin{align}\label{eq:4dCS_action_Sigma}
    \int_{\Sigma\times\mathbb{R}^2}\omega\wedge CS(\mathcal{A}).
\end{align}

\subsection*{Remarks on affine Gaudin models}

Having derived the form of the 4d Chern-Simons theory above, one might be tempted to relate this to the interpretation of the structure of 4d Chern-Simons theories in terms of affine Gaudin models \cite{FeiginFrenkel2011,10.1093/imrn/rny128,Vicedo2021}. \cite{Vicedo2021} considers 4d Chern-Simons theories on products of planes and punctured spheres. The point of the analysis being that one can interpret these punctures as the pole structure of the Gaudin Lax matrix.

It would be a natural direction to attempt to understand how the examples I derived fit in to the landscape of affine Gaudin models. However, the connection is unclear at this point. The logical first step would be to reinterpret these examples in terms of 4d Chern-Simons theories with spectral curve a punctured sphere. One can easily formulate the 4d Chern-Simons example derived here to this background - this is exactly a 4d version of the mechanism performed in 6d in section \ref{sec:model_build}. The problem arises in how the action changes under the pushforward. Just as in the 6d case, we would find a deformation of holomorphic-topological BF theory which does not have a known interpretation in terms of affine Gaudin models. 

\subsection{Abstract characterization of the compactified 4d Chern-Simons theory}\label{subsec:cptfy_4dCS}

In \cite{derryberry2021laxformulationharmonicmaps}, Derryberry outlines an abstract method for compactifying 4d Chern-Simons theories on higher genus curves. I will now outline his procedure to better understand the 2d integrable model associated with the 4d Chern-Simons example derived above (equation \ref{eq:4dCS_action_Sigma}). It is important to emphasize that this description of the 2d model is far less explicit than the rest of the theories discussed in this paper since an explicit compactification by solving equations of motion is poorly understood in this example which I elaborate on in subsection \ref{subsec:cmnt_4dCS}. This subsection is included so that I can make some basic comparisons between the two sides of the diamond correspondence in Fig. \ref{fig:diamond_theories}. I leave a concrete treatment of this for future work.

Derryberry's method for constructing the 2d model involves deriving a metric and 3-form on the target space for a $\sigma$-model with action
\begin{align*}
S[\sigma] = \int_{\mathbb{R}^2} \| d\sigma \|^2 d\mathrm{vol}_{\mathbb{R}^2} 
+ \frac{1}{3} \int_{\mathbb{R}^2 \times \mathbb{R}_{\geq 0}} \tilde{\sigma}^*(\Omega).
\end{align*}
Here $||\cdot||$ is the metric and $\Omega$ is the 3-form. We have also made the standard extension $\tilde{\sigma}$ of $\sigma$ to $\R^2\times\R_{\geq0}$.

The target space $\mathcal{M}$ is a subspace of the moduli of $G$-bundles on $\Sigma$. To define the target space explicitly, we define the divisor 
\begin{align*}
    D_x = \sum_{\text{poles of }\mathcal{A}_x} p
\end{align*}
and similarly for $D_y$. We can then define the subspace $\text{Bun}_G(\Sigma|D_x)\subset \text{Bun}_G(\Sigma)$ to be the open subset of the moduli space of $G$-bundles that satisfies $H^\bullet(\Sigma,\g_P(D_x)) = 0$. Here $\g_P=(P\times \g)/G$ and $\g_P(D_x) = \g_P\otimes\mathcal{O}(D_x)$. We then take the $\sigma$-model target space to be
\begin{align*}
    \mathcal{M} = \text{Bun}_G(\Sigma|D_x).
\end{align*}
Notice that by Serre duality, $\text{Bun}_G(\Sigma|D_x)=\text{Bun}_G(\Sigma|D_y)$. 

The next step is to understand a certain kernel known as the Szeg\"o kernel \cite{derryberry2021laxformulationharmonicmaps,Costello:2019tri}. I did not manage to derive the Szeg\"o kernel for my divisors on the genus 3 curve, so I could not find an explicit description of the 2d action arising from this method. However, we can still gather some important information by studying the general form of this action which will be useful to compare with the deformed Hitchin's equation of subsection \ref{subsec:sym_red}.

The first thing we can do is count the degrees of freedom of this $\sigma$-model. We will do so by analyzing the target space. The moduli space of $G$-bundles on a genus $3$ curve is $\C^{2\dim\g}$. To see this, notice that locally near the point in the moduli space corresponding to the trivial $G$-bundle, the patch is the quotient of $H^{0,1}(\Sigma)\otimes\g\cong\C^{3\dim\g}$ by the adjoint action of $G$. Concretely, we are quotienting the set of holomorphic structures on the trivial bundle by constant gauge transformations that leave the trivial connections fixed. Quotienting by the adjoint action thus brings the dimension down by $\dim\g$, so we see that the moduli of $G$-bundles is $2\dim\g$ dimensional. So, our $\sigma$-model describes two adjoint valued scalars with an action involving two derivatives.

\subsection*{Basic evidence that symmetry reduction commutes with compactification}

I will now compare the 2d model given in equation \ref{eq:genus3_Hitchin} with the $\sigma$-model described above. While I could not prove an explicit match between the 2d deformation of Hitchin's equation and the $\sigma$-model (since I do not have an explicit description of the metric/3-form in the Derryberry picture or the higher-order vertex coefficients in the deformed Hitchin model picture), I gather basic evidence that the two theories are consistent. While these two descriptions of the 2d model must be equivalent due to the work of Bittleston and Skinner \cite{Bittleston:2020hfv}, the following analysis is only a surface level comparison and a proof that these models are equivalent is an important future direction that I do not attempt here.

Let me first describe the degrees of freedom arising from the 2d deformation of Hitchin's equations. Since the symmetry reduction does not change the number of fields, we can perform this analysis on the 4d action. Suppose we turned on a VEV for $B$ in the 4d theory so that $B\to B_0+B$ where $B_0$ is constant and $B$ is generic. Then the action up to quadratic factors is
\begin{align*}
    \int_{\R^4}\Tr\left(BdA+\frac{1}{2}B_0[A,A]+B_0dBdBH\right).
\end{align*}
This is a Higgs mechanism which introduced a mass term for $A$. We can then vary $A$ to find that
\begin{align*}
    dB = [A,B_0].
\end{align*}
This fixes the form of $A$ entirely in terms of $B$ and $B_0$. It also tells us that we can integrate away the mass term for $A$ and use integration by parts to see that the kinetic terms for our theory can be recast as (at least schematically as)
\begin{align*}
    \int_{\R^4}\Tr(dBdB+B_0dBdBH).
\end{align*}
We can then use the available gauge symmetry to gauge away one of the components of $B$ so that the only remaining degrees of freedom are two of the components of $B$. Our theory thus describes a pair of adjoint valued scalars with a scalar kinetic term deformed by $H$. This is consistent with the Derryberry description. We also see that the higher-order terms each contain two derivatives in $B$. This provides some basic evidence that the symmetry reduction commutes with the compactification on $\PT_\Sigma$.

\subsection{Comments on higher genus 4d Chern-Simons}\label{subsec:cmnt_4dCS}

When $\Sigma$ is a genus $g>0$ surface, the 4d Chern-Simons compactification procedure is less understood, and only abstract characterizations such as the one discussed above are known. Recently, this was done for a genus 1 curve with gauge fields twisted by a carefully chosen $\sl(N)$ bundle \cite{Lacroix:2023qlz}. In the next subsection, I build an analogous bundle on $\PT_\Sigma$ for $\Sigma$ a genus 1 curve and show that symmetry reducing the theory lands on the one described in \cite{Lacroix:2023qlz}.

Let me explain why compactifying on the genus 3 curve is tricky. In the rational ($g=0$) case, the usual first step in compactifying the 4d Chern-Simons action is to demand that the gauge field be pure gauge in the $\Sigma$ direction. When $\Sigma = \CP^1$, this is possible since $\CP^1$ has no first Dolbeault cohomology. In the genus 1 setting, the authors of \cite{Lacroix:2023qlz} were also able to set the $\Sigma$ direction of the field to be gauge trivial because the field was twisted by the $\sl(N)$ bundle, which had vanishing first Dolbeault cohomology. In the setting studied in this paper, $H_{\dbar}^1(\Sigma) = \C^3$, so obviously, there is an obstruction to trivializing the $\Sigma$ direction of the gauge field. My hope is that one could choose clever representatives of the cohomology classes to gauge fix the $\Sigma$ component of $\tilde{\alpha}$ in such a way that compactifying the action in equation \ref{eq:4dCS_action_Sigma} becomes tractable.

If one could fix the coefficients of the correction terms in the symmetry reduced action discussed in \ref{subsec:sym_red}, I believe that the theory should match the compactified 4d Chern-Simons model discussed in \ref{subsec:sym_red_PTSigma}. This correspondence of theories is illustrated in Fig. \ref{fig:diamond_theories}. This would be a genus 3 analog of the correspondence of theories that Bittleston and Skinner found in the genus 0 setting \cite{Bittleston:2020hfv}.

\subsection{Introducing equivariance in the genus 1 setting}\label{subsec:equiv_4dCS}

This subsection deviates slightly from the theme of the paper and is primarily intended for those interested in 2d integrability. The purpose being to recover a more familiar presentation of the 2d model as a concrete $\sigma$-model. This gives an explicit realization of the left side of the type of diamond correspondences discussed previously in this section. I will also present an action for the 4d integrable model and its symmetry reduction using the techniques of section \ref{sec:model_build}, completing the right side of the diamond. As discussed at the start of this section, the 2d models coming from both sides of the diamond must coincide. However, I leave checking this equivalence for future work.

To connect with the work of \cite{Lacroix:2023qlz}, we want to build a principal holomorphic $SO_3(\C)$ ($\cong SL_2(\C)/\Z_2$) bundle over our double cover of twistor space $\PT_\Sigma$ in the case that $\Sigma$ has genus 1. I will now build this bundle on $\PT_\Sigma$ using the language of this paper. To start, we need a rank 3 vector bundle $V$, with an inner product $V\otimes V\to \mathcal{O}_\Sigma$. 

\textit{How do we intuitively think of line bundles on $\PT_\Sigma$?} The elliptic curve $\Sigma$ is topologically a torus. Consider the lattice representation of the complex torus as a quotient $\C/(2p\Z+2r\Z)$ with lattice periods $2p,2r$ ($p,r\in\C$) and origin $O$. Using these special points, we construct the line bundles 
\begin{align*}
\mathcal{L}_1 = \mathcal{O}(p-O)\quad \mathcal{L}_2 = \mathcal{O}(q-O)\quad  \mathcal{L}_3= \mathcal{O}(r-O) 
\end{align*}
whose sections have a first order pole at $p$ (resp. $r,q=p+r$) and zero at $O$. Viewing the torus as two Riemann spheres glued along branch cuts, $\mathcal{L}_1$ is the bundle whose sections pick up a sign when crossing one branch cut and $\mathcal{L}_3$ picks up a sign when crossing the other cut. $\mathcal{L}_2$ has sections that pick up a sign at both cuts (in other words, odd functions on $\Sigma$). This is the quasi-periodicity condition that Lacroix and Wallberg \cite{Lacroix:2023qlz} use in the case $N = 2$.

Our line bundles should satisfy:
\begin{enumerate}
    \item $H^1_{\dbar}(\mathcal{L}_i) = 0$
    \item There exists a pairing such that $\mathcal{L}_i\otimes\mathcal{L}_i=\mathcal{O}_\Sigma$
\end{enumerate}
The first condition ensures that we can use the methods of \cite{Lacroix:2021iit} to compactify along $\Sigma$ once we have reduced to the elliptic 4d Chern-Simons. To see that we get this condition, we note the following: On any compact Riemann surface $\Sigma$ and vector bundle $V$, we have the map 
\begin{align*}
H^0(\Sigma,V)\otimes H^1(\Sigma,V^*\otimes K_\Sigma)\to\mathbb{C}    \quad\text{given by}\quad v\times \tilde{v}d\bar{z}dz\mapsto\int_{\Sigma}\langle v,\tilde{v}\rangle d\bar{z}dz.
\end{align*}
This gives $H^0(\Sigma, V)=H^1(\Sigma,V^*\otimes K_\Sigma)^*$. For elliptic curves, $K_\Sigma = \mathcal{O}_\Sigma$. So, 
\begin{align*}
    H^1(\mathcal{L}_i) = H^0(\mathcal{L}_i^*)^*=\{\text{functions with a pole at $0$ and zero at $p = 0$}\} = 0
\end{align*}
where the last equality follows since there does not exist such a function.

For the second condition, we can see that $\mathcal{L}_1\otimes\mathcal{L}_1=\mathcal{O}(2p-2O) = \mathcal{O}_\Sigma$ which is true as it is possible to write down a doubly periodic function with this property. The same holds for $i = 2,3$. 

The theory is then given by the holomorphic Chern-Simons action on $\PT_\Sigma$ with gauge field 
\begin{align*}
    \alpha_i\in \Omega^{0,1}(\PT_\Sigma,\mathcal{L}_i\otimes\pi^*\mathcal{O}(-1))
\end{align*}
where $i$ is an $SO(3)$ index.

\begin{remark}
    The fact that $\mathcal{L}_i\otimes\mathcal{L}_i=\mathcal{O}_\Sigma$ ensures that we get a well defined kinetic term. One could have constructed a holomorphic Chern-Simons gauge field in the same fashion as the previous model builds by choosing a gauge field valued in $\pi^*\mathcal{O}(-1)$ rather than $\mathcal{L}_i\otimes\pi^*\mathcal{O}(-1)$. The kinetic term would still be well-defined, but the pushforward bundle
    \begin{align*}
        \pi_*\pi^*\mathcal{O}(-1)=\mathcal{O}(-1)\otimes(\mathcal{O}\oplus\mathcal{O}(-2))=\mathcal{O}(-1)\oplus\mathcal{O}(-3).
    \end{align*}
    This is technically fine but implies that the 4d theory we get involves a pair of bosonic spin 1/2 fields which is physically unsatisfying as it violates spin-statistics. Twisting the gauge field by $\mathcal{L}=\bigoplus_i\mathcal{L}_i$ remedies this as
    \begin{align*}
    \pi_*(\mathcal{L}_i\otimes\pi^*\mathcal{O}(-1)) = (\mathcal{O}(-1)\oplus \mathcal{O}(-1))\otimes\mathcal{O}(-1)=\mathcal{O}(-2)\oplus \mathcal{O}(-2).
\end{align*}
So, the 4d theory will be a theory of two scalars.
\end{remark}

\subsection*{Symmetry reduction to elliptic 4d Chern-Simons}

I will discuss the corresponding twistor theory and its Penrose transform towards the end of the section. Let's now see how this theory relates to the family of theories of \cite{Lacroix:2023qlz}. The computation in this section is identical to the one performed in subsection \ref{subsec:sym_red_PTSigma}. The only difference being that our \textit{meromorphic} volume form 
\begin{align*}
    \Omega_\Sigma=dwdv^{\dot{1}}dv^{\dot{2}}\in\Omega^{3,0}(\PT_\Sigma,\pi^*\mathcal{O}(2))
\end{align*}
involves the canonical holomorphic 1-form $dw$ on the elliptic curve $\Sigma$. In particular, our 4d Chern-Simons 1-form is of the form
\begin{align*}
        \omega = \frac{(c_1+zc_2)(c_3+zc_4)}{(z-z_\infty)^2}dw.
\end{align*}
Which has the same number of poles and zeros consistent with the Riemann-Roch theorem. The gauge field of the 4d Chern-Simons theory then adopts poles and zeros
\begin{align*}
    \mathcal{A}_x = -i\frac{z-z_\infty}{c_1+c_2z}\bar{\alpha}_{\dot{1}}\qquad\mathcal{A}_y = -i\frac{z-z_\infty}{c_3+c_4z}\bar{\alpha}_{\dot{2}}.
\end{align*}
Again, the gauge field leg in the direction of $\Sigma$ remains unchanged and I will denote it by $\mathcal{A}_0$ for the rest of the analysis.

In what follows I will explain that this data provides an example of the family of theories studied in \cite{Lacroix:2023qlz} and I borrow their tools to compute the corresponding 2d integrable model.

\subsection*{Compactification from elliptic 4d Chern-Simons to a 2d integrable model}

The first step in this analysis is to parameterize the $\Sigma$ component of the gauge field in terms of a smooth $SL_2(\mathbb{C})$ valued field $\hat{\sigma}$ on $\Sigma\times\mathbb{R}^2$
\begin{align*}
    \mathcal{A}_0 = \hat{\sigma}^{-1}\bar{\partial}_\Sigma\hat{\sigma}.
\end{align*}
This parametrization exists and is unique due to arguments of \cite{Lacroix:2023qlz}. This leads us to build the 2d Lax connection. Let
\begin{align*}
    \mathcal{L}_x = \hat{\sigma}^{-1}\mathcal{A}_x\hat{\sigma}+\hat{\sigma}^{-1}\partial_x\hat{\sigma}
\end{align*}
and similarly for $\mathcal{L}_y$. So, in terms of the fields $\hat{\sigma},\mathcal{L}_{x},\mathcal{L}_y$, the components of the gauge field takes the form
\begin{align*}
    \mathcal{A}_0 =  \hat{\sigma}^{-1}\bar{\partial}_\Sigma\hat{\sigma}\qquad\mathcal{A}_x = \hat{\sigma}^{-1}\mathcal{L}_{x}\hat{\sigma}-(\partial_{x}\hat{\sigma})\hat{\sigma}^{-1}\qquad\mathcal{A}_y = \hat{\sigma}^{-1}\mathcal{L}_{y}\hat{\sigma}-(\partial_{y}\hat{\sigma})\hat{\sigma}^{-1}.
\end{align*}
Now, recall that $\omega$ has two zeros: $z_{x} = -c_1/c_2$ and $z_{y} = -c_3/c_4$. We have to be careful however, as these are zeros of $\omega(z)$ when $z = z_{x},z_y$ which each correspond to two points on the curve (eg. $(z_x,w_x)$ and $(z_x,-w_x)$). In terms of the Weierstrass function, we consider the set of 4 zeros $\{\pm\gamma_x,\pm\gamma_y\}$ where $\wp(\gamma_{x}) = z_x$ and the same for $\gamma_y$. \cite{Lacroix:2023qlz} showed that there exists 2d fields $U_{x,\pm},U_{y,\pm}:\mathbb{R}^2\to so(3)$ such that 
\begin{align}\label{eq:explicit_Lax_cnntn}
    \mathcal{L}_x = \sum_{b\in\{\pm1\}}\sum_{a\in\{0,1,2\}}r^a(\gamma+b\hat{\gamma}_{\dot{\alpha}})U_{x,\text{sgn}(b)}^aT_a
\end{align}
where $T_a$ are the Pauli matrices (Belavin basis), $r^a$ is a family of functions on $\Sigma$ built from $\sigma$ functions in \cite{Lacroix:2023qlz}, and $U^a_{x,\pm}$ is the component of $U_{x,\pm}$ along the basis element $T_a$. We also have the analogous expression for $\mathcal{L}_y$.

Recall that $\omega$ had a pole of order 2 at $z=z_\infty$. This point corresponds to two points on $\Sigma$, one for each sheet of the curve. In the Weierstrass picture, We call these points $\gamma = p,-p$. Following the analysis of \cite{Lacroix:2023qlz}, we are led to consider the 2d fields
\begin{align*}
    \sigma_{\pm}(x,y) = \hat{\sigma}(x,y,\gamma = \pm p)
\end{align*}
Evaluating the Lax connection at the poles gives
\begin{align*}
    \mathcal{L}_{x}(\gamma=\pm p) = \sigma_\pm^{-1}\partial_{x}\sigma_\pm = j_{x,\pm}\qquad \mathcal{L}_{y}(\gamma=\pm p) = j_{y,\pm}
\end{align*}
which is the Maurer-Cartan current of $\sigma$. Using equation \ref{eq:explicit_Lax_cnntn}, we can write
\begin{align*}
    j_{x,\pm}^a = \sum_{b\in\{\pm1\}}r^a(\pm p+b\gamma_{x})U_{x,\text{sgn}(b)}^aT_a\qquad j_{y,\pm}^a = \sum_{b\in\{\pm1\}}r^a(\pm p+b\gamma_{y})U_{y,\text{sgn}(b)}^aT_a.
\end{align*}
This gives an invertible relation between the fields $j_{\dot{\alpha},\pm}^a$ and $U_{\dot{\alpha},b}^a$. In appendix \ref{app:equiv_4dCS}, I invert this map to derive the action for our $\sigma$-model:
\begin{align}\label{eq:2d_equiv1}
    S_{2d}[\sigma_{\pm}] = \int_{\mathbb{R}^2}dxdy\sum_{s,\tilde{s}=\pm}D^{(s\tilde{s})}_{ab}j_{\dot{1},\tilde{s}}^aj_{\dot{2},s}^b -\frac{1}{2}\sum_{s=\pm}\Res_{\gamma=sp}(\omega)\mathcal{I}_{WZ}[\sigma_{\text{sgn}(s)}]
\end{align}
where we have introduced the operator
\begin{align*}
    D^{(s\tilde{s})}_{ab} = \frac{1}{2}\sum_{i = 0,1}\sum_{\dot{\alpha}=\dot{1},\dot{2}}\tilde{s}s\Res_{\gamma = sp}(\omega)^{i+1(2)}F(p)^{i(2)}G_{\dot{\alpha}}^a(p) \frac{d^{i(2)}}{d\gamma^{i(2)}}r^a(\gamma+\tilde{s}\hat{\gamma}_{\dot{\alpha}})|_{\gamma=sp}\delta_{ab}.
\end{align*}
The functions $F,G$ are derived in appendix \ref{app:equiv_4dCS}. $\mathcal{I}_{WZ}$ denotes the Wess-Zumino term.

We have thus shown that the 2d integrable model corresponding to our elliptic 4d CS theory is a coupled deformed principal chiral model with two fields and a Wess-Zumino term.

\subsection*{The 4d integrable model and its symmetry reduction}

Following the general yoga introduced in section \ref{sec:model_build}, we consider the pushforward bundle $\pi_*(\mathcal{L}_i\otimes\pi^*\mathcal{O}(-1))=\mathcal{O}(-2)^2$. This tells us that the pushforward of the gauge field can be described by two fields twisted by $\mathcal{O}(-2)$ when we choose an appropriate basis for $\pi_*(\mathcal{L}_i\otimes\pi^*\mathcal{O}(-1))$. To form this basis, write
\begin{align*}
    H = \Pi_{i=1}^4(z-\xi_i)\quad\text{and define}\quad G_i = (z-\xi_i)(z-\xi_4).
\end{align*}
We then have the basis of quasi-periodic\footnote{To see that this basis describes the bundles $\mathcal{L}_i$, notice that two of the line bundles are spanned by functions that pick up a sign when crossing one branch cut and the third is spanned by functions that pick up a sign at both branch cuts. This behaviour is reflected in the choice of $G_i$.} functions $\{\sqrt{G_i},\sqrt{H/G_i}\}_{i=1,2,3}$. So, the components of our gauge field splits as
\begin{align*}
    \alpha_i\mapsto \phi_i\sqrt{G_i}+\tilde{\phi}_i\sqrt{H/G_i}
\end{align*}
where $i$ is the $SO(3)$ index and $\phi_i,\tilde{\phi}_i\in\Omega^{0,1}(\PT,\mathcal{O}(-2))$. Next, the $SO(3)$ indices behave as Pauli matrices\footnote{I.e. we only keep terms when all three indices are different.}, so the interaction terms all look like
\begin{align*}
    \phi_1\phi_2\phi_3\sqrt{G_1G_2G_3} & = \phi_1\phi_2\phi_3(z-\xi_4)\sqrt{H}\\
    \phi_1\phi_2\tilde{\phi}_3\sqrt{G_1G_2H/G_3} & = \phi_1\phi_2\tilde{\phi}_3(z-\xi_1)(z-\xi_2)(z-\xi_4)\\
    \tilde{\phi}_1\tilde{\phi}_2\tilde{\phi}_3\sqrt{H^3/G_1G_2G_3} & = \tilde{\phi}_1\tilde{\phi}_2\tilde{\phi}_3(z-\xi_1)(z-\xi_2)(z-\xi_3)\\
    \phi_1\tilde{\phi}_2\tilde{\phi}_3\sqrt{G_1H^2/G_2G_3} & = \phi_1\tilde{\phi}_2\tilde{\phi}_3(z-\xi_1)\sqrt{H}
\end{align*}
and terms coming from permuting the indices. So, the types of interactions we can write down are of the form 
\begin{align*}
    \int_{\PT}\Omega\wedge \left(\frac{\alpha_1\alpha_2\alpha_3}{\sqrt{H}}\right)_{\text{even}} = \int_{\PT} \Omega\wedge(\phi_1\phi_2\phi_3(z-\xi_4)+\phi_1\tilde{\phi}_2\tilde{\phi}_3(z-\xi_1))+\text{permutations}.
\end{align*}
The full twistor Lagrangian is thus:
\begin{align*}
    \int_{\mathbb{PT}}\Omega\wedge(\delta^{ij}\phi_i\bar{\partial}\tilde{\phi}_j+3\phi_1\phi_2\phi_3(z-\xi_4)+\epsilon^{ijk}\phi_i\tilde{\phi}_j\tilde{\phi}_k(z-\xi_i)).
\end{align*}
To leading order, the 4d action will be
\begin{align*}
    \int_{\mathbb{R}^4}\delta^{ij}\tilde{\Phi}_i\Delta\Phi_j+3\Phi_1d_{\gamma\dot{\alpha}}\Phi_2d_{\delta}^{\dot{\alpha}}\Phi_3T_4^\gamma T_\infty^\delta+\epsilon^{ijk}\Phi_id_{\gamma\dot{\alpha}}\tilde{\Phi}_jd_{\delta}^{\dot{\alpha}}\tilde{\Phi}_kT_i^\gamma T_\infty^\delta+\text{higher order terms}
\end{align*}
via the same symmetry arguments used to compactify in section \ref{sec:model_build}. Where 
\begin{align*}
T_l^\gamma = \begin{pmatrix}1 \\ \xi_l\end{pmatrix}\quad     T_\infty^\gamma = \begin{pmatrix}0\\1\end{pmatrix}\qquad\text{for $l = 1,...,4$.}
\end{align*}
The symmetry reducing this theory is trivial as we only have scalars. Abusing notation by denoting the 2d fields with the same symbol as the 4d fields, the leading order action is
\begin{align}\label{eq:2d_equiv2}
    \int_{\R^2}\tilde{\Phi}\Delta\Phi +3\Phi_1\partial_a\Phi_2\partial_b\Phi_3 G^{ab}+\epsilon^{ijk}\Phi_i\partial_a\tilde{\Phi}_j\partial_b\tilde{\Phi}_kF^{ab}_i+\dots
\end{align}
where
\begin{align*}
    G^{ab} = (\sigma^a)_{\alpha\dot{\alpha}}(\sigma^b)_\beta^{\dot{\alpha}}T_4^\alpha T_\infty^\beta\qquad F_i^{ab} = (\sigma^a)_{\alpha\dot{\alpha}}(\sigma^b)_\beta^{\dot{\alpha}}T_i^\alpha T_\infty^\beta\qquad \partial_a = (\sigma_{\alpha\dot{\alpha}})_ad^{\alpha\dot{\alpha}}.
\end{align*}

\subsection*{Comparing the 2d models}

Work of \cite{Bittleston:2020hfv} tells us that the 2d models given in equations \ref{eq:2d_equiv1} and \ref{eq:2d_equiv2} are equivalent. However, the presentations of the theories makes them appear quite different. The most important step in showing the equivalence would be finding the right field redefinition that matches what the degrees of freedom describe. 

Let me explain what the fields in the two presentations represent. In the $\sigma$-model action \ref{eq:2d_equiv1}, the fields $\sigma_\pm$ represent Wilson lines wrapping the torus. If we consider the two branches of $\Sigma$, there is a point at infinity on each branch and $\sigma_+$ is the Wilson line connecting the two points that passes through one branch cut while $\sigma_-$ is the line passing through the other cut. The second presentation as a deformed pair of free scalars describes the components of the equivariant 4d gauge field $\mathcal{A}$ that have certain quasi-periodicity properties as described by the line bundles $\mathcal{L}_i$. The field redefinition one needs to see the equivalence of these models will somehow map between these two ways of describing the degrees of freedom of the 4d gauge field.

Once the relationship between $\sigma_\pm$ and $\phi,\tilde{\phi}$ is understood, I expect that expanding the $\sigma$-model in terms of Lie algebra valued fields would match the action in \ref{eq:2d_equiv2} at leading order. While I hope the presentation of this example makes the task of verifying the diamond correspondence more concrete, such an analysis is beyond the scope of this paper. I also hope that this example makes explicit the connection between these 6d geometries and the interests of the wider 2d integrability community.

\section{Symmetry reduction for the genus 1 deformation}\label{sec:genus1_symred}

Recall the 4d action arising from the genus 1 deformation described in sections \ref{sec:model_build} and \ref{sec:genus1_bootstrap}
\begin{align*}
    \int_{\mathbb{R}^4}\Tr(BF(A)+\tilde{\Phi}D*D\Phi+B_{\alpha\beta}D_{\dot{\alpha}\delta}\Phi D^{\dot{\alpha}}_\gamma\Phi H^{\alpha\beta\delta\gamma}).
\end{align*}
I will symmetry reduce by the translation group generated by the vector fields given in equation \ref{eq:vec_f_sym}. Using the change of coordinates \ref{eq:s_t_coords}, we take our vector fields to be $\partial_{s_1},\partial_{s_2}$. Note that this setting differs from the genus 3 story since we are not studying a Chern-Simons theory on $\PT_\Sigma$ (meaning there is no 4d-Chern-Simons model we can look for to describe the symmetry reduced model), so there is no diamond correspondence of theories like the one in Fig. \ref{fig:diamond_theories}. This is still an interesting exercise as we expect to find some deformation of Hitchin's equations that describe the moduli of $G$-bundles on a genus 1 curve.

The first step in the reduction is then to demand the vanishing of
\begin{align*}
    \mathcal{L}_{V_i}B=\mathcal{L}_{V_i}A =\mathcal{L}_{V_i}\Phi=\mathcal{L}_{V_i}\tilde{\Phi} =0
\end{align*}
for $i = 1,2$. The treatment of the gauge fields is the same as in subsection \ref{subsec:sym_red} and I will use the same notation. The Lie-derivative vanishing condition on the scalars tells us that they are now independent of the $s_1,s_2$ directions. I will abuse notation slightly by leaving the symbol for the 2d scalars the same as the 4d scalars. 

We can then plug in the form of our fields back into the action. The $BF$ term reduces to the 2d action describing Hitchin's equation as described in subsection \ref{subsec:sym_red}. The scalar-scalar term can then be understood by first noting that
\begin{align*}
    D*D = D_T*D_T-[\phi_1,[\phi_1,\cdot]]-[\phi_2,[\phi_2,\cdot]].
\end{align*}
To understand the term involving $H$, we use the definition of the covariant derivative in equation \ref{eq:2d_cov_der} so that the vertex is given by
\begin{align*}
    \Tr(B_{\alpha\beta}D_{\dot{\alpha}\delta}\Phi D^{\dot{\alpha}}_\gamma\Phi H^{\alpha\beta\delta\gamma}) = \Tr(b_a\sigma^{i\dot{\alpha}}_\gamma \mathcal{D}_i\Phi\sigma^j_{\sigma\dot{\alpha}}\mathcal{D}_j\Phi\omega^a_{\alpha\beta}H^{\alpha\beta\gamma\delta}) = \Tr(b_a \mathcal{D}_i\Phi\mathcal{D}^i\Phi\tilde{H}^{a})
\end{align*}
using the definition $\tilde{H}^{a} = \omega^a_{\alpha\beta}\sigma^{i\dot{\alpha}}_\gamma\sigma_{i\sigma\dot{\alpha}}H^{\alpha\beta\gamma\delta}$. The full action is then
\begin{align*}
    \int_{T}\Tr&\left(
 b_1(F_T + [\phi_1,\phi_2])+
 b_2(D_{t_1}\phi_1 + D_{t_2}\phi_2)+
 b_3(D_{t_2}\phi_1 - D_{t_1}\phi_2\right)
\\
&\quad +\tilde{\Phi}(D_T*D_T-[\phi_1,[\phi_1,\Phi]]-[\phi_2,[\phi_2,\Phi]]) +b_a \mathcal{D}_i\Phi\mathcal{D}^i\Phi\tilde{H}^{a}+...).
\end{align*}

Where the higher-order vertices take the form
\begin{align*}
    \int_{\R^4}\Tr(B\Phi^{2n-2}D\Phi D\Phi H^n)\to \int_{\R^2}\Tr(b\Phi^{2n-2}\mathcal{D}\Phi\mathcal{D}\Phi \tilde{H}^n).
\end{align*}
Note the abuse of notation in writing $\tilde{H}$ above since the indices of these higher-order vertices may contract with $H$ in different ways.

\bibliographystyle{JHEP}

\bibliography{bib.bib}

@article{Costello:2013zra,
    author = "Costello, Kevin",
    title = "{Supersymmetric gauge theory and the Yangian}",
    eprint = "1303.2632",
    archivePrefix = "arXiv",
    primaryClass = "hep-th",
    month = "3",
    year = "2013"
}

@article{Costello:2013sla,
    author = "Costello, Kevin",
    editor = "Donagi, Ron and Douglas, Michael R. and Kamenova, Ljudmila and Rocek, Martin",
    title = "{Integrable lattice models from four-dimensional field theories}",
    eprint = "1308.0370",
    archivePrefix = "arXiv",
    primaryClass = "hep-th",
    doi = "10.1090/pspum/088/01483",
    journal = "Proc. Symp. Pure Math.",
    volume = "88",
    pages = "3--24",
    year = "2014"
}

@article{Costello:2017dso,
    author = "Costello, Kevin and Witten, Edward and Yamazaki, Masahito",
    title = "{Gauge Theory and Integrability, I}",
    eprint = "1709.09993",
    archivePrefix = "arXiv",
    primaryClass = "hep-th",
    reportNumber = "IPMU17-0136",
    doi = "10.4310/ICCM.2018.v6.n1.a6",
    journal = "ICCM Not.",
    volume = "06",
    number = "1",
    pages = "46--119",
    year = "2018"
}

@article{Costello:2018gyb,
    author = "Costello, Kevin and Witten, Edward and Yamazaki, Masahito",
    title = "{Gauge Theory and Integrability, II}",
    eprint = "1802.01579",
    archivePrefix = "arXiv",
    primaryClass = "hep-th",
    reportNumber = "IPMU18-0025",
    doi = "10.4310/ICCM.2018.v6.n1.a7",
    journal = "ICCM Not.",
    volume = "06",
    number = "1",
    pages = "120--146",
    year = "2018"
}

@article{Costello:2019tri,
    author = "Costello, Kevin and Yamazaki, Masahito",
    title = "{Gauge Theory And Integrability, III}",
    eprint = "1908.02289",
    archivePrefix = "arXiv",
    primaryClass = "hep-th",
    reportNumber = "IPMU19-0110",
    month = "8",
    year = "2019"
}

@article{Bittleston:2020hfv,
    author = "Bittleston, Roland and Skinner, David",
    title = "{Twistors, the ASD Yang-Mills equations and 4d Chern-Simons theory}",
    eprint = "2011.04638",
    archivePrefix = "arXiv",
    primaryClass = "hep-th",
    doi = "10.1007/JHEP02(2023)227",
    journal = "JHEP",
    volume = "02",
    pages = "227",
    year = "2023"
}

@article{Costello:2018zrm,
    author = "Costello, Kevin and Gaiotto, Davide",
    title = "{Twisted holography}",
    eprint = "1812.09257",
    archivePrefix = "arXiv",
    primaryClass = "hep-th",
    doi = "10.1007/JHEP01(2025)087",
    journal = "JHEP",
    volume = "01",
    pages = "087",
    year = "2025"
}

@article{Costello:2021bah,
    author = "Costello, Kevin J.",
    title = "{Quantizing local holomorphic field theories on twistor space}",
    eprint = "2111.08879",
    archivePrefix = "arXiv",
    primaryClass = "hep-th",
    month = "11",
    year = "2021"
}

@article{Costello:talk,
    author = "Costello, Kevin J.",
    title = "{Topological strings, twistors and Skyrmions}",
    eprint = "",
    archivePrefix = "",
    primaryClass = "hep-th",
    month = "",
    year = ""
}

@article{Lacroix:2023qlz,
    author = "Lacroix, Sylvain and Wallberg, Anders",
    title = "{An elliptic integrable deformation of the Principal Chiral Model}",
    eprint = "2311.09301",
    archivePrefix = "arXiv",
    primaryClass = "hep-th",
    reportNumber = "CERN-TH-2023-205",
    doi = "10.1007/JHEP05(2024)006",
    journal = "JHEP",
    volume = "05",
    pages = "006",
    year = "2024"
}

@article{Adamo:2017qyl,
    author = "Adamo, Tim",
    title = "{Lectures on twistor theory}",
    eprint = "1712.02196",
    archivePrefix = "arXiv",
    primaryClass = "hep-th",
    doi = "10.22323/1.323.0003",
    journal = "PoS",
    volume = "Modave2017",
    pages = "003",
    year = "2018"
}

@article{Costello:2022wso,
    author = "Costello, Kevin and Paquette, Natalie M.",
    title = "{Celestial holography meets twisted holography: 4d amplitudes from chiral correlators}",
    eprint = "2201.02595",
    archivePrefix = "arXiv",
    primaryClass = "hep-th",
    doi = "10.1007/JHEP10(2022)193",
    journal = "JHEP",
    volume = "10",
    pages = "193",
    year = "2022"
}

@article{Costello:2023hmi,
    author = "Costello, Kevin and Paquette, Natalie M. and Sharma, Atul",
    title = "{Burns space and holography}",
    eprint = "2306.00940",
    archivePrefix = "arXiv",
    primaryClass = "hep-th",
    doi = "10.1007/JHEP10(2023)174",
    journal = "JHEP",
    volume = "10",
    pages = "174",
    year = "2023"
}

@article{Penrose:1977in,
    author = "Penrose, R.",
    title = "{The Twistor Program}",
    doi = "10.1016/0034-4877(77)90047-7",
    journal = "Rept. Math. Phys.",
    volume = "12",
    pages = "65--76",
    year = "1977"
}

@article{Ward:1977ta,
    author = "Ward, R. S.",
    title = "{On Selfdual gauge fields}",
    doi = "10.1016/0375-9601(77)90842-8",
    journal = "Phys. Lett. A",
    volume = "61",
    pages = "81--82",
    year = "1977"
}

@article{Boels:2006ir,
    author = "Boels, Rutger and Mason, L. J. and Skinner, David",
    title = "{Supersymmetric Gauge Theories in Twistor Space}",
    eprint = "hep-th/0604040",
    archivePrefix = "arXiv",
    doi = "10.1088/1126-6708/2007/02/014",
    journal = "JHEP",
    volume = "02",
    pages = "014",
    year = "2007"
}

@article{Lacroix:2021iit,
    author = "Lacroix, Sylvain",
    title = "{Four-dimensional Chern\textendash{}Simons theory and integrable field theories}",
    eprint = "2109.14278",
    archivePrefix = "arXiv",
    primaryClass = "hep-th",
    doi = "10.1088/1751-8121/ac48ed",
    journal = "J. Phys. A",
    volume = "55",
    number = "8",
    pages = "083001",
    year = "2022"
}

@article{Bershadsky:1993cx,
    author = "Bershadsky, M. and Cecotti, S. and Ooguri, H. and Vafa, C.",
    title = "{Kodaira-Spencer theory of gravity and exact results for quantum string amplitudes}",
    eprint = "hep-th/9309140",
    archivePrefix = "arXiv",
    reportNumber = "HUTP-93-A025, RIMS-946, SISSA-142-93-EP",
    doi = "10.1007/BF02099774",
    journal = "Commun. Math. Phys.",
    volume = "165",
    pages = "311--428",
    year = "1994"
}

@article{Costello:2019jsy,
    author = "Costello, Kevin and Li, Si",
    title = "{Anomaly cancellation in the topological string}",
    eprint = "1905.09269",
    archivePrefix = "arXiv",
    primaryClass = "hep-th",
    doi = "10.4310/ATMP.2020.v24.n7.a2",
    journal = "Adv. Theor. Math. Phys.",
    volume = "24",
    number = "7",
    pages = "1723--1771",
    year = "2020"
}

@article{Donaldson:1985zz,
    author = "Donaldson, S. K.",
    title = "{ANTI SELF-DUAL YANG-MILLS CONNECTIONS OVER COMPLEX ALGEBRAIC SURFACES AND STABLE VECTOR BUNDLES}",
    doi = "10.1112/plms/s3-50.1.1",
    journal = "Proc. Lond. Math. Soc.",
    volume = "50",
    pages = "1--26",
    year = "1985"
}

@article{Guevara:2021abz,
    author = "Guevara, Alfredo and Himwich, Elizabeth and Pate, Monica and Strominger, Andrew",
    title = "{Holographic symmetry algebras for gauge theory and gravity}",
    eprint = "2103.03961",
    archivePrefix = "arXiv",
    primaryClass = "hep-th",
    doi = "10.1007/JHEP11(2021)152",
    journal = "JHEP",
    volume = "11",
    pages = "152",
    year = "2021"
}

@article{Strominger:2021mtt,
    author = "Strominger, Andrew",
    title = "{$w_{1+\infty}$ Algebra and the Celestial Sphere: Infinite Towers of Soft Graviton, Photon, and Gluon Symmetries}",
    eprint = "2105.14346",
    archivePrefix = "arXiv",
    primaryClass = "hep-th",
    doi = "10.1103/PhysRevLett.127.221601",
    journal = "Phys. Rev. Lett.",
    volume = "127",
    number = "22",
    pages = "221601",
    year = "2021"
}

@article{Adamo:2021lrv,
    author = "Adamo, Tim and Mason, Lionel and Sharma, Atul",
    title = "{Celestial $w_{1+\infty}$ Symmetries from Twistor Space}",
    eprint = "2110.06066",
    archivePrefix = "arXiv",
    primaryClass = "hep-th",
    doi = "10.3842/SIGMA.2022.016",
    journal = "SIGMA",
    volume = "18",
    pages = "016",
    year = "2022"
}

@article{Bittleston:2023bzp,
    author = "Bittleston, Roland and Heuveline, Simon and Skinner, David",
    title = "{The celestial chiral algebra of self-dual gravity on Eguchi-Hanson space}",
    eprint = "2305.09451",
    archivePrefix = "arXiv",
    primaryClass = "hep-th",
    doi = "10.1007/JHEP09(2023)008",
    journal = "JHEP",
    volume = "09",
    pages = "008",
    year = "2023"
}

@article{Costello:2022jpg,
    author = "Costello, Kevin and Paquette, Natalie M. and Sharma, Atul",
    title = "{Top-Down Holography in an Asymptotically Flat Spacetime}",
    eprint = "2208.14233",
    archivePrefix = "arXiv",
    primaryClass = "hep-th",
    doi = "10.1103/PhysRevLett.130.061602",
    journal = "Phys. Rev. Lett.",
    volume = "130",
    number = "6",
    pages = "061602",
    year = "2023"
}

@article{Bittleston:2024efo,
    author = "Bittleston, Roland and Costello, Kevin and Zeng, Keyou",
    title = "{Self-Dual Gauge Theory from the Top Down}",
    eprint = "2412.02680",
    archivePrefix = "arXiv",
    primaryClass = "hep-th",
    month = "12",
    year = "2024"
}

@article{Costello:2022upu,
    author = "Costello, Kevin and Paquette, Natalie M.",
    title = "{Associativity of One-Loop Corrections to the Celestial Operator Product Expansion}",
    eprint = "2204.05301",
    archivePrefix = "arXiv",
    primaryClass = "hep-th",
    doi = "10.1103/PhysRevLett.129.231604",
    journal = "Phys. Rev. Lett.",
    volume = "129",
    number = "23",
    pages = "231604",
    year = "2022"
}

@article{Atanasov:2021oyu,
    author = "Atanasov, Alexander and Ball, Adam and Melton, Walker and Raclariu, Ana-Maria and Strominger, Andrew",
    title = "{(2, 2) Scattering and the celestial torus}",
    eprint = "2101.09591",
    archivePrefix = "arXiv",
    primaryClass = "hep-th",
    doi = "10.1007/JHEP07(2021)083",
    journal = "JHEP",
    volume = "07",
    pages = "083",
    year = "2021"
}

@article{Penrose:1976jq,
    author = "Penrose, R.",
    title = "{The Nonlinear Graviton}",
    doi = "10.1007/BF00763433",
    journal = "Gen. Rel. Grav.",
    volume = "7",
    pages = "171--176",
    year = "1976"
}

@article{Parke:1986gb,
    author = "Parke, Stephen J. and Taylor, T. R.",
    title = "{An Amplitude for $n$ Gluon Scattering}",
    reportNumber = "FERMILAB-PUB-86-042-T",
    doi = "10.1103/PhysRevLett.56.2459",
    journal = "Phys. Rev. Lett.",
    volume = "56",
    pages = "2459",
    year = "1986"
}

@article{Witten:2003nn,
    author = "Witten, Edward",
    title = "{Perturbative gauge theory as a string theory in twistor space}",
    eprint = "hep-th/0312171",
    archivePrefix = "arXiv",
    doi = "10.1007/s00220-004-1187-3",
    journal = "Commun. Math. Phys.",
    volume = "252",
    pages = "189--258",
    year = "2004"
}

@article{Mason:2005zm,
    author = "Mason, L. J.",
    title = "{Twistor actions for non-self-dual fields: A Derivation of twistor-string theory}",
    eprint = "hep-th/0507269",
    archivePrefix = "arXiv",
    doi = "10.1088/1126-6708/2005/10/009",
    journal = "JHEP",
    volume = "10",
    pages = "009",
    year = "2005"
}

@article{Jarov,
    author = "Jarov, Seraphim",
    title = "",
    eprint = "XXXX.XXXX",
    archivePrefix = "arXiv",
    primaryClass = "hep-th",
    month = "xx",
    year = "xxxx"
}

@article{Hitchin:1986vp,
    author = "Hitchin, Nigel J.",
    title = "{The Selfduality equations on a Riemann surface}",
    doi = "10.1112/plms/s3-55.1.59",
    journal = "Proc. Lond. Math. Soc.",
    volume = "55",
    pages = "59--131",
    year = "1987"
}

@book{Mason:1991rf,
    author = "Mason, L. J. and Woodhouse, N. M. J.",
    title = "{Integrability, selfduality, and twistor theory}",
    year = "1991"
}

@misc{derryberry2021laxformulationharmonicmaps,
      title={Lax formulation for harmonic maps to a moduli of bundles}, 
      author={Richard Derryberry},
      year={2021},
      eprint={2106.09781},
      archivePrefix={arXiv},
      primaryClass={math.AG},
      url={https://arxiv.org/abs/2106.09781}, 
}

@article{Cole_2024,
   title={Gauging the diamond: integrable coset models from twistor space},
   volume={2024},
   ISSN={1029-8479},
   url={http://dx.doi.org/10.1007/JHEP12(2024)202},
   DOI={10.1007/jhep12(2024)202},
   number={12},
   journal={Journal of High Energy Physics},
   publisher={Springer Science and Business Media LLC},
   author={Cole, Lewis T. and Cullinan, Ryan A. and Hoare, Ben and Liniado, Joaquin and Thompson, Daniel C.},
   year={2024},
   month=dec }

@article{Himwich:2021dau,
    author = "Himwich, Elizabeth and Pate, Monica and Singh, Kyle",
    title = "{Celestial operator product expansions and w$_{1+\infty}$ symmetry for all spins}",
    eprint = "2108.07763",
    archivePrefix = "arXiv",
    primaryClass = "hep-th",
    doi = "10.1007/JHEP01(2022)080",
    journal = "JHEP",
    volume = "01",
    pages = "080",
    year = "2022"
}

@article{Vicedo_2022,
   title={3-Dimensional mixed BF theory and Hitchin’s integrable system},
   volume={112},
   ISSN={1573-0530},
   url={http://dx.doi.org/10.1007/s11005-022-01567-6},
   DOI={10.1007/s11005-022-01567-6},
   number={4},
   journal={Letters in Mathematical Physics},
   publisher={Springer Science and Business Media LLC},
   author={Vicedo, Benoît and Winstone, Jennifer},
   year={2022},
   month=aug }

@article{Gabai:2018tmm,
    author = "Gabai, Barak and Maz{\'a}{\v{c}}, Dalimil and Shieber, Andrei and Vieira, Pedro and Zhou, Yehao",
    title = "{No Particle Production in Two Dimensions: Recursion Relations and Multi-Regge Limit}",
    eprint = "1803.03578",
    archivePrefix = "arXiv",
    primaryClass = "hep-th",
    doi = "10.1007/JHEP02(2019)094",
    journal = "JHEP",
    volume = "02",
    pages = "094",
    year = "2019"
}

@article{Ren:2022sws,
    author = "Ren, Lecheng and Spradlin, Marcus and Yelleshpur Srikant, Akshay and Volovich, Anastasia",
    title = "{On effective field theories with celestial duals}",
    eprint = "2206.08322",
    archivePrefix = "arXiv",
    primaryClass = "hep-th",
    doi = "10.1007/JHEP08(2022)251",
    journal = "JHEP",
    volume = "08",
    pages = "251",
    year = "2022"
}

@incollection{FeiginFrenkel2011,
  author    = {Boris Feigin and Edward Frenkel},
  title     = {Quantization of Soliton Systems and Langlands Duality},
  booktitle = {Exploring New Structures and Natural Constructions in Mathematical Physics},
  editor     = {K. Hasegawa and T. Hayashi and S. Hosono and Y. Yamada},
  series     = {Advanced Studies in Pure Mathematics},
  volume     = {61},
  pages      = {185--274},
  year       = {2011},
  publisher  = {Mathematical Society of Japan},
  doi        = {10.2969/aspm/06110185}
}

@article{10.1093/imrn/rny128,
    author = {Vicedo, Benoît},
    title = {On Integrable Field Theories as Dihedral Affine Gaudin Models},
    journal = {International Mathematics Research Notices},
    volume = {2020},
    number = {15},
    pages = {4513-4601},
    year = {2020},
    month = {08},
    issn = {1073-7928},
    doi = {10.1093/imrn/rny128},
    url = {https://doi.org/10.1093/imrn/rny128},
    eprint = {https://academic.oup.com/imrn/article-pdf/2020/15/4513/33559615/rny128.pdf},
}

@article{Vicedo2021,
  author       = {Beno{\^\i}t Vicedo},
  title        = {4d Chern--Simons Theory and Affine Gaudin Models},
  journal      = {Letters in Mathematical Physics},
  volume       = {111},
  number       = {1},
  pages         = {24},
  year         = {2021},
  doi          = {10.1007/s11005-021-01354-9},
  publisher    = {Springer}
}

@article{Bittleston:2026tdr,
    author = "Bittleston, Roland and Mason, Lionel and Moosavian, Seyed Faroogh",
    title = "{The Self-Duality Equations on a Riemann Surface and Four-Dimensional Chern-Simons Theory}",
    eprint = "2601.05309",
    archivePrefix = "arXiv",
    primaryClass = "hep-th",
    month = "1",
    year = "2026"
}

\appendix
\numberwithin{equation}{section}

\section{Line bundle computations}\label{app:line_bdl}

Fix an elliptic genus 1 curve $\pi:\Sigma\to\CP^1$. This appendix is concerned with computing 
\begin{align*}
    \pi_*K_\Sigma\quad \pi_*\mathcal{O}_\Sigma
\end{align*}
where $K_\Sigma$ and $\mathcal{O}_\Sigma$ are the canonical and trivial bundle of an elliptic curve $\Sigma$ respectively. Note that this is an easy algebraic geometry exercise that I have repeated here in an effort to make the paper self-contained.

First, we note that both will be rank 2 vector bundles over $\CP^1$ as $\pi$ is a 2 to 1 map. This means that
\begin{align*}
    \pi_*K_\Sigma = \mathcal{O}_{\CP^1}(n)\oplus\mathcal{O}_{\CP^1}(m)\quad \pi_*\mathcal{O}_\Sigma=\mathcal{O}_{\CP^1}(\tilde{n})\oplus\mathcal{O}_{\CP^1}(\tilde{m})
\end{align*}
for some integers $n,m,\tilde{n},\tilde{m}$, as any vector bundle over the Riemann sphere is a sum of line bundles. Let's focus on the pushforward of the trivial bundle over $\Sigma$. The first component in the direct sum can easily be seen to be $\mathcal{O}_{\CP^1}$ (ie. $\tilde{n}=0$) as these correspond to functions that are unchanged under changing branches. As for $\tilde{m}$, we count dimensions
\begin{align*}
    g_\Sigma=\dim H^1(\Sigma,\mathcal{O}_\Sigma) = \dim H^1(\CP^1,\mathcal{O}_{\CP^1}\oplus\mathcal{O}_{\CP^1}(\tilde{m})) = -\tilde{m}-1.
\end{align*}
Where $g_\Sigma$ is the genus of $\Sigma$, and we have used the following facts
\begin{align*}
    H^1(M,A\oplus B)=H^1(M,A)\oplus H^1(M,B)\quad \text{and}\quad\dim H^1(\CP^1,\mathcal{O}_{\CP^1}(n))=\begin{cases}
        0\quad&\text{if $n>-2$}\\
        -n-1\quad&\text{else}.
    \end{cases}
\end{align*}
Plugging in $g_\Sigma=1$, we get that $\tilde{m}=-2$. We thus conclude that $\pi_*\mathcal{O}_\Sigma=\mathcal{O}_{\CP^1}\oplus\mathcal{O}_{\CP^1}(-2)$. As for the pushforward of the canonical bundle, we first note that 
\begin{align*}
    \dim H^0(\Sigma,K_\Sigma) = g_\Sigma\quad \text{and}\quad\dim H^1(\Sigma,K_\Sigma) = 1.
\end{align*}
This, along with the formula
\begin{align*}
    \dim H^0(\CP^1,\mathcal{O}_{\CP^1}(n)) = \begin{cases}
        0\quad&\text{if $n<0$}\\
        n+1\quad&\text{else}
    \end{cases}
\end{align*}
gives us the following relation on the integers defining the pushforward of the canonical bundle:
\begin{align*}
    1 = g_\Sigma = \begin{cases}
        n+1\quad&\text{if $n\geq0$ and $m<0$}\\
        m+1\quad&\text{if $m\geq0$ and $n<0$}\\
        n+m+2\quad&\text{if $n,m\geq0$}\\
        0\quad&\text{else}
    \end{cases}\quad1 = \begin{cases}
        -n-1\quad&\text{if $n\leq-2$ and $m>-2$}\\
        -m-1\quad&\text{if $m\leq-2$ and $n>-2$}\\
        -m-n-2\quad&\text{if $n,m\leq-2$}\\
        0\quad&\text{else}.
    \end{cases}
\end{align*}
The choice that satisfies these relations is $n=-2$ and $m=0$. This means that
\begin{align*}
    \pi_*K_\Sigma = \mathcal{O}_{\CP^1}(-2)\oplus\mathcal{O}_{\CP^1}\quad \pi_*\mathcal{O}_\Sigma=\mathcal{O}_{\CP^1}\oplus\mathcal{O}_{\CP^1}(-2).
\end{align*}

\section{OPE and form factor calculations}\label{app:OPE}

This appendix explicitly computes some of the OPEs appearing in section \ref{sec:celestial}.

\subsection{OPEs for the genus 3 deformation}

This section computes OPEs for the states given in equation \ref{eq:genus3_OPE}. Here is the OPE:
\begin{align*}
    J_{\pm,a}(\tilde{\lambda}_1;z_1)J_{\pm,b}(\tilde{\lambda}_2;z_2) & \sim 2\frac{f_{ab}^c}{\langle12\rangle}J_c(\tilde{\lambda}_1+\tilde{\lambda}_2;z_2)\pm 2\frac{1}{\sqrt{H(z_2)}}\frac{f_{ab}^c}{\langle12\rangle}\tilde{J}_c(\tilde{\lambda}_1+\tilde{\lambda}_2;z_2)\\
    &=2\frac{f_{ab}^c}{\langle12\rangle}J_{\pm,c}(\tilde{\lambda}_1+\tilde{\lambda}_2;z_2)\\
    J_{\pm,a}(\tilde{\lambda}_1;z_1)J_{\mp,b}(\tilde{\lambda}_2;z_2) & \sim \frac{f_{ab}^c}{\langle12\rangle}J_c(\tilde{\lambda}_1+\tilde{\lambda}_2;z_2)-\frac{f_{ab}^c}{\langle12\rangle}J_c(\tilde{\lambda}_1+\tilde{\lambda}_2;z_2)\pm \frac{1}{\sqrt{H(z_2)}}\frac{f_{ab}^c}{\langle12\rangle}\tilde{J}_c(\tilde{\lambda}_1+\tilde{\lambda}_2;z_2)\\
    &\quad \mp \frac{1}{\sqrt{H(z_2)}}\frac{f_{ab}^c}{\langle12\rangle}\tilde{J}_c(\tilde{\lambda}_1+\tilde{\lambda}_2;z_2)\\
    & = 0.
\end{align*}

\subsection{OPEs for the genus 1 deformation}

Note that there is a slight abuse in notation as the states of this theory share some notation with the ones in the previous part of this appendix. We compute the OPEs of the states defined in equation \ref{eq:chiral_alg_states}:
\begin{align*}
    J_{\pm,a}(\tilde{\lambda}_1;z_1)J_{\pm,b}(\tilde{\lambda}_2;z_2) =& J_{a}(\tilde{\lambda}_1;z_1)J_{b}(\tilde{\lambda}_2;z_2)+\frac{1}{H}\mathcal{O}_{a}(\tilde{\lambda}_1;z_1)\mathcal{O}_{b}(\tilde{\lambda}_2;z_2)\\
    &\quad \pm\frac{1}{\sqrt{H}}(J_{a}(\tilde{\lambda}_1;z_1)\mathcal{O}_{b}(\tilde{\lambda}_2;z_2)+\mathcal{O}_{a}(\tilde{\lambda}_1;z_1)J_{b}(\tilde{\lambda}_2;z_2))\\
    \sim & \frac{2}{\langle12\rangle}f_{ab}^c\left(J_{c}(\tilde{\lambda}_2;z_2))\pm\frac{1}{\sqrt{H}}\mathcal{O}_{c}(\tilde{\lambda}_2;z_2))\right) = \frac{2}{\langle12\rangle}f_{ab}^cJ_{\pm,c}(\tilde{\lambda}_1+\tilde{\lambda}_2;z_2)\\
    J_{+,a}(\tilde{\lambda}_1;z_1)J_{-,b}(\tilde{\lambda}_2;z_2) =& J_{a}(\tilde{\lambda}_1;z_1)J_{b}(\tilde{\lambda}_2;z_2)-\frac{1}{H}\mathcal{O}_{a}(\tilde{\lambda}_1;z_1)\mathcal{O}_{b}(\tilde{\lambda}_2;z_2)\\
    &\quad \pm\frac{1}{\sqrt{H}}(J_{a}(\tilde{\lambda}_1;z_1)\mathcal{O}_{b}(\tilde{\lambda}_2;z_2)-\mathcal{O}_{a}(\tilde{\lambda}_1;z_1)J_{b}(\tilde{\lambda}_2;z_2))
    \sim 0\\
    J_{\pm,a}(\tilde{\lambda}_1;z_1)\tilde{\mathcal{O}}_{\pm,b}(\tilde{\lambda}_2;z_2) = & J_{a}(\tilde{\lambda}_1;z_1)\tilde{\mathcal{O}}_{b}(\tilde{\lambda}_2;z_2)+\frac{1}{H}\mathcal{O}_{a}(\tilde{\lambda}_1;z_1)\tilde{J}_{b}(\tilde{\lambda}_2;z_2)\\
    &\quad \pm\frac{1}{\sqrt{H}}(\mathcal{O}_{a}(\tilde{\lambda}_1;z_1)\tilde{\mathcal{O}}_{b}(\tilde{\lambda}_2;z_2)+J_{a}(\tilde{\lambda}_1;z_1)\tilde{J}_{b}(\tilde{\lambda}_2;z_2))\\
    \sim&  \frac{2}{\langle12\rangle}f_{ab}^c\left(\tilde{\mathcal{O}}_{c}(\tilde{\lambda}_2;z_2))\pm\frac{1}{\sqrt{H}}\tilde{J}_{c}(\tilde{\lambda}_2;z_2))\right) = \frac{2}{\langle12\rangle}f_{ab}^c\tilde{\mathcal{O}}_{\pm,c}(\tilde{\lambda}_1+\tilde{\lambda}_2;z_2)\\
    J_{+,a}(\tilde{\lambda}_1;z_1)\tilde{\mathcal{O}}_{-,b}(\tilde{\lambda}_2;z_2) = &J_{a}(\tilde{\lambda}_1;z_1)\tilde{\mathcal{O}}_{b}(\tilde{\lambda}_2;z_2)-\frac{1}{H}\mathcal{O}_{a}(\tilde{\lambda}_1;z_1)\tilde{J}_{b}(\tilde{\lambda}_2;z_2)\\
    &\quad \pm\frac{1}{\sqrt{H}}(\mathcal{O}_{a}(\tilde{\lambda}_1;z_1)\tilde{\mathcal{O}}_{b}(\tilde{\lambda}_2;z_2)-J_{a}(\tilde{\lambda}_1;z_1)\tilde{J}_{b}(\tilde{\lambda}_2;z_2))\sim0.
\end{align*}

\subsection{Form factors}\label{app:form_factor}

Some of the form factors of sections \ref{sec:bootstrap} and \ref{sec:genus1_bootstrap} are computed here.

\subsection*{$n$-point, 3-minus in the genus 3 deformation}

\begin{align*}
    \langle B_{\alpha\beta}&D_\gamma^{\dot{\alpha}}B_{\delta\rho}D_{\sigma\dot{\alpha}}B_{\zeta\eta}H^{\alpha\beta\gamma\delta\rho\sigma\zeta\eta}|\tilde{J}[1](z_1)\tilde{J}[1](z_2)\tilde{J}(z_3)J(z_4)...J(z_n)\rangle\\
    & = \frac{\langle31\rangle\langle B_{\alpha\beta}D_\gamma^{\dot{\alpha}}B_{\delta\rho}D_{\sigma\dot{\alpha}}B_{\zeta\eta}H^{\alpha\beta\gamma\delta\rho\sigma\zeta\eta}|\tilde{J}[1](z_1)\tilde{J}[1](z_2)\tilde{J}(z_3)\rangle}{\langle34\rangle...\langle n-1,n\rangle\langle n1\rangle} = \frac{\langle31\rangle[12]\langle1^32^33^2|H\rangle}{\langle34\rangle...\langle n-1,n\rangle\langle n1\rangle}\\
    \langle B_{\alpha\beta}&D_\gamma^{\dot{\alpha}}B_{\delta\rho}D_{\sigma\dot{\alpha}}B_{\zeta\eta}H^{\alpha\beta\gamma\delta\rho\sigma\zeta\eta}|\tilde{J}[1](z_1)\tilde{J}(z_2)\tilde{J}[1](z_3)J(z_4)...J(z_n)\rangle\\
    & = \frac{\langle31\rangle[31]\langle1^32^23^3|H\rangle}{\langle34\rangle...\langle n-1,n\rangle\langle n1\rangle}\\
    \langle B_{\alpha\beta}&D_\gamma^{\dot{\alpha}}B_{\delta\rho}D_{\sigma\dot{\alpha}}B_{\zeta\eta}H^{\alpha\beta\gamma\delta\rho\sigma\zeta\eta}|\tilde{J}[1](z_1)\tilde{J}(z_2)\tilde{J}(z_3)J(z_4)...J[1](z_i)...J(z_n)\rangle\\
    & = \frac{\langle3i\rangle\langle i1\rangle\langle B_{\alpha\beta}D_\gamma^{\dot{\alpha}}B_{\delta\rho}D_{\sigma\dot{\alpha}}B_{\zeta\eta}H^{\alpha\beta\gamma\delta\rho\sigma\zeta\eta}|\tilde{J}[1](z_1)\tilde{J}(z_2)\tilde{J}(z_3)J[1](z_i)\rangle}{\langle 34\rangle...\langle i-1,i\rangle\langle i,i+1\rangle...\langle n-1,n\rangle\langle n1\rangle}\\
    & = \frac{\langle i1\rangle[i1]\langle1^32^23^3|H\rangle}{\langle 34\rangle...\langle i-1,i\rangle\langle i,i+1\rangle...\langle n-1,n\rangle\langle n1\rangle}\\
    \langle B_{\alpha\beta}&D_\gamma^{\dot{\alpha}}B_{\delta\rho}D_{\sigma\dot{\alpha}}B_{\zeta\eta}H^{\alpha\beta\gamma\delta\rho\sigma\zeta\eta}|\tilde{J}(z_1)\tilde{J}[1](z_2)\tilde{J}[1](z_3)J(z_4)...J(z_n)\rangle\\
    & = \frac{[23]\langle31\rangle\langle1^22^33^3|H\rangle}{\langle34\rangle...\langle n1\rangle}\\
    \langle B_{\alpha\beta}&D_\gamma^{\dot{\alpha}}B_{\delta\rho}D_{\sigma\dot{\alpha}}B_{\zeta\eta}H^{\alpha\beta\gamma\delta\rho\sigma\zeta\eta}|\tilde{J}(z_1)\tilde{J}[1](z_2)\tilde{J}(z_3)J(z_4)...J[1](z_i)...J(z_n)\rangle\\
    & = \frac{[2i]\langle i1\rangle\langle1^22^33^3|H\rangle}{\langle34\rangle...\langle n1\rangle}+\frac{[i2]\langle 3i\rangle\langle1^32^33^2|H\rangle}{\langle34\rangle...\langle n1\rangle}\\
    \langle B_{\alpha\beta}&D_\gamma^{\dot{\alpha}}B_{\delta\rho}D_{\sigma\dot{\alpha}}B_{\zeta\eta}H^{\alpha\beta\gamma\delta\rho\sigma\zeta\eta}|\tilde{J}(z_1)\tilde{J}(z_2)\tilde{J}[1](z_3)J(z_4)...J[1](z_i)...J(z_n)\rangle\\
    & = \frac{[3i]\langle 3i\rangle\langle1^32^23^3|H\rangle}{\langle34\rangle...\langle n1\rangle}\\
    \langle B_{\alpha\beta}&D_\gamma^{\dot{\alpha}}B_{\delta\rho}D_{\sigma\dot{\alpha}}B_{\zeta\eta}H^{\alpha\beta\gamma\delta\rho\sigma\zeta\eta}|\tilde{J}(z_1)\tilde{J}(z_2)\tilde{J}(z_3)J(z_4)...J[1](z_i)...J[1](z_j)...J(z_n)\rangle\\
    & = \frac{[ij]\langle ij\rangle\langle1^32^23^3|H\rangle}{\langle34\rangle...\langle n1\rangle}.
\end{align*}

\section{The 3-minus, 1-plus form factor via Feynman diagrams}\label{app:4pt_ff}

This appendix provides a sanity check for the chiral algebra computation of the 4-point form factor done in subsection \ref{subsec:form_factor}. To make the comparison, I just need to find the colour-ordered term in the form factor that closely follows the computation of the 4-point amplitude. The only difference is that we must be careful about moving the derivatives in the vertex on each of the factors of $B$ explicitly, since we cannot apply momentum conservation. With that in mind, the diagrams we need to account for are the $s$ and $t$-channel diagrams shown in Fig. \ref{fig:channel_diagrams} and the contact term in Fig. \ref{fig:contact_term} (but with different terms from each way the derivatives can hit the legs). I will use the same propagator, vertex, and external leg rules I introduced in section \ref{sec:bootstrap}. In particular, we introduce a reference spinor $q^\alpha$ to gauge fix the positive helicity gluons.

The $t$-channel diagram contributes the following terms to the form factor
\begin{align*}
    \frac{\langle q3\rangle[21]}{\langle 4q\rangle\langle43\rangle}\langle1^32^33^2|H\rangle+\frac{\langle q3\rangle}{\langle4q\rangle\langle43\rangle}([32]\langle1^22^33^3|H\rangle+[42]\langle1^22^33^24|H\rangle)\\+\frac{\langle q3\rangle}{\langle4q\rangle\langle43\rangle}([13]\langle1^32^23^3|H\rangle +[14]\langle1^32^23^24|H\rangle)
\end{align*}
while the $s$-channel contributes
\begin{align*}
    \frac{\langle q1\rangle[23]}{\langle4q\rangle\langle41\rangle}\langle1^22^33^3|H\rangle+\frac{\langle q1\rangle}{\langle4q\rangle\langle41\rangle}([31]\langle1^32^23^3|H\rangle+[34]\langle1^22^23^34|H\rangle)\\
    +\frac{\langle q1\rangle}{\langle4q\rangle\langle41\rangle}([12]\langle1^32^33^2|H\rangle+[42]\langle1^22^33^24|H\rangle).
\end{align*}

For each channel diagram, we see three terms corresponding to the position of the derivative on the three copies of $B$ in the $BDBDBH$ vertex.

Adding terms between the $t$ and $s$-channel contributions that share the same contractions with $H$ (eg. we match terms that have a factor of $\langle1^32^33^2|H\rangle$ and terms with a factor of $\langle1^32^23^3|H\rangle$, etc.), along with repeated use of the Schouten identity to remove the gauge fixing spinor $q$, gives us
\begin{align*}
    \frac{[12]\langle13\rangle}{\langle41\rangle\langle34\rangle}\langle1^32^33^2|H\rangle+&\frac{[23]\langle13\rangle}{\langle41\rangle\langle34\rangle}\langle1^22^33^3|H\rangle+\frac{[31]\langle13\rangle}{\langle41\rangle\langle34\rangle}\langle1^32^23^3|H\rangle\\
    &+\frac{[24]}{\langle41\rangle}\langle1^32^33^2|H\rangle+\frac{[24]}{\langle43\rangle}\langle1^22^33^3|H\rangle+\frac{[41]}{\langle43\rangle}\langle1^32^23^3|H\rangle+\frac{[43]}{\langle41\rangle}\langle1^32^23^3|H\rangle.
\end{align*}

\section{Elliptic deformed PCM computations}\label{app:equiv_4dCS}

The authors of \cite{Lacroix:2023qlz} show that the 2d model will have the following action:
\begin{align}\label{eq:2d_sig_model}
    S_{2d}[\sigma_{\pm}] = \frac{1}{2}\sum_{b=\pm1}\left(\int_{\mathbb{R}^2}\langle \Res_{\gamma = bp}(\omega\wedge\mathcal{L}),\sigma^{-1}_{\text{sgn}(b)}d_{\mathbb{R}^2}\sigma_{\text{sgn}(b)}\rangle-\Res_{\gamma=bp}(\omega)\mathcal{I}_{WZ}[\sigma_{\text{sgn}(b)}]\right).
\end{align}
Where the Wess-Zumino term is
\begin{align*}
    \mathcal{I}_{WZ}[\sigma_{\pm}] = \frac{1}{3}\int_{M}\langle \sigma_\pm^{-1}d_M\sigma_{\pm},\sigma_\pm^{-1}d_M\sigma_{\pm}\wedge\sigma_\pm^{-1}d_M\sigma_{\pm}\rangle
\end{align*}
for $M$ a 3-manifold with boundary $\partial M = \mathbb{R}^2$ and $\sigma_\pm:M\to SU(2)$ is an extension of the original field that agrees on the boundary of $M$.

We now figure out what this is explicitly. First, we compute
\begin{align*}
    \Res_{\gamma = p}(\omega) = & \frac{1}{2}\lim_{\gamma\to p}\frac{d}{d\gamma}(\gamma-p)^2\frac{(c_1+\wp(\gamma)c_2)((c_3+\wp(\gamma)c_4))}{(\wp(\gamma)-\wp(p))^2}\\
    = & \frac{1}{2}\lim_{\gamma\to p}(\wp'(\gamma)c_2(c_3+\wp(\gamma)c_4)+\wp'(\gamma)c_4(c_1+\wp(\gamma)c_2))\frac{(\gamma-p)^2}{(\wp(\gamma)-\wp(p))^2}\\
    & + \frac{1}{2}\lim_{\gamma\to p}(c_1+\wp(\gamma)c_2)(c_3+\wp(\gamma)c_4)\frac{d}{d\gamma}\frac{(\gamma-p)^2}{(\wp(\gamma)-\wp(p))^2}\\
    = & \frac{c_2(c_3+\wp(p)c_4)+c_4(c_1+\wp(p)c_2)}{2\wp'(p)}-\frac{\wp''(p)(c_1+\wp(p)c_2)(c_3+\wp(p)c_4)}{4(\wp'(p))^3}
\end{align*}
and similarly,
\begin{align*}
     \Res_{\gamma = -p}(\omega) = -\frac{c_2(c_3+\wp(p)c_4)+c_4(c_1+\wp(p)c_2)}{2\wp'(p)}+\frac{\wp''(p)(c_1+\wp(p)c_2)(c_3+\wp(p)c_4)}{4(\wp'(p))^3}
\end{align*}
using the fact that even derivatives of $\wp$ are even functions and odd derivatives are odd. So, we see that the Wess-Zumino term in the action does not vanish and we get one term for each field $\sigma_\pm$ weighted by the residues we just computed.

Next, we need to find an expression for the fields $U_{\dot{\alpha},\pm}^a$ in terms of the currents. This amounts to inverting the matrix
\begin{align*}
    M^a = \begin{pmatrix}
        r^a(p-\hat{\gamma}_{\dot{\alpha}}) & r^a(p+\hat{\gamma}_{\dot{\alpha}})\\
        r^a(-p-\hat{\gamma}_{\dot{\alpha}}) & r^a(-p+\hat{\gamma}_{\dot{\alpha}})    \end{pmatrix}\xrightarrow[]{(\cdot)^{-1}}\frac{1}{\det M^a}\begin{pmatrix}
             r^a(-p+\hat{\gamma}_{\dot{\alpha}}) & -r^a(p+\hat{\gamma}_{\dot{\alpha}})\\
        -r^a(-p-\hat{\gamma}_{\dot{\alpha}}) &   r^a(p-\hat{\gamma}_{\dot{\alpha}})
        \end{pmatrix}
\end{align*}
where
\begin{align}\label{eq:det}
    \det M^a & = r^a(p-\hat{\gamma}_{\dot{\alpha}})r^a(-p+\hat{\gamma}_{\dot{\alpha}})-r^a(p+\hat{\gamma}_{\dot{\alpha}})r^a(-p-\hat{\gamma}_{\dot{\alpha}})\\
    & = -r^a(p-\hat{\gamma}_{\dot{\alpha}})^2+r^a(p+\hat{\gamma}_{\dot{\alpha}})^2
\end{align}
using the fact that $r^a(\gamma)$ is an odd function\footnote{In \cite{Lacroix:2023qlz}, they use the convention that $a\in\mathbb{Z}_N\times\mathbb{Z}_N\backslash\{0,0\}$, and use the fact that $r^{-a}(\gamma) = -r^a(-\gamma)$ and $r^{-a'}(\gamma) = r^{a'}(-\gamma)$. In our $SU(2)$ case, we have that $a=-a$, so we simply have that $r^a(\gamma)$ and $r^{a'}(\gamma)$ are odd and even functions respectively.}. We thus have that
\begin{align*}
    \begin{pmatrix}
        U^a_{\dot{\alpha},+}\\U^a_{\dot{\alpha},-}
    \end{pmatrix} = \frac{1}{\det M^a}\begin{pmatrix}
             r^a(-p+\hat{\gamma}_{\dot{\alpha}}) & -r^a(p+\hat{\gamma}_{\dot{\alpha}})\\
        -r^a(-p-\hat{\gamma}_{\dot{\alpha}}) &   r^a(p-\hat{\gamma}_{\dot{\alpha}})
        \end{pmatrix}\begin{pmatrix}
        j^a_{\dot{\alpha},+}\\j^a_{\dot{\alpha},-}
    \end{pmatrix}.
\end{align*}

Writing the Lax connection in terms of the currents we find that
\begin{align*}
    \mathcal{L}_{\dot{\alpha}} = \sum_{b\in\{\pm1\}}\sum_{a\in\{0,1,2\}}\frac{bj^a_{\dot{\alpha},\text{sgn}(b)}}{\det M^a} r^a(\gamma+b\hat{\gamma}_{\dot{\alpha}})(r^a(-bp+\hat{\gamma}_{\dot{\alpha}})-r^a(-bp-\hat{\gamma}_{\dot{\alpha}}))T_a.
\end{align*}
We note the forms of
\begin{align*}
    \mathcal{L}_{\dot{\alpha}}^a(p) & = \frac{r^a(p+\hat{\gamma}_{\dot{\alpha}})+r^a(-p+\hat{\gamma}_{\dot{\alpha}})}{\det M^a}(j_{\dot{\alpha},+}^ar^a(p+\hat{\gamma}_{\dot{\alpha}})-j_{\dot{\alpha},-}^ar^a(p-\hat{\gamma}_{\dot{\alpha}}))\\
    \mathcal{L}_{\dot{\alpha}}^a(-p) & =\frac{r^a(p+\hat{\gamma}_{\dot{\alpha}})+r^a(-p+\hat{\gamma}_{\dot{\alpha}})}{\det M^a}(j_{\dot{\alpha},+}^ar^a(-p+\hat{\gamma}_{\dot{\alpha}})+j_{\dot{\alpha},-}^ar^a(p+\hat{\gamma}_{\dot{\alpha}}))
\end{align*}
which will be useful later. Using this, the final step is to compute $\Res_{\gamma = \pm p}(\omega\wedge\mathcal{L}_{\dot{\alpha}})$. Computing directly leaves
\begin{align*}
    \Res_{\gamma = p}(\omega\wedge\mathcal{L}_{\dot{\alpha}}) & = \mathcal{L}_{\dot{\alpha}}(p)\Res_{\gamma = p}(\omega)+\frac{1}{2}\lim_{\gamma\to p}(\gamma-p)^2\frac{(c_1+\wp(\gamma)c_2)(c_3+\wp(\gamma)c_4)}{(\wp(\gamma)-\wp(p))^2}\mathcal{L}_{\dot{\alpha}}'(\gamma)\\
    & = \mathcal{L}_{\dot{\alpha}}(p)\Res_{\gamma = p}(\omega)+\frac{(c_1+\wp(p)c_2)(c_3+\wp(p)c_4)}{2(\wp'(p))^2}\mathcal{L}_{\dot{\alpha}}'(p).
\end{align*}
Similarly,
\begin{align*}
    \Res_{\gamma = -p}(\omega\wedge\mathcal{L}_{\dot{\alpha}}) & =-\mathcal{L}_{\dot{\alpha}}(-p)\Res_{\gamma = p}(\omega)+\frac{(c_1+\wp(p)c_2)(c_3+\wp(p)c_4)}{2(\wp'(p))^2}\mathcal{L}_{\dot{\alpha}}'(-p).
\end{align*}

So, we find that
\begin{align*}
    S_{2d}[\sigma_{\pm}] = & \frac{1}{2}\sum_{s=\pm}\left(\int_{\mathbb{R}^2}dxdy\sum_{i = 0,1}s\Res_{\gamma = sp}(\omega)^{i+1(2)}F(p)^{i(2)}\sum_{\dot{\alpha}\ne\dot{\beta}}\sum_{\tilde{s}=\pm}\tilde{s}G_{\dot{\alpha}}^a(p)\right.\\
    & \left.\cdot \frac{d^{i(2)}}{d\gamma^{i(2)}}r^a(\gamma+\tilde{s}\hat{\gamma}_{\dot{\alpha}})|_{\gamma=sp}\langle T_a,T_b\rangle j_{\dot{\alpha},\text{sgn}(\tilde{s})}^aj_{\dot{\beta},\text{sgn}(s)}^b-\Res_{\gamma=sp}(\omega)\mathcal{I}_{WZ}[\sigma_{\text{sgn}(s)}]\right).
\end{align*}
where
\begin{align*}
    F(p) = \frac{(c_1+\wp(p)c_2)(c_3+\wp(p)c_4)}{2(\wp'(p))^2}\quad G_{\dot{\alpha}}^a(p) = \frac{r^a(p+\hat{\gamma}_{\dot{\alpha}})+r^a(-p+\hat{\gamma}_{\dot{\alpha}})}{\det M^a_{\dot{\alpha}}}
\end{align*}
and $i(2)$ denotes $i$ modulo $2$.

Note that we have used the following properties: $r^a(\gamma)$ is odd, $r^{a'}(\gamma)$ is even, and $\langle T_a,T_b\rangle=\delta_{ab}$.

\end{document}